\documentclass{article}
\usepackage{authblk}
\usepackage[export]{adjustbox}
\usepackage[nosumlimits]{amsmath}
\usepackage{amssymb,amsthm,MnSymbol}
\usepackage{color}
\usepackage{algorithm}
\usepackage{algorithmic}
\usepackage{graphicx}
\usepackage{float}

\usepackage{subfig}
\def\MM#1{\boldsymbol{#1}}

\newcommand{\dd}[2]{\frac{\diff#1}{\diff#2}}

\def\MM#1{\boldsymbol{#1}}
\DeclareMathOperator{\diff}{d}

\usepackage{amscd}

\usepackage{mathrsfs}
\usepackage[colorlinks]{hyperref}
\usepackage[square,authoryear]{natbib}

\newcommand{\ShowFigures}[2]{\def\tempvara{#1}\ifnum\tempvara=1 #2\fi}

\newcommand{\vecx}[1]{\MM{#1}}

\newtheorem{theorem}{Theorem}
\newtheorem{assumption}{A}
\newtheorem{definition}[theorem]{Definition}
\newtheorem{lemma}[theorem]{Lemma}
 
\usepackage[left=2cm,right=2cm,top=2cm,bottom=2cm]{geometry}
\newtheorem{remark}{Remark}

\numberwithin{equation}{section}

\DeclareMathOperator{\curl}{curl}
\DeclareMathOperator{\vecu}{\bf u}

\DeclareMathOperator{\gradperp}{\nabla^{\perp}}

\newcommand{\lyxdot}{.}

\def\bx{{\mathbf {x} }}
\def\by{{\mathbf {y} }}
\def\bu{{\mathbf {u} }}

\def\bsxi{{\boldsymbol {\xi} }}

\def\dd{{\color{red}\mathsf d}}

\usepackage{fancybox}

\usepackage{lineno}
\usepackage[color=yellow]{todonotes}
\usepackage{enumitem}

\begin{document}


\title{Numerically Modelling Stochastic Lie Transport in Fluid 
Dynamics\footnote{This work was partially supported by the EPSRC Standard Grant EP/N023781/1.}}
\author[1]{Colin Cotter\thanks{colin.cotter@imperial.ac.uk}}
\author[1]{Dan Crisan\thanks{d.crisan@imperial.ac.uk}}
\author[1]{Darryl D. Holm\thanks{d.holm@imperial.ac.uk}}
\author[1]{Wei Pan\thanks{wei.pan@imperial.ac.uk}}
\author[1]{Igor Shevchenko\thanks{i.shevchenko@imperial.ac.uk}}
\affil[1]{Department of Mathematics, Imperial College London}

\date{\today}

\maketitle

\begin{abstract}
	We present a numerical investigation of stochastic transport in ideal fluids. According to Holm (Proc Roy Soc, 2015) and Cotter et al. (2017), the principles of transformation theory and multi-time homogenisation, respectively, imply a physically meaningful, data-driven approach for decomposing the fluid transport velocity into its drift and stochastic parts, for a certain class of fluid flows. In the current paper, we develop new methodology to implement this velocity decomposition and then numerically integrate the resulting stochastic partial differential equation using a finite element discretisation for incompressible 2D Euler fluid flows. The new methodology tested here is found to be suitable for coarse graining in this case. Specifically, we perform uncertainty quantification tests of the velocity decomposition of Cotter et al. (2017), by comparing ensembles of coarse-grid realisations of solutions of the resulting stochastic partial differential equation with the ``true solutions'' of the deterministic fluid partial differential equation, computed on a refined grid. The time discretisation used for approximating the solution of the stochastic partial differential equation is shown to be consistent.  We include comprehensive numerical tests that confirm the non-Gaussianity of the stream function, velocity and vorticity fields in the case of incompressible 2D Euler fluid flows.
	
\end{abstract}

\tableofcontents

\section{Introduction}

A fundamental challenge in observational sciences, such as weather forecasting and climate change prediction, is the modelling of measurement error and uncertainty due, for example, to unknown or neglected physical effects, and incomplete information in both the data and the formulations of the theoretical models for prediction. To meet this challenge, new types of dynamical parameterisations, called \textit{Data Driven Models} have been developing recently for the observational sciences. Data-driven models accommodate uncertainty in observational science, by making predictions of both the values of the expected future measurements and of their uncertainties, or variabilities, based on input from measurements and statistical analysis of the initial data. 

Such predictions are made in a probabilistic sense. They may also use \emph{data assimilation} to take into account the time integrated information obtained from the data being observed along the solution path during the forecast interval as ``in flight corrections''. Data assimilation is a term used mainly in the computational geoscience community, and refers to methodologies that combine past knowledge of a system in the
form of a numerical model with new information about that system in the
form of observations of that system.
It is a major component of Numerical Weather Prediction, where  it is used to improve forecasting, reduce model uncertainties and adjust
model parameters. 
To reduce the uncertainty, a stochastic feedback loop between the model and the data may be introduced, through which assimilation of more data during the prediction interval will decrease the uncertainty of the forecasts based on the initial data, by selecting the likely paths as more observational data is accrued. This is the basis of the so-called ensemble data assimilation which uses a set of model trajectories that are intermittently
updated according to data. The availability for several years of large
grid computing systems has made ensemble data assimilation increasingly
popular.
Ensemble data assimilation can use  \textit{particle filters}\footnote{See \cite{Crisan2014} for a new approach for handling high dimensional models using particle filters.} as a basis for  the uncertainty reduction.  

Thus, in modern observational science, predictive dynamics meets: (i) observation error, (ii) incomplete information, (iii) uncertainty, (iv) resolution errors in numerical simulations and (v) data assimilation. In the new science of data-driven modelling, all four of these endeavours should be placed into the same framework. As a minimum requirement, the framework for the introduction of noise should preserve the fundamental mathematical structure of the deterministic model. The geometric mechanics approach that we take here is designed to preserve the fundamental structure of fluid dynamics, which is based on the theory of transformations by smooth invertible maps. 

Our approach introduces a new type of stochasticity -- called Stochastic Lie Transport (SLT) -- that has been designed specifically for fluid dynamics, based on transformation theory and geometric mechanics. Instead of trying to predict the effects of what cannot be resolved in each case by going to even higher resolution, SLT uses observed spatial correlation data to model the effects of the uncertainty as spatially correlated stochastic transport.

\subsection*{Properties of Stochastic Lie Transport in Ideal Fluid Dynamics} 

\textit{Stochastic Lie Transport} (SLT) for  ideal fluid dynamics was first derived in \cite{Holm2015} 
by applying transformation theory from geometric mechanics (based on the smooth invertible Lagrange-to-Euler map) to the Hamilton variational principle for the equations of ideal fluid motion.  SLT also follows from \textit{Newton's Law of Motion}, provided one includes the stochastic Lagrange-to-Euler transformation of the reference coordinate basis under the fluid flow, as shown in \cite{CrFlHo2017}. In addition, homogenisation theory shows that SLT can be regarded as a true decomposition of the deterministic solution for the fluid velocity into a mean flow and rapid fluctuations in velocity around the mean \cite{CoGaHo2017}.  Via homogenisation theory, the rapid velocity fluctuations rigorously transform into a sum of stochastic vector fields, as in equation \eqref{StochVF}, in the limit as the fluctuation frequency increases.

As a true decomposition of the solution, the \textit{analytical properties} of the SLT fluid model should not differ from those of the corresponding deterministic fluid equations. This property was proven to hold for the 3D SLT Euler fluid equations in \cite{CrFlHo2017}. In particular, the solutions of the 3D SLT Euler fluid equations \eqref{Euler-vort-eqn-2form-stoch} derived in \cite{Holm2015} are shown in \cite{CrFlHo2017} to possess local-in-time existence and uniqueness, as well as to satisfy a Beale-Kato-Majda criterion for blow-up, corresponding to the same properties as for the deterministic 3D Euler fluid equations \cite{BKM1984}. 

In summary, SLT is a new type of stochasticity, designed to account for the effects of unresolved fluid degrees of freedom, such as turbulence, on the resolved scale dynamics of fluid flows. In SLT,  the noise multiplies both the solution and the \textit{spatial gradient} of the solution; so its influence tends to increase as the gradients of the solution increase. The additional Lie transport terms in SLT turn out to be necessary to complete the \textit{Stochastic Kelvin Circulation Theorem}. Thus, in the SLT approach, the Eulerian fluid equations acquire the additional Stratonovich stochastic transport vector field seen in equation \eqref{StochVF}, which leads to the Kelvin Circulation Theorem in \eqref{StochKelThm}. These additional stochastic Lie transport terms do not appear in other theories, such as \cite{MiRo2004} and \cite{Me2014}. The present paper will demonstrate how to use SLT as a means of performing both \textit{uncertainty quantification} and \textit{account for the resolution error} in numerical simulations by the following steps, see Section \ref{subsec:Methodology}:
\begin{itemize}
        \item Simulate Lagrangian trajectories moving with velocity described by a deterministic PDE which we assume can be accurately approximated on a fine grid. Call these the \emph{fine grid trajectories}.
        \item Simulate Lagrangian trajectories moving with a spatially-filtered velocity on a coarse grid. Call these the \emph{coarse grid trajectories}.
        \item Calculate the differences between the fine grid trajectories and the coarse grid trajectories over non-overlapping time intervals, which are then used to estimate the velocity-velocity correlation tensors.
        \item The estimated velocity-velocity correlation tensors are substituted into the Euler stochastic PDE with Lie transport noise to perform uncertainty quantification analysis.
\end{itemize}

\subsection*{The Interaction Between Noise and Transport Mechanisms in Ideal Fluids}
Our aim in this paper is to investigate the interaction between noise and transport in ideal fluids using the framework of geometric mechanics, \cite{MaRa1994,Ho2011,HoScSt2009}.
The understanding of transport mechanisms in fluid dynamics is at the core of some of the main open problems in mathematics and physics. The introduction of random perturbations into the fluid equations can be expected to profoundly influence the properties of fluid transport and thereby raise many open questions. For a mathematical review of the literature and recent progress on the interaction between noise and transport in the vorticity equation for 2D ideal incompressible fluids, see \cite{BrFlMa2016}.

A variational approach to the full theory of stochastic ideal fluid dynamics in 3D was derived in \cite{Holm2015}, by using transformation theory from geometric mechanics based on the Lagrange-to-Euler map for stochastic Lagrangian particle trajectories. Its analytical properties have been investigated in \cite{CrFlHo2017} for the particular case of the 3D stochastic Euler equation for incompressible fluid flow, ${\rm div}\mathbf{u}=0$, given by
\begin{align}
0 = (\dd + \mathcal{L} _{\dd{\by}_{t}}) (\boldsymbol{\omega} \cdot
d\boldsymbol{S}) 
= \Big(\dd \boldsymbol{\omega} - \mathrm{curl}
\,(\dd{\by}_{t}\times\boldsymbol{\omega})\Big) \cdot d\boldsymbol{S} \,,
\label{Euler-vort-eqn-2form-stoch}%
\end{align}
for the Eulerian vorticity 2-form 
$\boldsymbol{\omega} \cdot d\boldsymbol{S}=d(\bu\cdot d\bx)
=({\rm curl}\,\bu) \cdot d\boldsymbol{S}$, which is Lie transported by the \textit{Stratonovich} stochastic vector field $\dd{\by}_{t}$ corresponding to the following stochastic process,
\begin{align}
\dd{\by}_{t} = \bu(\by_t,t)dt + \sum_{i} \bsxi_{i}(\by_t) \circ dW^{i}_{t}
\,. \label{StochVF}%
\end{align}
Here, $\dd$ represents stochastic differentiation and the second term in \eqref{StochVF} constitutes cylindrical Stratonovich noise, in which the amplitude of the noise depends on space, but not time. { In It\^o form, \eqref{StochVF} is written as
\[
\dd{\by}_{t} = \bu(\by_t,t)dt + \sum_{i} \bsxi_{i}(\by_t) dW^{i}_{t} +  \frac{1}{2}\sum_{i} (\bsxi_i(\by_t) \cdot \nabla) \bsxi_i(\by_t) dt
\,.%
\] For an extension of this method to include non-stationary correlation statistics, see \cite{GBHo2018}.}

In the case of 2D planar incompressible fluid motion, the vorticity has only one component, denoted as $q$, and equation \eqref{Euler-vort-eqn-2form-stoch} reduces to
\begin{align}
0 = (\dd + \mathcal{L} _{\dd{\by}_t}) q 
= \dd \omega + \dd\by_t\cdot\nabla q
\,.
\label{Euler-vort-eqn-2D-stoch}%
\end{align}

For in-depth treatments of cylindrical noise, see \cite{Pa2007,Sc1988}. In our case, the vectors  $\bsxi_{i}(\bx)$, $i=1,2,\dots,N$, appearing in the stochastic vector field in \eqref{StochVF} comprise $N$ prescribed, time independent, divergence free vectors which are to be obtained from data. That is, we incorporate SLT into fluid dynamics as a \textit{Data--Driven Model}. For example, the $\bsxi_{i}(\bx)$ may be determined as Empirical Orthogonal Functions (EOFs), which are eigenvectors of the velocity-velocity correlation tensor for a certain 
measured flow with stationary statistics, see \cite{HaJoSt2007, Hannachi2004primer}. As discussed below, the  $\bsxi_{i}(\bx)$ in equation \eqref{StochVF} may also be obtained numerically by comparisons of Lagrangian trajectory simulations at fine and coarse space and time scales. 

The introduction of cylindrical Stratonovich noise into Euler's fluid equation by using its variational and Hamiltonian structure has introduced an additional, stochastic vector field $\sum_{i} \bsxi_{i}(\bx) \circ dW^{i}_{t}$ into equation \eqref{StochVF} which augments the Lie transport in equation \eqref{Euler-vort-eqn-2form-stoch}. This is natural, because the essence of Euler fluid dynamics is Lie transport, see \cite{HoMaRa1998}. In particular, equation \eqref{Euler-vort-eqn-2form-stoch} produces a natural Kelvin Circulation Theorem of the form
\begin{align}
\dd \oint_{c(\dd\by_t)}\hspace{-4mm} \bu\cdot d\bx 
=  \oint_{c(\dd\by_t)}\hspace{-4mm}(\dd + \mathcal{L} _{\dd{\by}_{t}}) ( \bu\cdot d\bx)
=  \int\hspace{-2mm}\int_{\partial S = c(\dd\by_t)}\hspace{-4mm}(\dd + \mathcal{L} _{\dd{\by}_{t}}) (\boldsymbol{\omega} \cdot
d\boldsymbol{S})
= 0
\,, \label{StochKelThm}%
\end{align}
where $c(\dd\by_t)$ is a closed fluid loop moving with the stochastic vector field velocity $\dd\by_t$.\medskip

Other 3D stochastic Euler fluid equations have been derived by different methods in \cite{MiRo2004} and \cite{Me2014}. However, those other derivations have produced equations which differ from \eqref{Euler-vort-eqn-2form-stoch} in their stochastic transport terms and consequently  do not admit the Kelvin Circulation Theorem in \eqref{StochKelThm}.

\begin{remark}[Kelvin's Circulation Theorem]\rm
From the viewpoint of geometric mechanics, Kelvin's Circulation Theorem \eqref{StochKelThm} is always a fundamental property in fluid dynamics. It is the relation obtained via Noether's Theorem from invariance of Eulerian fluid variables under smooth invertible transformations of the Lagrangian particle labels. This invariance is called \textit{relabelling symmetry}. If one starts with the Lagrange-to-Euler map, this invariance yields a momentum map which satisfies Kelvin's Circulation Theorem. Even when the Lagrange-to-Euler map is stochastic, as in equation \eqref{StochVF}, the relabelling symmetry is still maintained, and this invariance implies a stochastic Kelvin's Circulation Theorem. In the stochastic case, the closed circulation loop around which one integrates in Kelvin's theorem follows the flow lines of the stochastic Lagrangian paths, and it remains a loop because the stochastic Lagrange-to-Euler map is still a diffeomorphism. Thus, Kelvin's Circulation Theorem has the same geometric transport interpretation in both the deterministic and stochastic cases. 
This preservation of the Kelvin Circulation Theorem interpretation for stochastic fluid dynamics is unique to the present approach. The same interpretation also applies to the equivalent Hamiltonian formulation of these equations. 
\end{remark}

\subsection*{The main content of the paper}
The rest of the paper is structured as follows. Section \ref{sec:2dequations} describes the damped and forced deterministic system and the numerical methodology we use to solve the system on a fine resolution spatial grid. This corresponds to the simulated \textit{truth}. We then describe the stochastic version of this system, derived by using
the variational approach formulated in \cite{Holm2015}. The numerical methodology we use for solving the deterministic system is extended to solve the stochastic version and a proof for the numerical consistency of the method is provided. 
        
Section \ref{sec:calibration} describes our numerical calibration methodology for the stochastic model. Here, numerical simulations and tests are provided to show that using our methodology, one can \emph{sufficiently }estimate the velocity-velocity spatial correlation structure from \emph{data}, so that an ensemble of flow paths described by the stochastic system accurately tracks the large-scale behaviour of the underlying deterministic system for a physically adequate period of time.
        
Finally, Section \ref{sec:conlusion} concludes the present work and discusses the outlook for future research.

The following is a list of the numerical experiments contained in this paper.
\begin{itemize}
        \item Simulation of the deterministic Euler equations on a fine grid of size $512\times512$ for $146$ large eddy turnover times (abbrev. ett). Here, we determine a suitable initial condition which is spun-up from a chosen initial configuration. See Section \ref{subsec:pde_results} for a detailed description of the initial configuration and spin-up. See Figures \ref{fig:omega_spin}--\ref{fig:energy_series} for visualisations of the results.
        
        \item The fine grid PDE simulations are coarse grained to obtain the coarse-grained PDE solutions define on a coarse grid of size $64\times64$, which we call the \emph{truth}. See Section \ref{subsec:pde_results} for a description of the coarse-graining procedure, and Figures \ref{fig:pde_solution_t0}--\ref{fig:pde_solution_t146} for visualisations of the results.
        
        \item We compute the fine grid and coarse grid Lagrangian trajectories, which are used to estimate the velocity-velocity correlation tensor, interpreted as spatial correlation EOFs, see \cite{HaJoSt2007, Hannachi2004primer}. See Section \ref{subsec:Methodology} for a description of the parameterisation methodology. Figure \ref{fig:EOF-normalised-spectrum-64} shows a plot of the normalised spectrum corresponding to the estimated EOFs. The experiment is repeated for more refined coarse grids so we can investigate whether our { parameterisation} methodology is consistent with grid refinement. The normalised spectrums for the refined coarse grids of sizes $128\times128$ and $256\times256$ are shown in Figure \ref{fig:EOF-normalised-spectrum-128} and Figure \ref{fig:EOF-normalised-spectrum-256} respectively.
        
        \item The estimated EOFs are substituted into the Lagrangian trajectory equation. An ensemble of independent realisations of the stochastic Lagrangian trajectory equation are computed to do Lagrangian trajectory uncertainty quantification. Figure \ref{fig:Lagrangian-trajectory-uq} shows the results of one such tests.
        
        \item The stochastic PDE (SPDE) is defined on the coarse grid. Given the PDE solution at a fixed time $t_0$, we obtain an ensemble of initial conditions for the SPDE via the deformation procedure \eqref{eq:initial_cond_deformation}, so that the truth lies in the concentration of the probability density of the initial prior distribution. We call each realisation of the SPDE a \emph{particle}. See Section \ref{subsec:spde_initial_results} for a detailed description of how we generate the SPDE initial conditions. Figures \ref{fig:spde_solution}--\ref{fig:spde_solution-2} show plots of the truth and two particles at the initial time $t_0$, $t_0+3$ ett and $t_0+5$ ett.
        
        \item Using the estimate EOFs, we perform uncertainty quantification tests for the SPDE where the truth is compared with the ensemble one standard deviation region about the ensemble mean at the interior grid points of a $4\times4$ observation grid. We would like the ensemble spread to capture the truth for an adequate period of time starting from an initial ensemble that captures the initial truth. See Section \ref{subsec:spde_uq_results} for a description of the tests and see Figures \ref{fig:mike-cullen-psi}--\ref{fig:mike-cullen-u} for the results. The results show that our parameterisation methodology described in Section \ref{subsec:Methodology} works well at the $64\times64$ resolution level (eight times coarser than the fine grid), as the spread captures the truth for at least $5$ eddy turnover times before significant deviations occur.
        
        \item The SPDE uncertainty quantification tests are repeated for more refined coarse grids of sizes $128\times128$ and $256\times256$. See Section \ref{subsec:spde_uq_results} for a detailed description. See Figures \ref{fig:mike-cullen-psi-multires}--\ref{fig:mike-cullen-u-multires} for the test results where we plot the truths and the ensemble spreads together in a single figure to compare the differences. The results show that as the coarse grid gets refined, the ensemble spread captures the truth for longer time periods, confirming that the parameterisation methodology is consistent with grid refinement. 
        
        \item We also investigate the relative minimum $L^2$ distance between the SPDE ensemble and the truth defined by \eqref{eq:min_l2_distance}. See Section \ref{subsec:spde_uq_results} for a detailed explanation. Figures \ref{fig:min_l2_psi}--\ref{fig:min_l2_u} shows the results. The distance between the SPDE ensemble and the truth diverges as time goes on, indicating that the uncertainty whether the ensemble captures the truth increases with time.
        
        \item The Lie transport noise is not additive thus the SPDE ensemble should not be of Gaussian distribution. We test for this using Quantile-Quantile (QQ) plots and boxplots. The results are shown in Figures \ref{fig:qqplot-psi}--\ref{fig:qqplot-u} for the QQ tests, and Figures \ref{fig:boxplot-psi}--\ref{fig:boxplot-u} for the boxplot tests. Fat tails and non-symmetry in the distribution gives strong evidence to the fact that the ensemble is not Gaussian. See Section \ref{subsec:additional_stats_results} for a more detailed explanation.
\end{itemize}


\section{The two-dimensional incompressible flow equations \label{sec:2dequations}} %


\subsection{Deterministic version}

Let the state space $\mathcal{D}$ be the unit square in $\mathbb{R}^{2}.$
We consider a two dimensional incompressible fluid flow velocity, ${\bf u}$,
defined on $\mathcal{D}$, $\vecu:\mathcal{D}\times\left[0,\infty\right)\rightarrow\mathbb{R}^{2},$
$\vecu\left(x,y,t\right)=\left(u_{1}\left(x,y,t\right),u_{2}\left(x,y,t\right)\right)$,
whose dynamics is governed by the two-dimensional Euler equations with additional
forcing and damping. In what follows, we shall work with the vorticity
version of the Euler equation. 

Let $\omega=\hat{z}\cdot \curl{\bf u}$ denote the vorticity of $\vecu$,
where $\hat{z}$ denotes the $z$-axis. Note that for incompressible flow in two dimensions,
$\omega$ is formally a scalar field. For a scalar field $g:\mathcal{D}\rightarrow\mathbb{R},$
we write $\nabla^{\perp}g=\left(-\partial_{y}g,\partial_{x}g\right)=\hat{z}\times\nabla g.$
We also let $\psi:\mathcal{D}\times\left[0,\infty\right)\rightarrow\mathbb{R}$
denote the \emph{stream function}, another scalar field, related to the fluid velocity and vorticity by ${\bf u}=\nabla^{\perp}\psi$ and $\omega=\Delta\psi$, respectively, where $\Delta = \partial_x^2+\partial_y^2$ is the Laplacian operator in $\mathbb{R}^2$. Note that the existence of the
stream function is guaranteed by the incompressibility assumption.

We can now write down the model equations, as
\begin{eqnarray}
\partial_{t}\omega+\left({\bf u}\cdot\nabla\right)\omega & = & Q-r\omega\label{eq:2DEulerVorticity}\\
{\bf u} & = & \nabla^{\perp}\psi\label{eq:u_is_grad_perp_stream}\\
\Delta\psi & = & \omega\label{eq:laplace_stream_is_vorticity}
\end{eqnarray}
where we have chosen the forcing $Q$ to be
\begin{equation}
Q\left(x,y\right)=\alpha\sin\left(\beta\pi x\right),\qquad (x,y)\in \mathcal{D}\label{eq:forcing}
\end{equation}
and $\alpha,$ $\beta$ and $r$ are constants which have the following roles: $\alpha$ controls
the strength of the forcing; $\beta$ can be interpreted as the number
of gyres in the external forcing; and $r>0$ can be seen as the damping
rate.

We shall consider a slip flow boundary condition 
\begin{equation}
\left.\psi\right|_{\partial\mathcal{D}}=0.\label{eq:boundary_condition_streamfunction}
\end{equation}

This system is a special case of a nonlinear, one-layer quasigeostrophic (QG) model
that is driven by winds above.

\begin{remark}
    The term $\vecu\cdot\nabla\omega$ in equation (\ref{eq:2DEulerVorticity})
    can be expressed as the Jacobian $J\left(\psi,\omega\right)$ for the transformation 
    $d\psi\wedge d\omega=J\left(\psi,\omega\right)dx\wedge dy$, i.e.,
    \[
    \begin{aligned}\left({\bf u}\cdot\nabla\right)\omega & =u\partial_{x}\omega+v\partial_{y}\omega\\
    & =\partial_{y}\psi\partial_{x}\omega-\partial_{x}\psi\partial_{y}\omega\\
    & ={\rm det}\left[\begin{array}{cc}
    \partial_{x}\omega & \partial_{y}\omega\\
    \partial_{x}\psi & \partial_{y}\psi
    \end{array}\right]=:J\left(\psi,\omega\right).
    \end{aligned}
    \]
\end{remark}

\begin{remark}
        \label{rem:pde_exit_uniq_sol} To our knowledge, the mathematical analysis of the solution properties for the damped and forced deterministic system in equation \eqref{eq:2DEulerVorticity} has not yet appeared in the literature. For recent work on the stochastic version, see \cite{CrOLa2018}. 
\end{remark}


\subsubsection{Numerical implementation}

\label{sec:numerical-implementation-(firedrake)-wei}
We solve the system of equations (\ref{eq:2DEulerVorticity}), (\ref{eq:u_is_grad_perp_stream})
and (\ref{eq:laplace_stream_is_vorticity}) using finite element discretisation.
Without the source terms in (\ref{eq:2DEulerVorticity}), energy and
enstrophy are conserved quantities for the same choice of boundary
condition. Thus we follow \cite{bernsen2006dis,gottlieb2005high}
and use a combination of a mixed continuous and discontinuous Galerkin
finite element discretisation scheme and an optimal third order strong
stability preserving Runge-Kutta method for the time stepping, that
conserves numerical energy and enstrophy. We give a description of the numerical procedure. 

\subsubsection*{Streamfunction equation}

Let $H^{1}\left(\Omega\right)$ denote the Sobolev $W^{1,2}\left(\Omega\right)$
space and let $\left\Vert .\right\Vert _{\partial\Omega}$ denote
the $L^{2}\left(\partial\Omega\right)$ norm. For the elliptic equation
(\ref{eq:laplace_stream_is_vorticity}) we obtain its variational
formulation by multiplying both sides by a test function $\phi$ in
$W^{1}\left(\Omega\right):=\left\{ \nu\in H^{1}\left(\Omega\right)\left|\left\Vert \nu\right\Vert _{\partial\Omega}=0\right.\right\} $
then integrating over the domain $\Omega.$ Using integration by parts
we get

\begin{equation}
\begin{aligned}\left\langle \phi,\omega\right\rangle _{\Omega} & =\left\langle \phi,\Delta\psi\right\rangle _{\Omega}\\
& =\int_{\partial\Omega}\phi\nabla\psi\cdot\hat{n}d{\bf r}-\left\langle \nabla\phi,\nabla\psi\right\rangle _{\Omega}\\
& =-\left\langle \nabla\phi,\nabla\psi\right\rangle _{\Omega}
\end{aligned}
\label{eq:weak_formulation_elliptic}
\end{equation}
where the integral over $\partial\Omega$ is zero due to the boundary
condition (\ref{eq:laplace_stream_is_vorticity}). 

Equation (\ref{eq:weak_formulation_elliptic}) is discretised by using a continuous
Galerkin (CG) discretisation scheme. This simplifies the choice of
the discontinuous Galerkin numerical flux for the hyperbolic equation
\eqref{eq:2DEulerVorticity}, see \cite{bernsen2006dis}. This means choosing
approximations of $\psi$ and $\omega$ in the subspace
\[
\mathcal{W}_{h}^{k}=\left\{ \left.\phi_{h}\in W^{1}\left(\Omega\right)\right|\phi_{h}\in C\left(\Omega\right),\left.\phi_{h}\right|_{K}\in\mathcal{P}^{k}\left(K\right)\right\} 
\]
where $\mathcal{P}^{k}\left(K\right)$ is the space of continuous
polynomials of degree at most $k$ on each element $K$ of a triangulation
$\mathcal{T}_{h}$ of the space $\Omega.$ Thus the numerical approximation
of $\psi$ is given by $\psi_{h}\in\mathcal{W}_{h}^{k}$, such that
\begin{equation}
\left\langle \phi_{h},\omega_{h}\right\rangle _{\Omega}=-\left\langle \nabla\phi_{h},\nabla\psi_{h}\right\rangle _{\Omega}\label{eq:discretised_weak_formulation_elliptic}
\end{equation}
for all $\phi_{h}\in\mathcal{W}_{h}^{k}.$ 

\subsubsection*{Vorticity equation}

For the hyperbolic equation (\ref{eq:2DEulerVorticity}), a discontinuous
Galerkin (DG) scheme is used. This leads to the following variational
problem
\begin{equation}
\begin{aligned}\left\langle \nu_{h},\partial_{t}\omega\right\rangle _{K} & =\left\langle \nu_{h},Q-r\omega\right\rangle _{K}+\left\langle \omega{\bf u},\nabla\nu_{h}\right\rangle _{K}-\left\langle \omega{\bf u}\cdot\hat{n},\nu_{h}\right\rangle _{\partial K}\end{aligned}
\label{eq:weak_formulation_hyperbolic}
\end{equation}
for any test function {$\nu_{h}$} in the space of discontinuous test functions
$\mathcal{V}_{h}^{k}=\left\{ \left.v\right|\forall K\in\mathcal{T}_{h},\ \exists\phi_{h}\in\mathcal{P}^{k}\left(K\right):\ \left.v\right|_{K}=\phi_{h}\right\} .$
This choice of $\mathcal{V}_{h}^{k}$ ensures conservation of
energy for the numerical solution of \eqref{eq:2DEulerVorticity} minus
source terms, see \cite{bernsen2006dis}. 

In this DG setup, $\omega$ and {$\nu_{h}$} in (\ref{eq:weak_formulation_hyperbolic})
are discontinuous across elements $K\in\mathcal{T}_{h},$ but $\vecu$
is continuous. The latter is due to the CG discretisation for the
elliptic equation for $\psi$ and the fact that $\vecu\cdot\hat{n}=\nabla^{\perp}\psi\cdot\hat{n}=-\nabla\psi\cdot\hat{\tau}=-d\psi_{h}/d\hat{\tau}$
where $\hat{\tau}$ is the tangential unit vector to $\partial K.$
Thus for the integral along the boundary $\partial K,$ we need to
specify a unique numerical flux for each cell interface. 

Let {$\nu_{h}^{-}:=\lim_{\epsilon\uparrow0} \nu_h \left({\bf x}+\epsilon\hat{n}\right)$} and { $\nu_{h}^{+}:=\lim_{\epsilon\downarrow0}\nu_h \left({\bf x}+\epsilon\hat{n}\right)$} for ${\bf x}\in\partial K.$ Let {$\nu_{h}$} in $\left\langle \omega{\bf u}\cdot\hat{n},\nu_{h}\right\rangle _{\partial K}$ be {$\nu_{h}^{-},$} and replace $\omega\vecu\cdot\hat{n}$ by the numerical flux $\hat{f}\left(\omega_{h}^{+},\omega_{h}^{-},{\bf u}\cdot\hat{n}\right)$ given by the upwind scheme
\[
\hat{f}\left(\omega_{h}^{+},\omega_{h}^{-},{\bf u}\cdot\hat{n}\right)={\bf u}\cdot\hat{n}\begin{cases}
\omega_{h}^{+} & \text{if }{\bf u}\cdot\hat{n}<0\\
\omega_{h}^{-} & \text{if }{\bf u}\cdot\hat{n}\geq0.
\end{cases}
\]
This choice of $\hat{f}$ is consistent, conserves the numerical flux across neighbouring elements, and is $L^{2}$-stable in the enstrophy norm, see \cite{bernsen2006dis}. Note that the choice for $\hat{f}$ with these properties is not unique.

With these choices, the goal is to find $\omega_{h}\in V_{h}^{k}$ such that for all {$\nu_{h}\in V_{h}^{k}$} we have
{\begin{equation}
\begin{aligned}\left\langle \nu_{h},\partial_{t}\omega_{h}\right\rangle _{K} =\left\langle \nu_{h},Q_{h}-r\omega_{h}\right\rangle _{K} +
\left\langle \omega_{h}\nabla^{\perp}\psi_{h}, \nabla \nu_{h}\right\rangle _{K}  -
\left\langle \hat{f}\left(\omega_{h}^{+},\omega_{h}^{-},\nabla^{\perp}\psi_{h} \cdot \hat{n}\right), \nu_{h}^{-}\right\rangle _{\partial K}.
\end{aligned}
\label{eq:discretised_weak_formulation_hyperbolic}
\end{equation}
}
\subsubsection*{Time stepping}

For the time stepping scheme, we follow \cite{gottlieb2005high} and
use a strong stability preserving Runge Kutta method of order 3 (SSPRK3)
with the  Courant--Friedrich--Lewy (CFL) condition being $1/3.$

Writing the finite element spatial discretisation formally as $\partial_{t}\omega=f_{h}\left(\omega\right)$
where $f_{h}$ is the discretisation operator that follows from (\ref{eq:discretised_weak_formulation_hyperbolic})
and (\ref{eq:discretised_weak_formulation_elliptic}), the SSPRK3
time discretisation is as follows 
\begin{align*}
\omega^{(1)} & =\omega^{n}+\Delta f_{h}\left(\omega^{n}\right)\\
\omega^{\left(2\right)} & =\frac{3}{4}\omega^{n}+\frac{1}{4}\left(\omega^{\left(1\right)}+\Delta f_{h}\left(\omega^{\left(1\right)}\right)\right)\\
\omega^{n+1} & =\frac{1}{3}\omega^{n}+\frac{2}{3}\left(\omega^{\left(2\right)}+\Delta f_{h}\left(\omega^{\left(2\right)}\right)\right)
\end{align*}
where $\Delta=t_{n+1}-t_{n}$ each $n$. 

In variational form, we have

\begin{equation}
\begin{aligned}\left\langle v_{h},\omega^{(1)}\right\rangle _{K} & =\left\langle v_{h},\omega^{n}\right\rangle _{K}-\Delta t\left(\left\langle \nabla v_{h},-\omega^{n}\nabla^{\perp}\psi_{h}^{n}\right\rangle _{K}-\left\langle v_{h},Q-r\omega^{n}\right\rangle _{K}+\left\langle v_{h},\omega^{n}\nabla^{\perp}\psi_{h}^{n}\cdot\hat{n}\right\rangle _{\partial K}\right)\\
\left\langle v_{h},\omega^{\left(2\right)}\right\rangle _{K} & =\frac{3}{4}\left\langle v_{h},\omega^{n}\right\rangle _{K}+\frac{1}{4}\left\langle v_{h},\omega^{\left(1\right)}\right\rangle _{K}\\
& \qquad-\frac{\Delta t}{4}\left(\left\langle \nabla v_{h},-\omega^{\left(1\right)}\nabla^{\perp}\psi_{h}^{\left(1\right)}\right\rangle _{K}-\left\langle v_{h},Q-r\omega^{\left(1\right)}\right\rangle _{K}+\left\langle v_{h},\omega^{\left(1\right)}\nabla^{\perp}\psi_{h}^{\left(1\right)}\cdot\hat{n}\right\rangle _{\partial K}\right)\\
\left\langle v_{h},\omega^{n+1}\right\rangle _{K} & =\frac{1}{3}\left\langle v_{h},\omega^{n}\right\rangle _{K}+\frac{2}{3}\left\langle v_{h},\omega^{\left(2\right)}\right\rangle _{K}\\
& \qquad-\frac{2\Delta t}{3}\left(\left\langle -\omega^{\left(2\right)}\nabla^{\perp}\psi_{h}^{\left(2\right)},\nabla v_{h}\right\rangle _{K}-\left\langle v_{h},Q-r\omega^{\left(2\right)}\right\rangle _{K}+\left\langle v_{h},\omega^{\left(2\right)}\nabla^{\perp}\psi_{h}^{\left(2\right)}\cdot\hat{n}\right\rangle _{\partial K}\right)
\end{aligned}
\label{eq:ssprk3}
\end{equation}
for each $K\in\mathcal{T}_{h}.$


We summarise our numerical procedure as Algorithm \ref{alg:det_sys_alg}.
Our implementation of (\ref{eq:ssprk3}), (\ref{eq:discretised_weak_formulation_elliptic})
and (\ref{eq:discretised_weak_formulation_hyperbolic}) is done using
Firedrake\footnote{http://www.firedrakeproject.org/index.html}, which
is an efficient automated finite element method library that employs
the Unified Form Language (UFL), \cite{Rathgeber2016,Dalcin2011,petsc-efficient,petsc-user-ref}.

For the schemes we use, the spatial and time discretisations need
to be chosen so that the CFL condition 
\[
c\leq C_{\text{effective}}=\frac{1}{3}
\]
is satisfied in order to have numerical stability, c.f. \cite{gottlieb2005high}.

\begin{algorithm}
        \caption{\label{alg:det_sys_alg}Solver algorithm for the deterministic system
                (\ref{eq:2DEulerVorticity}) -(\ref{eq:laplace_stream_is_vorticity})}
        
        \begin{algorithmic}[1]
                
                \STATE Let $\Delta t$ and $\Delta x$ be the time discretisation
                step and the spatial discretisation step respectively, such that they
                satisfy the CFL condition of $1/3.$ Let $\omega_{0}$ be a given
                initial vorticity at $t=0.$
                
                \FOR{$t_{i}=i\Delta t$, $i=0,1,2\dots,N-1$, with $t_{N}=T$}
                
                \STATE Set $\omega_{h}^{n}=\omega_{i}.$
                
                \STATE Solve 
                \[
                \left\langle \phi_{h},\omega_{h}^{n}\right\rangle _{\Omega}=-\left\langle \nabla\phi_{h},\nabla\psi_{h}^{n}\right\rangle _{\Omega}
                \]
                to obtain $\psi_{h}^{n}$ which we then use to solve 
                \[
                \left\langle v_{h},\omega_{h}^{(1)}\right\rangle _{K}=\left\langle v_{h},\omega_{h}^{n}\right\rangle _{K}-\Delta t\left(\left\langle \nabla v_{h},-\omega_{h}^{n}\nabla^{\perp}\psi_{h}^{n}\right\rangle _{K}-\left\langle v_{h},Q-r\omega^{n}\right\rangle _{K}+\left\langle v_{h},\omega^{n}\nabla^{\perp}\psi_{h}^{n}\cdot\hat{n}\right\rangle _{\partial K}\right)
                \]
                $K\in\mathcal{T}_{h},$ to obtain $\omega_{h}^{\left(1\right)}.$
                
                \STATE Solve 
                \[
                \left\langle \phi_{h},\omega_{h}^{\left(1\right)}\right\rangle _{\Omega}=-\left\langle \nabla\phi_{h},\nabla\psi_{h}^{\left(1\right)}\right\rangle _{\Omega}
                \]
                to obtain $\psi_{h}^{\left(1\right)}$ which we then use to solve
                \[
                \begin{aligned}\left\langle v_{h},\omega_{h}^{\left(2\right)}\right\rangle _{K} & =\frac{3}{4}\left\langle v_{h},\omega_{h}^{n}\right\rangle _{K}+\frac{1}{4}\left\langle v_{h},\omega_{h}^{\left(1\right)}\right\rangle _{K}\\
                & \qquad-\frac{\Delta t}{4}\left(\left\langle \nabla v_{h},-\omega_{h}^{\left(1\right)}\nabla^{\perp}\psi_{h}^{\left(1\right)}\right\rangle _{K}-\left\langle v_{h},Q-r\omega_{h}^{\left(1\right)}\right\rangle _{K}+\left\langle v_{h},\omega_{h}^{\left(1\right)}\nabla^{\perp}\psi_{h}^{\left(1\right)}\cdot\hat{n}\right\rangle _{\partial K}\right)
                \end{aligned}
                \]
                $K\in\mathcal{T}_{h},$ to obtain $\omega_{h}^{\left(2\right)}.$
                
                \STATE Solve 
                \[
                \left\langle \phi_{h},\omega_{h}^{\left(2\right)}\right\rangle _{\Omega}=-\left\langle \nabla\phi_{h},\nabla\psi_{h}^{\left(2\right)}\right\rangle _{\Omega}
                \]
                to obtain $\psi_{h}^{\left(2\right)}$ which we then use to solve
                
                \[
                \begin{aligned}\left\langle v_{h},\omega_{h}^{n+1}\right\rangle _{K} & =\frac{1}{3}\left\langle v_{h},\omega_{h}^{n}\right\rangle _{K}+\frac{2}{3}\left\langle v_{h},\omega_{h}^{\left(2\right)}\right\rangle _{K}\\
                & \qquad-\frac{2\Delta t}{3}\left(\left\langle -\omega_{h}^{\left(2\right)}\nabla^{\perp}\psi_{h}^{\left(2\right)},\nabla v_{h}\right\rangle _{K}-\left\langle v_{h},Q-r\omega_{h}^{\left(2\right)}\right\rangle _{K}+\left\langle v_{h},\omega_{h}^{\left(2\right)}\nabla^{\perp}\psi_{h}^{\left(2\right)}\cdot\hat{n}\right\rangle _{\partial K}\right)
                \end{aligned}
                \]
                $K\in\mathcal{T}_{h},$ to obtain $\omega_{h}^{n+1}.$
                
                \STATE Set $\omega_{i+1}=\omega_{h}^{n+1}.$
                
                \ENDFOR
                
        \end{algorithmic}
\end{algorithm}


\subsection{Stochastic version}

Let $\left(\Omega,\mathcal{F},\left(\mathcal{F}_{t}\right)_{t\geq0},P\right)$
be a filtered probability space. Let $W_{t}^{i}:\Omega\rightarrow\mathbb{R},$
$i=1,2,\dots,$ be a sequence of independent Brownian motions. A stochastic
version of the Euler fluid equation (\ref{eq:2DEulerVorticity}) as
derived in \cite{Holm2015} is given by the following damped and forced stochastic
partial differential equation (SPDE)
\begin{equation}
\dd q+\mathcal{L}_{\vecu}qdt+\sum_{i=1}^{m}\mathcal{L}_{\boldsymbol{\xi}_{i}}q\circ dW_{t}^{i}=\left(Q-rq\right)dt\label{eq:stochastic2dEuler}
\end{equation}
where the vector fields $\boldsymbol{\xi}_{i}$ represent spatial
correlations defined by a velocity-velocity correlation matrix $C_{ij}=\boldsymbol{\xi}_{i}\boldsymbol{\xi}_{j}^{T}.$
Here $\mathcal{L}_{\vecu}q$ and $\mathcal{L}_{\boldsymbol{\xi}_{i}}q$
denote the Lie derivatives of $q$ with respect to the vector fields
$\vecu$ and $\boldsymbol{\xi}$ respectively. In particular we have
\[
\mathcal{L}_{\vecu}q=\left[\vecu,q\right]:=\left(\vecu\cdot\nabla\right)q-\left(q\nabla\cdot\right)\vecu.
\]

\begin{remark} \label{rem:spde_exist_uniq_sol}For the stochastic
 Euler fluid equations with no forcing in three dimensions defined
 on the torus $\mathbb{T}$, it is shown in \cite{CrFlHo2017} that,
 for an initial vorticity in $W^{2,2}\left(\mathbb{T}\right),$ there
 exists an unique local (strong) solution in the space $W^{2,2}\left(\mathbb{T}\right).$
 For the two dimensional case considered in the present work, a global existence and uniqueness proof
 is being prepared in \cite{CrOLa2018} 
\end{remark}
Equation (\ref{eq:stochastic2dEuler}) arises from the assumption that
the Eulerian transport velocity for this flow is described by the
Stratonovich stochastic differential equation \eqref{StochVF}.


\begin{remark} \label{rem:finite_num_eofs}
        One may ask whether the sum in \eqref{StochVF}
        should be over an infinite number of terms. For simplification, we make the assumption that $m$ is finite. This assumption allows us to avoid certain technical issues when we are interested in the practical aspects for data assimilation.
\end{remark}

Assuming $\vecu$ in \eqref{StochVF} is divergence free, and the $\boldsymbol{\xi}_{i}$ are taken to be the eigenvectors of the velocity-velocity correlation tensor, one can show that $\boldsymbol{\xi}_{i}$ are also divergent free vector fields. Hence for each $\boldsymbol{\xi}_{i},$
there exists a potential function, denoted by $\zeta_{i},$ such that
\[
\boldsymbol{\xi}_{i}=\nabla^{\perp}\zeta_{i}.
\]
Thus \eqref{StochVF} can be expressed in terms of $\psi$ and $\zeta_i$,
\begin{equation}
\dd \tilde{\vecx{x}} =\gradperp\psi dt+\sum_{i=1}^{m}\gradperp\zeta_{i}\circ dW_{t}^{i}.\label{eq:stochastic_transport_velocity_streamfunction_form}
\end{equation}
Expressing the transport velocity in this form is useful because it
allows us to introduce stochastic perturbation (i.e. terms with $\circ dW_{t}^{i}$)
via the streamfunction when solving the SPDE system numerically, thereby
keeping the discretisation of (\ref{eq:stochastic2dEuler}) the same
as the deterministic equation (\ref{eq:2DEulerVorticity}). In other words, upon using
the divergence free properties of $\vecu$ and $\boldsymbol{\xi}_{i}$,
we can rewrite (\ref{eq:stochastic2dEuler}) to obtain
\begin{equation}
\dd q+\nabla^{\perp}\left(\psi dt+\sum_{i=1}^{m}\zeta_{i}\circ dW_{t}^{i}\right)\cdot\nabla q=\left(Q-rq\right)dt
\,,\label{eq:2DEulerSpdeDeterministicForm}
\end{equation}
which has the same form as (\ref{eq:2DEulerVorticity}).
Thus, our numerical algorithm for solving the SPDE system is largely the
same as Algorithm \ref{alg:det_sys_alg}.

We describe the numerical method in the next subsection and show that it is consistent
with the SPDE. 

\begin{remark}
        Equation (\ref{eq:stochastic2dEuler}) is in Stratonovich form. To
        obtain the equivalent It\^o form of (\ref{eq:stochastic2dEuler}) we
        apply the identity
        \begin{equation}
        \int_{0}^{t}\mathcal{L}_{\boldsymbol{\xi}_{i}}q\left(s\right)\circ dW_{s}^{i}=\int_{0}^{t}\mathcal{L}_{\boldsymbol{\xi}_{i}}q\left(s\right)dW_{s}^{i}+\frac{1}{2}\left\langle \mathcal{L}_{\boldsymbol{\xi}_{i}}q,W^{i}\right\rangle _{t}\label{eq:ito_stochastic2deuler}
        \end{equation}
        where $\left\langle .,.\right\rangle _{t}$ is the cross-variation
        process and 
        \begin{align*}
        \left\langle \mathcal{L}_{\boldsymbol{\xi}_{i}}q,W^{i}\right\rangle _{t} & =\mathcal{L}_{\boldsymbol{\xi}_{i}}\left\langle q,W^{i}\right\rangle _{t}\\
                & =\mathcal{L}_{\boldsymbol{\xi}_{i}}\left\langle \int\{(Q-rq)dt-\mathcal{L}_{\vecu}qdt-\sum_{j=1}^{\infty}\mathcal{L}_{\boldsymbol{\xi}_{j}}q\circ dW_{t}^{j}\},W^{i}\right\rangle _{t}\\
                & =\mathcal{L}_{\boldsymbol{\xi}_{i}}\left\langle -\int_{0}^{.}\mathcal{L}_{\boldsymbol{\xi}_{i}}q\circ dW_{s}^{i},W^{i}\right\rangle _{t}\\
                & =\mathcal{L}_{\boldsymbol{\xi}_{i}}\left(-\int_{0}^{t}\mathcal{L}_{\boldsymbol{\xi}_{i}}q\left(s\right)ds\right)\\
                & =-\int_{0}^{t}\mathcal{L}_{\boldsymbol{\xi}_{i}}^{2}q(s)ds
        \end{align*}
        Hence
        \[
        \int_{0}^{t}\mathcal{L}_{\boldsymbol{\xi}_{i}}q\left(s\right)\circ dW_{s}^{i}=\int_{0}^{t}\mathcal{L}_{\boldsymbol{\xi}_{i}}q\left(s\right)dW_{s}^{i}-\frac{1}{2}\int_{0}^{t}\mathcal{L}_{\boldsymbol{\xi}_{i}}^{2}q(s)ds
        \]
        and (\ref{eq:ito_stochastic2deuler}) is thus
        \begin{equation}
        \dd q+\mathcal{L}_{\vecu}qdt+\sum_{i=1}^{m}\mathcal{L}_{\boldsymbol{\xi}_{i}}q\ dW_{t}^{i}=\frac{1}{2}\sum_{i=1}^{m}\mathcal{L}_{\boldsymbol{\xi}_{i}}^{2}q\ dt+(Q-rq)dt\label{eq:ito_spde_euler}
        \end{equation}
        where $\mathcal{L}_{\boldsymbol{\xi}_{i}}^{2}q=\mathcal{L}_{\boldsymbol{\xi}_{i}}\left(\mathcal{L}_{\boldsymbol{\xi}_{i}}q\right)=\left[\boldsymbol{\xi}_{i},\left[\boldsymbol{\xi}_{i},q\right]\right]$
        is the double Lie derivative of $q$ with respect to the divergence free vector field $\boldsymbol{\xi}_{i}$.
\end{remark}

\subsubsection{Numerical implementation}

The SPDE system (\ref{eq:stochastic2dEuler}) has Stratonovich stochastic
terms. Consequently, to solve it numerically, the scheme must take this into
account. Of course one could also work with the corresponding It\^o
form (\ref{eq:ito_spde_euler}), in which case the equation would have
a modified drift term.

To solve the stochastic system (\ref{eq:stochastic2dEuler}), we extend the SSPRK3 scheme used 
in the deterministic case. We will show that the numerical scheme we introduce
is consistent in the sense of Definition \ref{def:consistency}, see
\cite{lang2010lax}.

Let $\left(H_{p},\left(\cdot,\cdot\right)_{H_{p}}\right)$ be a separable
Hilbert space with norm $\left\Vert \cdot\right\Vert _{H_{p}}:=\sqrt{\left(\cdot,\cdot\right)_{H_{p}}}.$
In our setting, we have $H_{p}=W^{p,2},$ for some sufficiently large
$p,$ see Remark \ref{remark:well-posedness}. Let $V_{h}$ be a
finite dimensional subspace of $H_{p}.$ The parameter $h\in\left(0,1\right]$
controls the dimension of $V_{h}.$ For every $h,$ let $P_{h}:H_{p}\rightarrow V_{h}$
denote the Ritz projection operator mapping elements of $H_{p}$ to
the finite dimensional subspace $V_{h}$ such that 
\[
\lim_{h\rightarrow0}\left\Vert P_{h}q-q\right\Vert _{H}=0
\]
for all $q\in H_{p}$ and 
\begin{equation}
\left(P_{h}q,v_{h}\right)_{H_{p}}=\left(q,v_{h}\right)_{H_{p}}\label{eq:Ritz_projection}
\end{equation}
for all $q\in H_{p}$ and $v_{h}\in V_{h}.$

Let $f:H_{p}\times H_{p}\rightarrow H_{p-1}$ denote a nonlinear operator
which is affine in the second variable, and $g^{i}:H_{p}\rightarrow H_{p-1},$
$i=1,2,\dots,m,$ are linear mappings from $H_{p}$ to $H_{p-1}$.
For notational convenience, we shall write $f\left(\cdot\right):=f\left(\cdot,\cdot\right)$
when the two arguments are the same. Consider the following Stratonovich
SPDE 
\begin{equation}
dq(t)=f\left(q(t)\right)dt+\sum_{i=1}^{m}g^{i}\left(q(t)\right)\circ dW_{t}^{i}.\label{eq:stratonovich_spde}
\end{equation}
In our model (\ref{eq:stochastic2dEuler}), $f(q)=-\mathcal{L}_{\vecu}q+(Q-rq)$
and $g^{i}(q)=-\mathcal{L}_{\boldsymbol{\xi}_{i}}q.$ Since $\vecu$
and $q$ satisfy the relation \[\vecu = \nabla^\perp \psi =\nabla^\perp \Delta ^{-1} q,\] we write $f(q)$. 

As noted in Remark \ref{rem:spde_exist_uniq_sol}, for our choice of $f$ and $g^{i}$, for sufficiently large $p$ the SPDE is well-posed, see \cite{CrOLa2018}. However, we do not consider the well-posedness of \eqref{eq:stratonovich_spde} for general $f$ and $g^{i}$, it is beyond the scope of the present work. 

The stochastic SSPRK3 scheme for the SPDE (\ref{eq:stratonovich_spde}) is

\begin{align}
q^{(1)} & =q^{n}+f\left(q^{n}\right)\Delta+\sum_{i=1}^{m}g^{i}\left(q^{n}\right)\Delta W^{i}\nonumber \\
q^{\left(2\right)} & =\frac{3}{4}q^{n}+\frac{1}{4}\left(q^{\left(1\right)}+f\left(q^{\left(1\right)}\right)\Delta+\sum_{i=1}^{m}g^{i}\left(q^{\left(1\right)}\right)\Delta W^{i}\right)\nonumber \\
q^{n+1} & =\frac{1}{3}q^{n}+\frac{2}{3}\left(q^{\left(2\right)}+f\left(q^{\left(2\right)}\right)\Delta+\sum_{i=1}^{m}g^{i}\left(q^{\left(2\right)}\right)\Delta W^{i}\right),\label{eq:ssprk3_stochastic}
\end{align}
which computes the approximation $q^{n+1}$ given $q^{n}.$ We will show
this time stepping method is consistent in the next subsection.

We let $S_{\Delta}:H_{p}\times\Omega\rightarrow H_{p-1}$ denote the
one step stochastic SSPRK3 temporal discretisation, that is 
\[
q^{n+1}=S_{\Delta}\left(q^{n}\right).
\]
The operator $S_{\Delta}$ can be seen as the discrete approximation
of the solution semigroup operator $S\left(t\right):H_{p}\times\Omega\rightarrow H_{p}$
for (\ref{eq:stratonovich_spde}). In other words, for an initial
condition $q_{0}:\Omega\rightarrow H_{p}$ the solution to the SPDE
(\ref{eq:stratonovich_spde}) is given by 
\[
q(t)=S\left(t\right)\left( q_{0}\right).
\]
Note that in the continuous case, no differentiability is lost.

Consider the semi-discrete problem on $V_{h}$ 
\begin{equation}
\dd q_{h}(t)=f_{h}\left(q_{h}(t)\right)dt+\sum_{i=1}^{m}g_{h}^{i}\left(q_{h}(t)\right)\circ dW_{t}^{i}.\label{eq:semi_discrete_spde}
\end{equation}
where $f_{h}$ and $g_{h}^{i}$ are spatial approximations to the
operators $f$ and $g^{i}.$ In our implementation, like in the PDE case, we use 
finite element discretisation to obtain $f_{h}$
and $g_{h}^{i}.$ 
The scheme for the SPDE system is a combination of the stochastic SSPRK3 scheme for the temporal variable, and a mix of continuous and discontinuous Galerkin finite
element approximation for the spatial variables. Algorithm \ref{alg:stochastic_sys_alg} summarises the numerical methodology for the SPDE system. It is largely the same as Algorithm \ref{alg:det_sys_alg}, with the differences (i.e. the additional stochastic terms) highlighted in red. Note that at each corresponding step in the algorithm, we add the perturbations via the streamfunction, see \eqref{eq:2DEulerSpdeDeterministicForm}, the result of which is then used to obtain the velocity field $\vecu$ used in the subsequent numerical step. 

By substituting $q_{h}^{(1)}$ into $q_{h}^{(2)},$ and $q_{h}^{(2)}$
into $q^{n+1}$ in (\ref{eq:ssprk3_stochastic}) and expanding $f_{h}$
and $g_{h}^{i}$, the combined spatial and temporal scheme can be
expressed in leading order terms as 
\begin{equation}
q_{h}^{n+1}=S_{\Delta}\left( q_{h}^{n} \right)=q_{h}^{n}+f_{h}\left(q_{h}^{n}\right)\Delta+\sum_{i=1}^{m}g_{h}^{i}\left(q_{h}^{n}\right)\Delta W^{i}+\frac{1}{2}\sum_{i,j=1}^{m}g_{h}^{i}g_{h}^{j}\left(q_{h}^{n}\right)\Delta W^{i}\Delta W^{j}+H.O.T.\label{eq:ssprk3_leading_order}
\end{equation}
where $H.O.T.$ denotes \emph{higher order terms}.

\begin{algorithm}
        \caption{\label{alg:stochastic_sys_alg}Solver algorithm for the SPDE system}
        
        \begin{algorithmic}[1] 
                
                \STATE Let $\Delta t$ and $\Delta x$ be the time discretisation
                step and the spatial discretisation step respectively, such that they
                satisfy the CFL condition of $1/3.$ Let $q_{0}$ be a given initial
                vorticity at $t=0.$ 
                
                \FOR{$t_i = i\Delta t$, $i=0,1,2\dots,N-1$, with $t_N = T$}
                
                \STATE Set $q_{h}^{n}=q_{i}.$
                
                \textcolor{red}{\STATE Let $\theta_{i}:=\sum_{j}^{N}\zeta_{j}\Delta W^{j}$
                        for iid $\Delta W^{j}\sim\mathcal{N}\left(0,\Delta t\right)$}
                
                \STATE Solve 
                \[
                \left\langle \phi_{h},q_{h}^{n}\right\rangle _{\Omega}=-\left\langle \nabla\phi_{h},\nabla\psi_{h}^{n}\right\rangle _{\Omega}
                \]
                to obtain $\psi_{h}^{n}$. \textcolor{red}{Let $\tilde{\psi}_{h}^{n}:=\psi_{h}^{n}+\theta_{i}$
                        which we then use to solve }
                \[
                \left\langle v_{h},q_{h}^{(1)}\right\rangle _{K}=\left\langle v_{h},q_{h}^{n}\right\rangle _{K}-\Delta t\left(\left\langle \nabla v_{h},-q_{h}^{n}\nabla^{\perp}\tilde{\psi}_{h}^{n}\right\rangle _{K}-\left\langle v_{h},Q-rq^{n}\right\rangle _{K}+\left\langle v_{h},q^{n}\nabla^{\perp}\tilde{\psi}_{h}^{n}\cdot\hat{n}\right\rangle _{\partial K}\right)
                \]
                $K\in\mathcal{T}_{h},$ to obtain $q_{h}^{\left(1\right)}.$
                
                \STATE Solve
                \[
                \left\langle \phi_{h},q_{h}^{\left(1\right)}\right\rangle _{\Omega}=-\left\langle \nabla\phi_{h},\nabla\psi_{h}^{\left(1\right)}\right\rangle _{\Omega}
                \]
                to obtain $\psi_{h}^{\left(1\right)}.$ \textcolor{red}{Let $\tilde{\psi}_{h}^{\left(1\right)}:=\psi_{h}^{\left(1\right)}+\theta_{i}$
                        which we then use to solve }
                \[
                \begin{aligned}\left\langle v_{h},q_{h}^{\left(2\right)}\right\rangle _{K} & =\frac{3}{4}\left\langle v_{h},q_{h}^{n}\right\rangle _{K}+\frac{1}{4}\left\langle v_{h},q_{h}^{\left(1\right)}\right\rangle _{K}\\
                & \qquad-\frac{\Delta t}{4}\left(\left\langle \nabla v_{h},-q_{h}^{\left(1\right)}\nabla^{\perp}\tilde{\psi}_{h}^{\left(1\right)}\right\rangle _{K}-\left\langle v_{h},Q-rq_{h}^{\left(1\right)}\right\rangle _{K}+\left\langle v_{h},q_{h}^{\left(1\right)}\nabla^{\perp}\tilde{\psi}_{h}^{\left(1\right)}\cdot\hat{n}\right\rangle _{\partial K}\right)
                \end{aligned}
                \]
                $K\in\mathcal{T}_{h},$ to obtain $q_{h}^{\left(2\right)}.$
                
                \STATE Solve
                \[
                \left\langle \phi_{h},q_{h}^{\left(2\right)}\right\rangle _{\Omega}=-\left\langle \nabla\phi_{h},\nabla\psi_{h}^{\left(2\right)}\right\rangle _{\Omega}
                \]
                to obtain $\psi_{h}^{\left(2\right)}.$ \textcolor{red}{Let $\tilde{\psi}_{h}^{\left(2\right)}:=\psi_{h}^{\left(2\right)}+\theta_{i}$
                        which we then use to solve }
                
                \[
                \begin{aligned}\left\langle v_{h},q_{h}^{n+1}\right\rangle _{K} & =\frac{1}{3}\left\langle v_{h},q_{h}^{n}\right\rangle _{K}+\frac{2}{3}\left\langle v_{h},q_{h}^{\left(2\right)}\right\rangle _{K}\\
                & \qquad-\frac{2\Delta t}{3}\left(\left\langle -q_{h}^{\left(2\right)}\nabla^{\perp}\tilde{\psi}_{h}^{\left(2\right)},\nabla v_{h}\right\rangle _{K}-\left\langle v_{h},Q-rq_{h}^{\left(2\right)}\right\rangle _{K}+\left\langle v_{h},q_{h}^{\left(2\right)}\nabla^{\perp}\tilde{\psi}_{h}^{\left(2\right)}\cdot\hat{n}\right\rangle _{\partial K}\right)
                \end{aligned}
                \]
                $K\in\mathcal{T}_{h},$ to obtain $q_{h}^{n+1}.$
                
                \STATE Set $q_{i+1}=q_{h}^{n+1}.$ 
                
                \ENDFOR
                
        \end{algorithmic}
\end{algorithm}

\subsubsection{Consistency of the numerical method for the SPDE}

We now define consistency for the time stepping scheme for the SPDE
(\ref{eq:stratonovich_spde}). 

The approximation operator $S_{\Delta}$ can be decomposed into a
deterministic part $S_{\Delta}^{d}$ and a stochastic part $S_{\Delta}^{s}.$
The deterministic part correspond to the SSPRK3 discretisation \eqref{eq:ssprk3}
for the PDE system \eqref{eq:2DEulerVorticity}. The stochastic part
$S_{\Delta}^{s}$ should satisfy additional compatibility condition, see Definition \ref{def:stochastic_compatibility}, in order to be consistent
with Stratonovich integrals.

\begin{definition}
                The local truncation error $e_{j}\left(\Delta\right)$ of the discretisation 
                scheme (\ref{eq:ssprk3_stochastic}) is defined by 
                \begin{equation}
                e_{j}\left(\Delta\right)=q\left(t_{j+1}\right)-S_{\Delta}q\left(t_{j}\right).\label{eq:lte}
                \end{equation}
                The corresponding deterministic local truncation error is 
                \begin{equation}
                e_{j}^{d}\left(\Delta\right)=\omega\left(t_{j+1}\right)-S_{\Delta}^{d}\omega\left(t_{j}\right),\label{eq:deterministic_lte}
                \end{equation}
                where $\omega$ solves the deterministic system \eqref{eq:2DEulerVorticity}.
\end{definition}

\begin{definition}
        \label{def:stochastic_compatibility}For some
        $\gamma>1$, the discrete approximation operators $S_{\Delta}^{s}$
        is called $\mathcal{F}$-compatible if $S_{\Delta}^{s}$ is $\mathcal{F}_{t_{j+1}}$
        measurable and 
        \begin{equation}
        E\left(\left.S_{\Delta}^{s}\left(q\left(t_{j}\right)\right)\right|\mathcal{F}_{t_{j}}\right)=\frac{1}{2}\Delta\sum_{i=1}^{m}g^{i}g^{i}q\left(t_{j}\right)+\mathcal{O}(\Delta^{\gamma})\label{eq:stochastic_compatibility}
        \end{equation}
        for all $j=0,\dots,n-1.$ 
\end{definition}

\begin{definition}{[}Consistency{]} 
                \label{def:consistency} We say
                the numerical scheme $S_{\Delta}$ is consistent in mean square of
                order $\gamma>1$ with respect to (\ref{eq:stratonovich_spde}) if
                there exists a constant $c$ independent of $\Delta\in(0,T]$ and
                if for all $\epsilon>0,$ there exist $\eta,\delta>0$ such that for
                all $0<\Delta<\delta,$ and $j\in\left\{ 1,2,\dots,N\right\} $ 
                \begin{equation}
                E\left(\left\Vert e_{j}\left(\Delta\right)\right\Vert _{H}^{2}\right)<c\Delta^{\gamma}\label{eq:consistency}
                \end{equation}
                and 
                \begin{equation}
                E\left(\left\Vert e_{j}^{d}\left(\Delta\right)\right\Vert _{H}\right)<c\Delta^{\gamma}\label{eq:deterministic_lte_consistency}
                \end{equation}
                and $S_{\Delta}^{s}$ is $\mathcal{F}$-compatible. 
\end{definition}

Henceforth we introduce
a few notational simplifications. Let $f_{s}(q_{t})$ denote $f\left(q_{s},q_{t}\right)$
when $f$ depends on the solution at two different times. In our model,
this means $f_{s}(q_{t})=-\mathcal{L}_{{\bf u}_{s}}q_{t}+(Q-rq_{t}),$
and without the subscript we have $f(q_{t})=-\mathcal{L}_{{\bf u}_{t}}q_{t}+(Q-rq_{t}).$
Also, since $f_{s}(\cdot)$ is affine, we can express it
as the summation of a linear part and a translation, $f_{s}(\cdot)=A_{s}(\cdot)+B,$ where $A_{s}$ denotes the linear part, and $B$ denotes the translation.

\begin{assumption} \label{assump:regularity}
        We assume the following are bounded, 
        
        $E\left(\sup_{k}\left\Vert f_{k}q\left(t_{n}\right)-f\left(q^{n}\right)\right\Vert _{H}^{2}\right)$,
        
        $E\left(\sup_{s}\sup_{r}\left\Vert A_{s}f_{r}\left(q_{r}\right)\right\Vert _{H}^{2}\right)$,
        
        $E\left(\sup_{s}\sup_{r}\left\Vert A_{s}g^{i}g^{i}\left(q_{r}\right)\right\Vert _{H}^{2}\right)$,
        
        $E\left(\sup_{s}\sup_{r}\left\Vert A_{s}g^{i}\left(q_{r}\right)\right\Vert _{H}^{2}\right)$,
        
        $E\left(\sup_{r}\left\Vert g^{i}f_{r}\left(q_{r}\right)\right\Vert _{H}^{2}\right)$,
        
        $E\left(\sup_{r}\left\Vert g^{i}g^{j}\left(q_{r}\right)\right\Vert _{H}^{2}\right)$,
        
        $E\left(\sup_{r}\left\Vert g^{i}g^{j}g^{j}\left(q_{r}\right)\right\Vert _{H}^{2}\right)$.
        
        We also assume that the terms in \emph{H.O.T.} in \eqref{eq:ssprk3_leading_order} are bounded in expected $H$ norm squared.
\end{assumption}
{
\begin{remark}\label{remark:well-posedness}
	Assumption A\ref{assump:regularity} holds provided the SPDE is well-posed and for all
	$T>0$ and for sufficiently large $p$, we have 
	\[
	E\left(\sup_{t\in\left[0,T\right]}\left\Vert q(t)\right\Vert _{p,2}^{2}\right)<\infty,
	\]
	see \cite{CrOLa2018}.
\end{remark}
}
\begin{lemma}
        \label{lem:consistency} Assuming the SPDE (\ref{eq:stratonovich_spde})
        is well-posed, and Assumption \ref{assump:regularity} is satisfied, the numerical scheme
        $S_{\Delta}$ described by (\ref{eq:ssprk3_stochastic}) is consistent
        with $\gamma=2$. 
\end{lemma}

We prove Lemma \ref{lem:consistency} next. To prove Lemma \ref{lem:consistency}, first note that the mean square of the local truncation error (\ref{eq:consistency})
can be bounded as follows. 

\begin{lemma}
        \begin{align}
        E\left(\left\Vert e_{n}\left(\Delta\right)\right\Vert _{H}^{2}\right) & \leq4E\left(\left\Vert \int_{t_{n}}^{t_{n+1}}\left(f\left(q(s)\right)-f\left(q^{n}\right)\right)ds\right\Vert _{H}^{2}\right)\nonumber \\
        & +4E\left(\left\Vert \sum_{i=1}^{m}\int_{t_{n}}^{t_{n+1}}\left(g^{i}\left(q_{s}\right)-g^{i}\left(q^{n}\right)\right)dW_{s}^{i}\right\Vert _{H}^{2}\right)\nonumber \\
        + & 4E\left(\left\Vert \frac{1}{2}\sum_{i=1}^{m}\left(\int_{t_{n}}^{t_{n+1}}g^{i}g^{i}\left(q_{s}\right)ds-\sum_{j=1}^{m}g^{i}g^{j}\left(q^{n}\right)\Delta W^{i}\Delta W^{j}\right)\right\Vert _{H}^{2}\right)\nonumber \\
        & +4E\left(\left\Vert H.O.T. \right\Vert _{H}^{2}\right)\label{eq:lte_decomposition}
        \end{align}
\end{lemma}

\begin{proof} 
        Writing the SPDE (\ref{eq:stratonovich_spde}) in It\^o
        integral form we have 
        \begin{equation}
        q\left(t_{n+1}\right)  =q\left(t_{n}\right)+\int_{t_{n}}^{t_{n+1}}f\left(q_{s}\right)ds+\sum_{i=1}^{m}\int_{t_{n}}^{t_{n+1}}g^{i}\left(q_{s}\right)dW_{s}^{i}
        +\frac{1}{2}\sum_{i=1}^{m}\int_{t_{n}}^{t_{n+1}}g^{i}g^{i}\left(q_{s}\right)ds \label{eq:ito_spde_integral_form}
        \end{equation}
        
        Thus, using (\ref{eq:ssprk3_leading_order}) we get
        
        \begin{align}
        E\left(\left\Vert q\left(t_{n+1}\right)-S_{\Delta}\left(q^{n}\right)\right\Vert _{H}^{2}\right) & =E\left(\left\Vert \int_{t_{n}}^{t_{n+1}}\left(f\left(q_{s}\right)-f\left(q^{n}\right)\right)ds\right.\right.\nonumber \\
        & +\sum_{i=1}^{m}\int_{t_{n}}^{t_{n+1}}\left(g^{i}\left(q_{s}\right)-g^{i}\left(q^{n}\right)\right)dW_{s}^{i}\nonumber \\
        +\frac{1}{2}\sum_{i=1}^{m} & \left(\int_{t_{n}}^{t_{n+1}}g^{i}g^{i}\left(q_{s}\right)ds-\sum_{j=1}^{m}g^{i}g^{j}\left(q^{n}\right)\Delta W^{i}\Delta W^{j}\right)\nonumber \\
        - & \left.\left.H.O.T.\right\Vert _{H}^{2}\right)\label{eq:consistency_error}
        \end{align}
        Using the inequality $\left(x_{1}+x_{2}+\dots+x_{n}\right)^{2}\leq n\left(x_{1}^{2}+x_{2}^{2}+\dots+x_{n}^{2}\right)$,
        we have the result. 
\end{proof} 

We bound each term in (\ref{eq:lte_decomposition}) individually.

\begin{lemma}
        \[
        E\left(\left\Vert \int_{t_{n}}^{t_{n+1}}\left(f\left(q_{s}\right)-f\left(q^{n}\right)\right)ds\right\Vert _{H}^{2}\right)=\mathcal{O}\left(\Delta^{2}\right)
        \]
\end{lemma}

\begin{proof}Using (\ref{eq:ito_spde_integral_form}), we have 
        \[
        f\left(q_{s}\right)=f_{s}q(t_{n})+\int_{t_{n}}^{s}A_{s}f_{r}\left(q_{r}\right)dr+\sum_{i=1}^{m}\int_{t_{n}}^{s}A_{s}g^{i}\left(q_{r}\right)dW_{r}^{i}+\frac{1}{2}\sum_{i=1}^{m}\int_{t_{n}}^{s}A_{s}g^{i}g^{i}\left(q_{r}\right)dr
        \]
        where $A_{s}$ is the linear part of $f_{s}.$ Hence 
        \[
        \begin{aligned}E\left(\left\Vert \int_{t_{n}}^{t_{n+1}}\left(f_{s}\left(q_{s}\right)-f\left(q^{n}\right)\right)ds\right\Vert _{H}^{2}\right) & =E\left(\left\Vert \int_{t_{n}}^{t_{n+1}}f_{s}q(t_{n})-f(q^{n})ds+\int_{t_{n}}^{t_{n+1}}\int_{t_{n}}^{s}A_{s}f_{r}\left(q_{r}\right)drds\right.\right.\\
        +\sum_{i=1}^{m}\int_{t_{n}}^{t_{n+1}} & \int_{t_{n}}^{s}A_{s}g^{i}\left(q_{r}\right)dW_{r}^{i}ds\left.\left.+\frac{1}{2}\sum_{i=1}^{m}\int_{t_{n}}^{t_{n+1}}\int_{t_{n}}^{s}A_{s}g^{i}g^{i}\left(q_{r}\right)drds\right\Vert _{H}^{2}\right).
        \end{aligned}
        \]
        Using the Cauchy--Schwarz inequality we have
        \[
        \begin{aligned}E\left(\left\Vert \int_{t_{n}}^{t_{n+1}}f_{s}q\left(t_{n}\right)-f\left(q^{n}\right)ds\right\Vert _{H}^{2}\right) & \leq\Delta\int_{t_{n}}^{t_{n+1}}E\left(\left\Vert f_{s}q\left(t_{n}\right)-f\left(q^{n}\right)\right\Vert _{H}^{2}\right)ds\\
        & \leq E\left(\sup_{k}\left\Vert f_{k}q\left(t_{n}\right)-f\left(q^{n}\right)\right\Vert _{H}^{2}\right)\Delta^{2}
        \end{aligned}
        \]
        \[
        \begin{aligned}E\left(\left\Vert \int_{t_{n}}^{t_{n+1}}\int_{t_{n}}^{s}A_{s}f_{r}\left(q_{r}\right)drds\right\Vert _{H}^{2}\right) & \leq\Delta\int_{t_{n}}^{t_{n+1}}\left(s-t_{n}\right)\int_{t_{n}}^{s}E\left(\left\Vert A_{s}f_{r}\left(q_{r}\right)\right\Vert _{H}^{2}\right)drds\\
        & \leq E\left(\sup_{s}\sup_{r}\left\Vert A_{s}f_{r}\left(q_{r}\right)\right\Vert _{H}^{2}\right)\Delta\int_{t_{n}}^{t_{n+1}}\left(s-t_{n}\right)^{2}ds\\
        & \leq E\left(\sup_{s}\sup_{r}\left\Vert A_{s}f_{r}\left(q_{r}\right)\right\Vert _{H}^{2}\right)\frac{1}{3}\Delta^{4}
        \end{aligned}
        \]
        \begin{align*}
        E\left(\left\Vert \frac{1}{2}\sum_{i=1}^{m}\int_{t_{n}}^{t_{n+1}}\int_{t_{n}}^{s}A_{s}g^{i}g^{i}\left(q_{r}\right)drds\right\Vert _{H}^{2}\right) & \leq\frac{m}{4}\sum_{i=1}^{m}\Delta\int_{t_{n}}^{t_{n+1}}(s-t_{n})\int_{t_{n}}^{s}E\left(\left\Vert A_{s}g^{i}g^{i}\left(q_{r}\right)\right\Vert _{H}^{2}\right)drds\\
        & \leq\frac{m^{2}}{12}E\left(\max_{i}\sup_{s}\sup_{r}\left\Vert A_{s}g^{i}g^{i}\left(q_{r}\right)\right\Vert _{H}^{2}\right)\Delta^{4}.
        \end{align*}
        Using the Cauchy--Schwarz inequality and It\^o isometry we have
        \begin{align*}
        E\left(\left\Vert \sum_{i=1}^{m}\int_{t_{n}}^{t_{n+1}}\int_{t_{n}}^{s}A_{s}g^{i}\left(q_{r}\right)dW_{r}^{i}ds\right\Vert _{H}^{2}\right) & \leq m\Delta\sum_{i=1}^{m}\int_{t_{n}}^{t_{n+1}}\int_{t_{n}}^{s}E\left(\left\Vert A_{s}g^{i}\left(q_{r}\right)\right\Vert _{H}^{2}\right)drds\\
        & \leq\frac{m^{2}}{2}E\left(\max_{i}\sup_{s}\sup_{r}\left\Vert A_{s}g^{i}\left(q_{r}\right)\right\Vert _{H}^{2}\right)\Delta^{3}.
        \end{align*}
        Collect the bounds together to obtain the result.
\end{proof}

\begin{lemma}
        \[
        E\left(\left\Vert \sum_{i=1}^{m}\int_{t_{n}}^{t_{n+1}}\left(g^{i}\left(q_{s}\right)-g^{i}\left(q^{n}\right)\right)dW_{s}^{i}\right\Vert _{H}^{2}\right)=\mathcal{O}\left(\Delta^{2}\right)
        \]
\end{lemma}

\begin{proof}
        Using (\ref{eq:ito_spde_integral_form}), we obtain
        \[
        g^{i}\left(q_{s}\right)=g^{i}q(t_{n})+\int_{t_{n}}^{s}g^{i}f_{r}\left(q_{r}\right)dr+\sum_{j=1}^{m}\int_{t_{n}}^{s}g^{i}g^{j}\left(q_{r}\right)dW_{r}^{j}+\frac{1}{2}\sum_{j=1}^{m}\int_{t_{n}}^{s}g^{i}g^{j}g^{j}\left(q_{r}\right)dr.
        \]
        
        Hence
        \[
        \begin{aligned}E\left(\left\Vert \sum_{i=1}^{m}\int_{t_{n}}^{t_{n+1}}\left(g^{i}\left(q_{s}\right)-g^{i}\left(q^{n}\right)\right)dW_{s}^{i}\right\Vert _{H}^{2}\right) & =E\left(\left\Vert \sum_{i=1}^{m}\int_{t_{n}}^{t_{n+1}}\int_{t_{n}}^{s}g^{i}f_{r}\left(q_{r}\right)drdW_{s}^{i}\right.\right.\\
        & +\sum_{i=1}^{m}\sum_{j=1}^{m}\int_{t_{n}}^{t_{n+1}}\int_{t_{n}}^{s}g^{i}g^{j}\left(q_{r}\right)dW_{r}^{j}dW_{s}^{i}\\
        & \left.\left.+\frac{1}{2}\sum_{i=1}^{m}\sum_{j=1}^{m}\int_{t_{n}}^{t_{n+1}}\int_{t_{n}}^{s}g^{i}g^{j}g^{j}\left(q_{r}\right)drdW_{s}^{i}\right\Vert _{H}^{2}\right).
        \end{aligned}
        \]
        We bound each term individually and obtain the following
        
        \begin{align*}
        E\left(\left\Vert \sum_{i=1}^{m}\int_{t_{n}}^{t_{n+1}}\int_{t_{n}}^{s}g^{i}f_{r}\left(q_{r}\right)drdW_{s}^{i}\right\Vert _{H}^{2}\right) & \leq m\sum_{i=1}^{m}\int_{t_{n}}^{t_{n+1}}E\left(\left\Vert \int_{t_{n}}^{s}g^{i}f_{r}\left(q_{r}\right)dr\right\Vert _{H}^{2}\right)ds\\
        & \leq\frac{m^{2}}{3}E\left(\max_{i}\sup_{r}\left\Vert g^{i}f_{r}\left(q_{r}\right)\right\Vert _{H}^{2}\right)\Delta^{3}
        \end{align*}
        \begin{align*}
        E\left(\left\Vert \sum_{i=1}^{m}\sum_{j=1}^{m}\int_{t_{n}}^{t_{n+1}}\int_{t_{n}}^{s}g^{i}g^{j}\left(q_{r}\right)dW_{r}^{j}dW_{s}^{i}\right\Vert _{H}^{2}\right) & \leq\frac{m^{4}}{2}E\left(\max_{i,j}\sup_{r}\left\Vert g^{i}g^{j}\left(q_{r}\right)\right\Vert _{H}^{2}\right)\Delta^{2}
        \end{align*}
        \begin{align*}
        E\left(\left\Vert \frac{1}{2}\sum_{i=1}^{m}\sum_{j=1}^{m}\int_{t_{n}}^{t_{n+1}}\int_{t_{n}}^{s}g^{i}g^{j}g^{j}\left(q_{r}\right)drdW_{s}^{i}\right\Vert _{H}^{2}\right) & \leq\frac{m^{4}}{12}E\left(\max_{i,j}\sup_{r}\left\Vert g^{i}g^{j}g^{j}\left(q_{r}\right)\right\Vert _{H}^{2}\right)\Delta^{3}.
        \end{align*}
        Collect the bounds together to obtain the result.
\end{proof}

\begin{lemma}
       \[
       E\left(\left\Vert \frac{1}{2}\sum_{i=1}^{m}\int_{t_{n}}^{t_{n+1}}g^{i}g^{i}\left(q_{s}\right)ds-\frac{1}{2}\sum_{i=1}^{m}\sum_{j=1}^{m}g^{i}g^{j}\left(q^{n}\right)\Delta W^{i}\Delta W^{j}\right\Vert _{H}^{2}\right)=\mathcal{O}\left(\Delta^{2}\right)
       \]
\end{lemma}

\begin{proof} We have
        \begin{align*}
        & E\left(\left\Vert \frac{1}{2}\sum_{i=1}^{m}\int_{t_{n}}^{t_{n+1}}g^{i}g^{i}\left(q_{s}\right)ds-\frac{1}{2}\sum_{i=1}^{m}\sum_{j=1}^{m}g^{i}g^{j}\left(q^{n}\right)\Delta W^{i}\Delta W^{j}\right\Vert _{H}^{2}\right)\\
        \leq & \frac{m}{4}\sum_{i=1}^{m}E\left(\left\Vert \int_{t_{n}}^{t_{n+1}}g^{i}g^{i}\left(q_{s}\right)ds-\sum_{j=1}^{m}g^{i}g^{j}\left(q^{n}\right)\Delta W^{i}\Delta W^{j}\right\Vert _{H}^{2}\right)\\
        \leq & \frac{m}{4}\sum_{i=1}^{m}E\left(\left\Vert \int_{t_{n}}^{t_{n+1}}g^{i}g^{i}\left(q(t_{n})\right)ds-\sum_{j=1}^{m}g^{i}g^{j}\left(q^{n}\right)\Delta 
        W^{i}\Delta W^{j}+\int_{t_{n}}^{t_{n+1}}\int_{t_{n}}^{s}H.O.T.\,ds\right\Vert _{H}^{2}\right)\\
        \leq & \frac{m}{4}\Delta^{2}\sum_{i=1}^{m}4\left[E\left(\left\Vert g^{i}g^{i}\left(q(t_{n})\right)\right\Vert _{H}^{2}\right)+m\sum_{j=1}^{m}E\left(\left\Vert g^{i}g^{j}\left(q^{n}\right)\right\Vert _{H}^{2}\right)\right]\\
        & \quad+\frac{m}{4}\sum_{i=1}^{m}2E\left(\left\Vert \int_{t_{n}}^{t_{n+1}}\int_{t_{n}}^{s} H.O.T.\,ds\right\Vert _{H}^{2}\right)
        \end{align*}
\end{proof}

By combining the estimates together we obtain the consistency condition
(\ref{eq:consistency}). For (\ref{eq:deterministic_lte_consistency}),
the deterministic part is simply the deterministic SSPRK3 scheme \eqref{eq:ssprk3},
see \citet{gottlieb2005high}. The compatibility condition \eqref{eq:stochastic_compatibility}
for the stochastic part follows from the leading order term expression
\eqref{eq:ssprk3_leading_order} of the stochastic SSPRK3 scheme
and the fact that
\[
E\left(\left.\sum_{j=1}^{m}g^{i}g^{j}\left(q^{n}\right)\Delta W^{i}\Delta W^{j}\right|\mathcal{F}_{t_{n}}\right)=E\left(\left.g^{i}g^{i}\left(q^{n}\right)\Delta W^{i}\Delta W^{i}\right|\mathcal{F}_{t_{n}}\right)=g^{i}g^{i}\left(q^{n}\right)\Delta
\]
as $W^{i}$ and $W^{j}$ are independent for $i\neq j.$

\section{Calibration of the correlation eigenvectors \label{sec:calibration}} %

\subsection{Methodology} \label{subsec:Methodology}

In the stochastic geophysical fluid dynamics framework, SPDEs are derived from the
starting assumption that (averaged) fluid particles satisfy the
equation
\begin{equation}
  \label{eq:barx}
\dd \vecx{x}(a,t) = \bar{\vecu}(\vecx{x}(a,t),t)d t + \sum_{i=1}^m
\vecx{\xi}_i(\vecx{x}(a,t)) \circ d W_t^{i},
\end{equation}
where $a$ is the Lagrangian label. The assumption \eqref{eq:barx}
leads, for example, to the Eulerian stochastic QG equation,
\begin{equation} \label{eq:sqg}
  \dd \bar{q}(\vecx{x},t) + \left(\bar{\vecu}(\vecx{x},t) d t + \sum_i\vecx{\xi}_i(\vecx{x},t)\circ
  d W_t^{i}\right)\cdot\nabla q(\vecx{x},t) = 0.
\end{equation}
Equation \eqref{eq:sqg} is what we actually solve. Equation \eqref{eq:barx} is not
explicitly solved. However, \eqref{eq:barx} describes the motion of fluid particles under
the SPDE solution, and is used to derive the SPDE. 

The goal of the stochastic PDE is to model the coarse-grained
components of a deterministic PDE that exhibits rapidly fluctuating
components. We can estimate the components $\vecx{\xi}_i$ in the stochastic
term by comparing \eqref{eq:barx} with the deterministic equation for
unapproximated trajectories,
\begin{equation}
  \label{eq:dx}
\diff {\vecx{x}}(a,t) = {\vecu}(\vecx{x}(a,t),t)d t, \quad \vecx{x}(a,0)=\vecx{x}_0^{a}
\end{equation}
moving with the unapproximated velocity $\vecu$ and starting from $\vecx{x}_0^{a}$. We assume that the
velocity can be written as $\vecu=\bar{\vecu}+\vecx{\zeta}$, where $\bar{\vecu}$ is a
spatially-filtered velocity that can be represented accurately in a
coarse-grid simulation. By comparing \eqref{eq:barx} and \eqref{eq:dx},
\begin{align*}
  \dd \vecx{x}(a,t) & = \bar{\vecu}(\vecx{x}(a,t),t)d t + \sum_{i=1}^m
  \vecx{\xi}_i(\vecx{x}(a,t)) \circ d W_t^{i} \\
  & \approx \vecu(\vecx{x}(a,t),t) dt,
\end{align*}
where we determine $\dd \vecx{x}(a,t)$ at the coarse resolution from $\vecu(\vecx{x}(a,t),t)\,dt$ at the fine resolution, 
we see that we are seeking an approximation such that
\begin{equation} \label{eq:lagrangian_approximation_model}
  \sum_i\vecx{\xi}_i(\vecx{x}(a,t),t)\circ   d W_t^{i}
  \approx
  \vecu(\vecx{x}(a,t),t)d t - \bar{\vecu}(\vecx{x}(a,t),t) d t.
\end{equation}

Our methodology is as follows. We spin up a fine grid simulation from
$t=-T_{spin}$ to $t=0$ (till some statistical equilibrium is reached),
then we record velocity time series from $t=0$ to $t=M\Delta t$, where
$\Delta t=k\delta t$ and $\delta t$ is the fine grid timestep. We
define $\vecx{X}_{ij}^0$ as coarse grid points.

For each $m=0,1,\ldots,M-1$, we
\begin{enumerate}
\item Solve $\dot{\vecx{X}}_{ij}(t)=\vecu(\vecx{X}_{ij}(t),t)$ with initial condition
  $\vecx{X}_{ij}(m\Delta t)=\vecx{X}^0_{ij}$, where $\vecu(x,t)$ is the solution from
  the fine grid simulation.
\item Compute $\bar{\vecu}_{ij}(t)$ by spatially averaging $u$ over the coarse
  grid box size around gridpoint $ij$. 
\item Compute $\bar{\vecx{X}}_{ij}$ by solving
  $\dot{\bar{\vecx{X}}}_{ij}(t) = \bar{\vecu}_{ij}(t)$ with the same initial condition.
\item Compute the difference $\Delta \vecx{X}_{ij}^m = \bar{\vecx{X}}_{ij}((m+1)\Delta t)
  - \vecx{X}_{ij}((m+1)\Delta t)$, which measures the
  error between the fine and coarse trajectory.
\end{enumerate}

Having obtained $\Delta \vecx{X}_{ij}^m$, we would like to extract the basis
for the noise. This amounts to a Gaussian model of the form
\[
\frac{\Delta \vecx{X}_{ij}^m}{\sqrt{\Delta t}} = \bar{\Delta \vecx{X}_{ij}} + \sum_{k=1}^N \vecx{\xi}_{ij}^k\Delta W^k_m,
\]
where $\Delta W^k_m$ are i.i.d. standard Gaussian random variables. 

We estimate $\vecx{\xi}$ by minimising
\[
\mathbb{E}\left[\left\|\sum_{ijm}\frac{\Delta \vecx{X}_{ij}^m}{\sqrt{\delta t}} - \bar{\Delta \vecx{X}_{ij}} - \sum_{k=1}^N \vecx{\xi}_{ij}^k\Delta W^k_m\right\|^2\right],
\]
where the choice of $N$ can be informed by using EOFs.

{
\begin{remark}[EOF]
	Empirical orthogonal functions (EOFs) can be thought of as principal components that correspond to the spatial correlations of a field, see \cite{Hannachi2004primer,HaJoSt2007}. In our case, EOFs are the eigenvectors of the velocity-velocity spatial covariance tensor. Writing the data time series $\Delta \vecx{X}_{ij}^m$, $m=0,\dots,M-1$ as a matrix $\tilde{F}$ whose entries are two dimensional vectors, and whose rows (row index $m$) correspond to serialised $\Delta \vecx{X}_{ij}^m$. Let $F:=\text{detrend}(\tilde{F})$ where the $\text{detrend}$ function removes the column mean from each entry. We then estimate the spatial covariance tensor by computing $R:=\frac{1}{M-1}F^TF$. We take the EOFs to be the eigenvectors of $R$, ranked in descending order according to the eigenvalues.
\end{remark}}


\subsection{Numerical experiments \label{sec:numerical_experiments}}

\subsubsection{Fine grid PDE solution and its coarse graining \label{subsec:pde_results}}

We solve the PDE system \eqref{eq:2DEulerVorticity}--\eqref{eq:boundary_condition_streamfunction} on a fine grid of size $512\times512$. Our choices for $\alpha$ and $\beta$ in the forcing term \eqref{eq:forcing} are $0.1$ and $8$, respectively. We apply a coarse-graining procedure to the fine grid solution to obtain its coarse-grained version on a coarse grid of size $64\times64$, which we call the \emph{truth}. The coarse-graining procedure comprises of two steps. First we apply spatial averaging, using for example the Helmholtz operator, to the fine grid streamfunction. Then this spatially-averaged streamfunction is directly projected onto the coarse grid. The coarse-grained versions of the vorticity and velocity fields are obtained from the coarse-grained streamfunction.

We choose the following initial configuration for the vorticity, denoted by $\omega_{\text{spin}}$,
\begin{equation}
\begin{aligned}
\omega_{\text{spin}} & =\sin(8\pi x)\sin(8\pi y)+0.4\cos(6\pi x)\cos(6\pi y)\\
& \qquad+0.3\cos(10\pi x)\cos(4\pi y) +0.02\sin(2\pi y)+0.02\sin(2\pi x) \label{eq:omega_spin}
\end{aligned}
\end{equation}
from which we spin--up the system until an energy equilibrium state seems to have been reached.  This equilibrium state, denoted by $\omega_{\text{initial}}$, is then chosen as the initial condition for our numerical experiments.

For our numerical experiments, time is measured in terms of \emph{large eddy turnover time}, henceforth abbreviated to \emph{ett}. It describes the time scale of the large scale flow features and is defined by $\tau_{L}=L/U$, where $L$ is the length scale of the largest eddy and $U$ is the mean velocity. {Using the numerical PDE solution we estimate that, in our setup $1$ ett is equivalent to $2.5$ time units corresponding to the deterministic system.}

Figure \ref{fig:omega_spin} shows a plot of $\omega_{\text{spin}}.$ As indicated in the legend of the plot, the red and blue colours represent the different signs of vorticity. The colour shades indicate the different levels of magnitude of the function. Figure \ref{fig:energy_series} shows a plot of the kinetic energy time series computed for the numerical PDE solution for $186$  ett, starting with the configuration $\omega_{\text{spin}}$. We see that the energy reaches an approximate equilibrium point after $40$ ett, or equivalently $100$ times units, at which point we take the numerical solution to be the initial condition $\omega_{\text{initial}}$. This chosen equilibrium point is marked as a red dot in Figure \ref{fig:energy_series}. 

Figure \ref{fig:pde_solution_t0} shows plots of vorticity $\omega$ (left column), velocity ${\vecu}$ (middle column) and streamfunction $\psi$ (right column) corresponding to the numerical PDE solution at time $t=0$. The top row shows the plots which correspond to the fine grid solution, and the bottom row shows the plots which correspond to the truth, i.e. the coarse-grained fine grid solution. Here the fine grid vorticity is exactly $\omega_{\text{initial}}$. For the vorticity scalar field, the red and blue colours represent opposing signs of the function, and the colour shades indicate the different levels of magnitude of the function. For the velocity field, scaled arrow fields are plotted to indicate the direction and magnitude of the velocity vectors at each spatial location, and the colours highlight the magnitude of the velocity vectors. For the streamfunction scalar field, the colours indicate the contour lines of the function, along which the velocity vectors travel. Due to coarse-graining, only the large scale features remain in the plots for the truth. This is most apparent for vorticity because it is the least smooth of the three functions. The loss of small scale details is also noticeable for the velocity field. For the streamfunction, one can see the coarse-grained streamfunction have slightly smoother contours when compared with the fine grid streamfunction. Additionally, the coarse-grained plots show weaker magnitudes for the vorticity and velocity fields when compared with the fine grid solutions.

Figure \ref{fig:pde_solution_t146} shows plots of vorticity, velocity and streamfunction corresponding to the numerical PDE solution at time $t=146$ ett with similar features to Figure \ref{fig:pde_solution_t0}.

\begin{figure}[htb]
    \begin{minipage}{0.45\textwidth}%
            \includegraphics[width=1\textwidth, right]{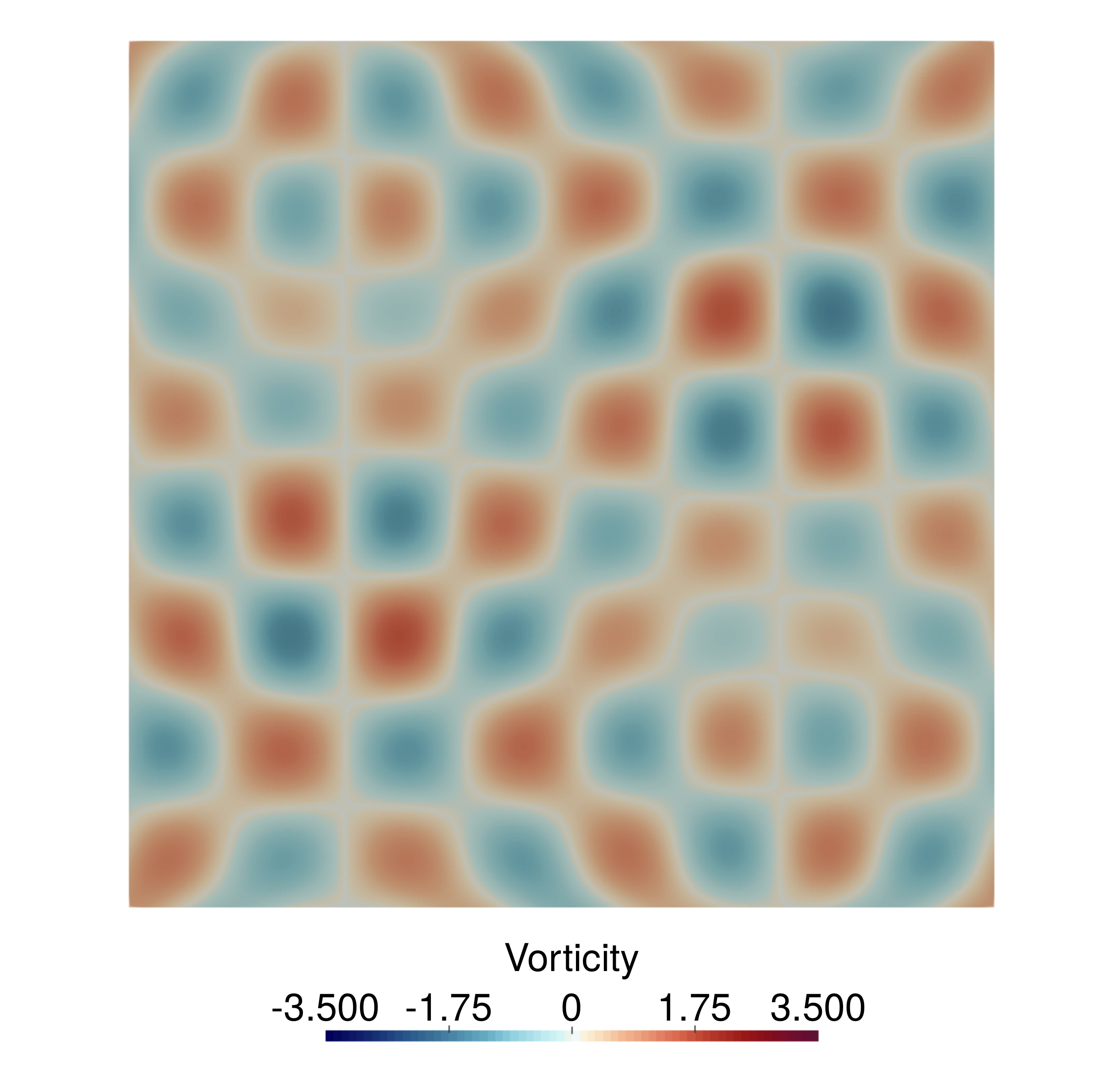}%
            \caption{\label{fig:omega_spin} This figure shows a plot of our chosen initial configuration $\omega_{\text{spin}}$ for the vorticity, given by Equation \eqref{eq:omega_spin}, from which we spin-up the PDE system until some energy equilibrium state. The red and blue colours represent opposing signs of the function. The colour shades indicate the different levels of magnitude of the function. See Section \ref{subsec:pde_results}.}
    \end{minipage}\hfill{\quad}%
    \begin{minipage}{0.45\textwidth}%
            \centering
                    \includegraphics[width=1\textwidth]{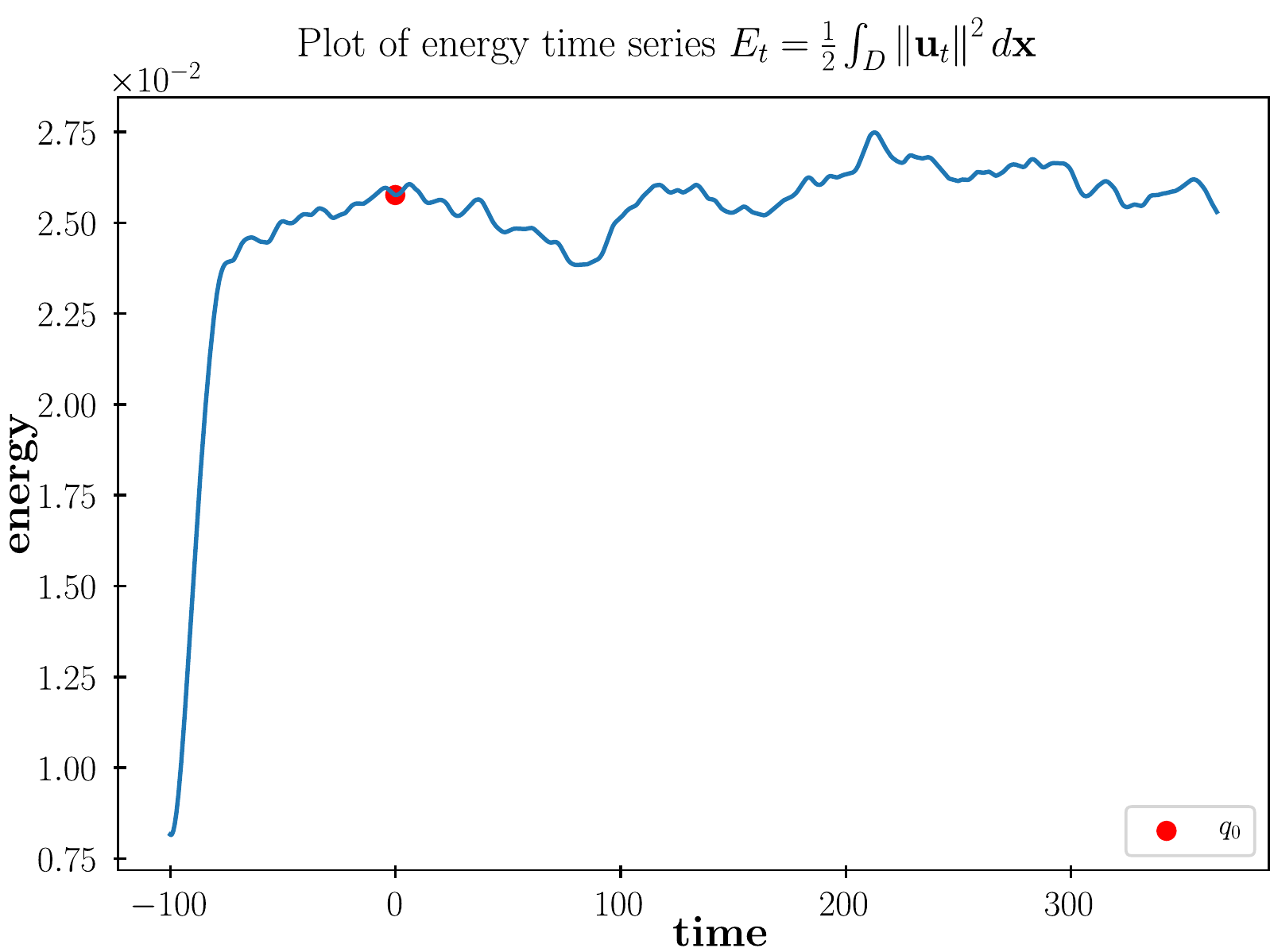}%
            \caption{\label{fig:energy_series} This figure shows a plot of the kinetic energy time series computed for the numerical PDE solution for $186$ large eddy turnover times (abbrev. ett), or equivalently $465$ time units, starting from the chosen initial configuration $\omega_{\text{spin}}$, see \eqref{eq:omega_spin}. The system reaches an approximate energy equilibrium state after $40$ ett, or $100$ time units, and the solution at this equilibrium point is set to be the initial condition $\omega_{\text{initial}}$ from which we start our numerical experiments. The plotted red dot marks the chosen energy equilibrium point. See Section \ref{subsec:pde_results}.}
    \end{minipage}
\end{figure}

\begin{figure}
        \begin{centering}
                \begin{minipage}[t]{1\textwidth}%
                        \begin{center}
                                \includegraphics[width=1\textwidth]{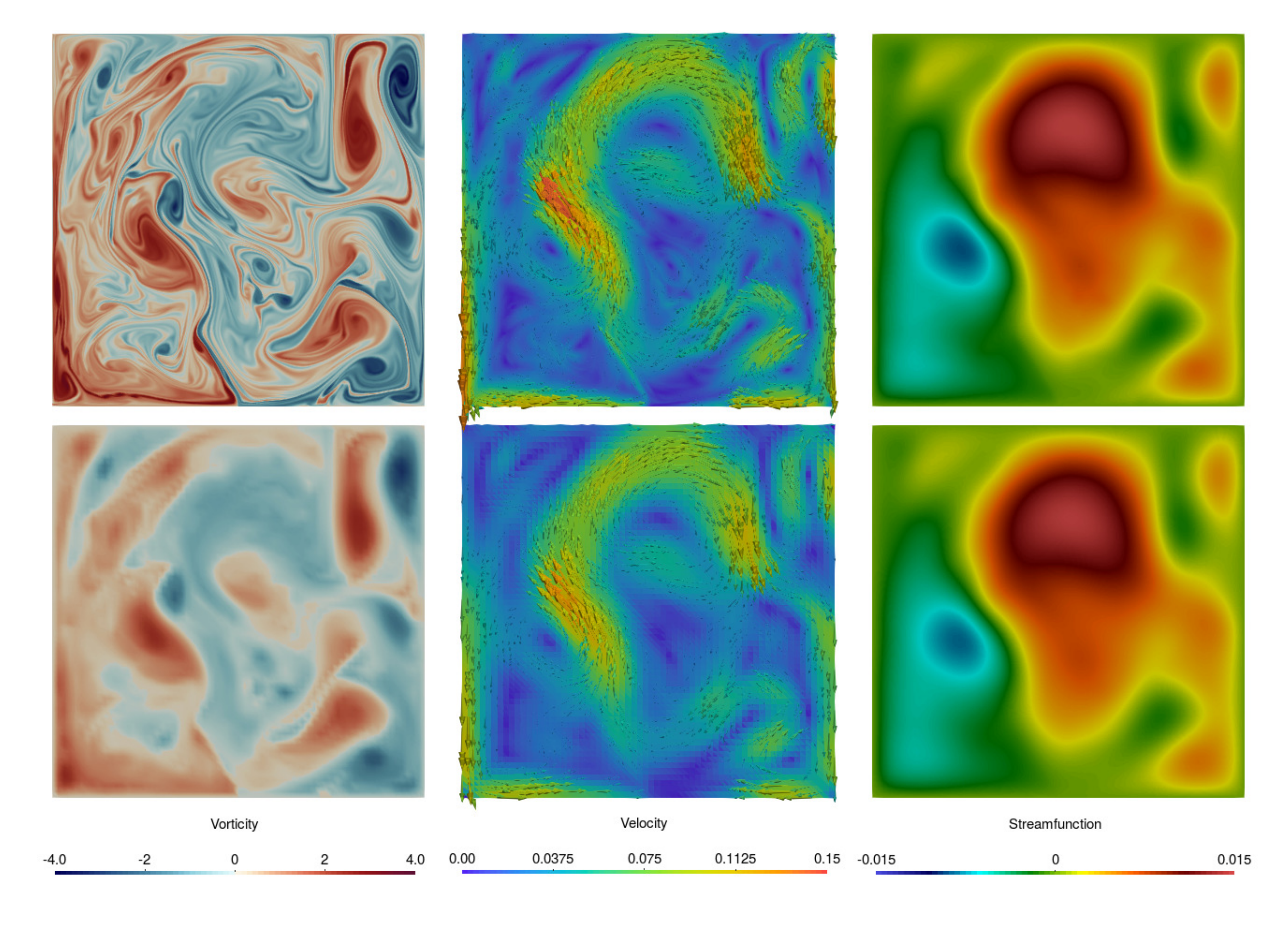}
                                \par\end{center}%
                                \centering{}\caption{\label{fig:pde_solution_t0} Plot of the numerical PDE solution at the initial time $t_{\text{initial}}$. This corresponds to the initial condition $\omega_{\text{initial}}$ (set after spin-up) from which we start our numerical experiments. The plots show vorticity (left column), velocity (middle column) and streamfunction (right column) for the truth (top row), which is defined on the fine grid, and its coarse-graining (bottom row), which is defined on the coarse grid. The coarse-graining is done via spatial averaging and projection of the fine grid streamfunction to the coarse grid. For the vorticity scalar field, the red and blue colours represent opposing signs of the function, and the colour shades indicate the different levels of magnitude of the function. For the velocity field, scaled arrow fields are plotted to indicate the direction and magnitude of the velocity vectors at each spatial location, and the colours highlight the magnitude of the velocity vectors. For the streamfunction scalar field, the colours indicate the contour lines of the function, along which the velocity vectors travel. Due to coarse-graining, only the large scale features remain. This is most apparent for the vorticity because it is the least smooth of the three functions. The loss of small scale details is also noticeable for the velocity field. For the streamfunction, one can see the coarse-grained streamfunction have slightly smoother contours when compared with the fine grid streamfunction. Additionally, the coarse-grained plots show weaker magnitudes for the vorticity and velocity fields when compared with the fine grid solutions. See Section \ref{subsec:pde_results}.}
                \end{minipage}
                \par\end{centering}
\end{figure}

\begin{figure}
        \begin{centering}
                \begin{minipage}[t]{1\textwidth}%
                        \begin{center}
                                \includegraphics[width=1\textwidth]{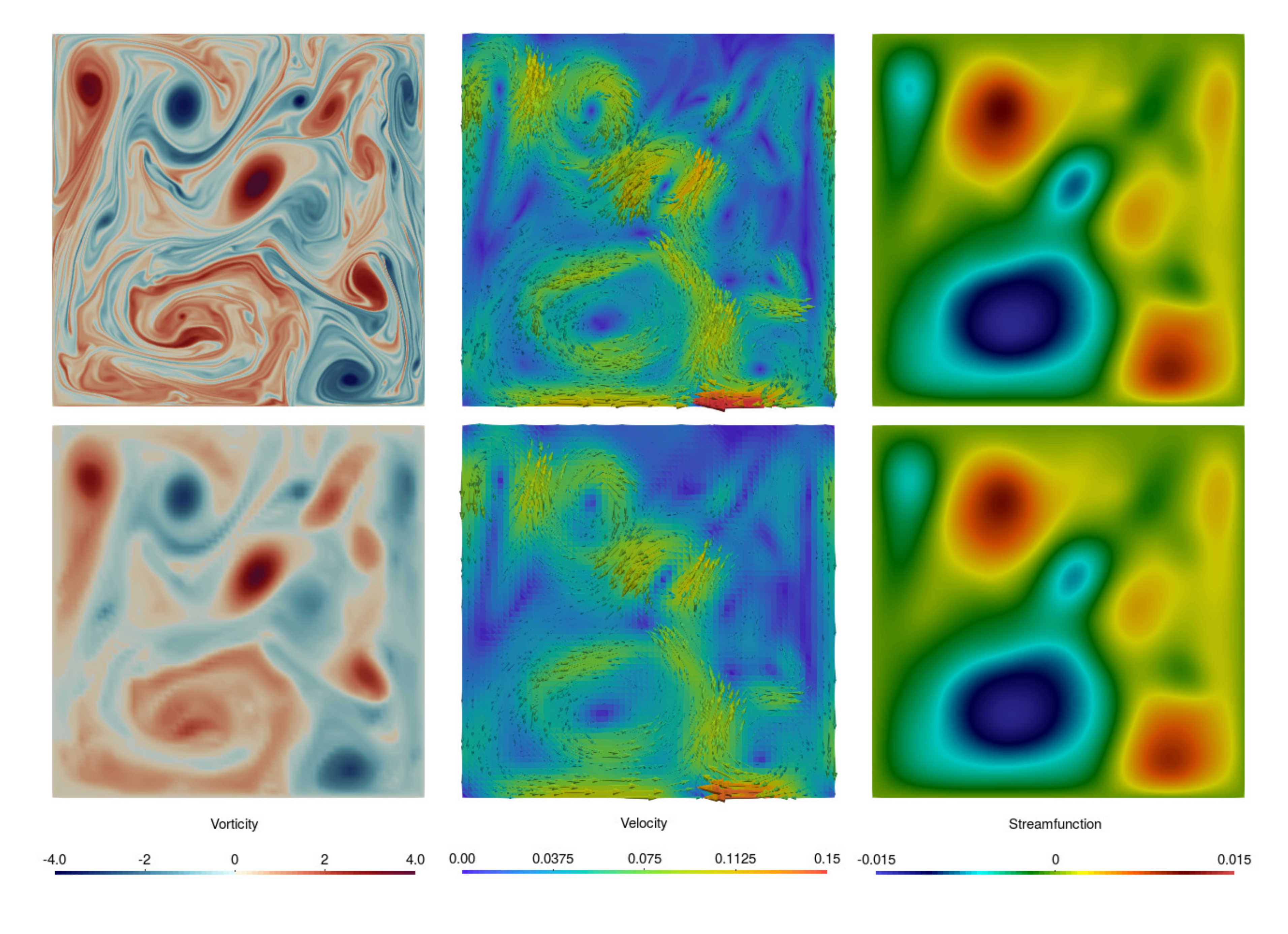}
                                \par\end{center}%
                        \centering{}\caption{\label{fig:pde_solution_t146} Plot of the numerical PDE solution at the final time of the PDE simulation, $t=t_{\text{initial}} + 146$ large eddy turnover times (abbrev. ett), i.e. the solution after $146$ ett starting from the initial vorticity $\omega_{\text{initial}}$. The plots show vorticity (left column), velocity (middle column) and streamfunction (right column) for the truth (top row), which is defined on the fine grid, and its coarse-graining (bottom row), which is defined on the coarse grid. The coarse-graining is done via spatial averaging and projection of the fine grid streamfunction to the coarse grid. For the vorticity scalar field, the red and blue colours represent opposing signs of the function, and the colour shades indicate the different levels of magnitude of the function. For the velocity field, scaled arrow fields are plotted to indicate the direction and magnitude of the velocity vectors at each spatial location, and the colours highlight the magnitude of the velocity vectors. For the streamfunction scalar field, the colours indicate the contour lines of the function, along which the velocity vectors travel. Due to coarse-graining, only the large scale features remain. This is most apparent for the vorticity because it is the least smooth of the three functions. The loss of small scale details is also noticeable for the velocity field. For the streamfunction, one can see the coarse-grained streamfunction have slightly smoother contours when compared with the fine grid streamfunction. Additionally, the coarse-grained plots show weaker magnitudes for the vorticity and velocity fields when compared with the fine grid solutions. See Section \ref{subsec:pde_results}.}
                \end{minipage}
                \par\end{centering}
\end{figure}

\subsubsection{Lagrangian trajectories and estimating the correlation eigenvectors \label{subsec:Lagrangian-trajectories-tests}}

Following the methodology described in Section \ref{subsec:Methodology}, we compute each $\Delta \vecx{X}_{ij}^{m},$ $m=0,1,\dots,M-1,$ with the time length $\Delta t$ equals to the coarse resolution time step, i.e. given the spatial resolution, $\Delta t$ satisfies the corresponding Courant-Friedrichs-Lewy (CFL) condition for the PDE system \eqref{eq:2DEulerVorticity}--\eqref{eq:boundary_condition_streamfunction}, see Section \ref{sec:numerical-implementation-(firedrake)-wei}. The computed $\Delta \vecx{X}_{ij}^{m}$ are then used to estimate the correlation eigenvectors $\vecx{\xi}_i,$ see Section \ref{subsec:Methodology}. 

We have assumed the sum $\sum_i \vecx{\xi}_i \circ dW_t^i$  is finite, see Remark \ref{rem:finite_num_eofs}. Let $n_{\xi}$ denote the number of $\vecx{\xi}_i$. Our choice for $n_{\xi}$ is informed by the computed eigenvalues, so that a given percentage of the total variability in $\Delta \vecx{X}_{ij}$ is captured. 

Figure \ref{fig:EOF-normalised-spectrum-64} shows a plot of the normalised spectrum. To illustrate how we choose $n_{\xi}$, the coloured dots: cyan, magenta and red, mark the number of $\vecx{\xi}_i$ required to capture $50\%$, $70\%$ and $90\%$ of the total variability in $\Delta \vecx{X}_{ij}$ respectively. As shown in the plot, to capture $50\%$ of the total variability, we need $n_{\xi}=51$; to capture $70\%$ of the total variability, we need $n_{\xi}=107$; to capture $90\%$ of the total variability, we need $n_{\xi}=225$. Note that for this numerical experiment, the coarse grid is of size $64\times64$. At this resolution, there are $4096$ EOFs in total.

We substitute the computed $\vecx{\xi}_i$ into \eqref{eq:barx} and simulate an ensemble of independent realisations of \eqref{eq:barx} in order to do Lagrangian trajectory uncertainty quantification tests. In these tests we wish to see if the estimated $\vecx{\xi}_i$ are adequate so that the approximation \eqref{eq:lagrangian_approximation_model} holds.

Figure \ref{fig:Lagrangian-trajectory-uq} shows a plot of the result for one indexed time interval (time interval indexed by $m$ in the methodology described in Section \ref{subsec:Methodology}), where we have simulated $200$ Lagrangian trajectories driven by the stochastic equation \eqref{eq:barx}, with the number of EOFs capturing $90\%$ of the total variance. We denote this by $n_{\xi} \equiv 90\%$. The stochastic trajectory positions are coloured using blue, and the deterministic trajectory positions are coloured using red. The length of time, $\Delta t$, over which \eqref{eq:barx} and \eqref{eq:dx} are simulated, corresponds to one coarse resolution PDE CFL time step and is too small to show significant deviations between the stochastic trajectories and the deterministic trajectory. Nevertheless, the result shows that at each position, the ensemble perfectly captures the deterministic trajectory. 

\begin{figure}[htb]
        \begin{minipage}[t]{0.48\textwidth}%
                \begin{center}
                        \includegraphics[width=1\textwidth]{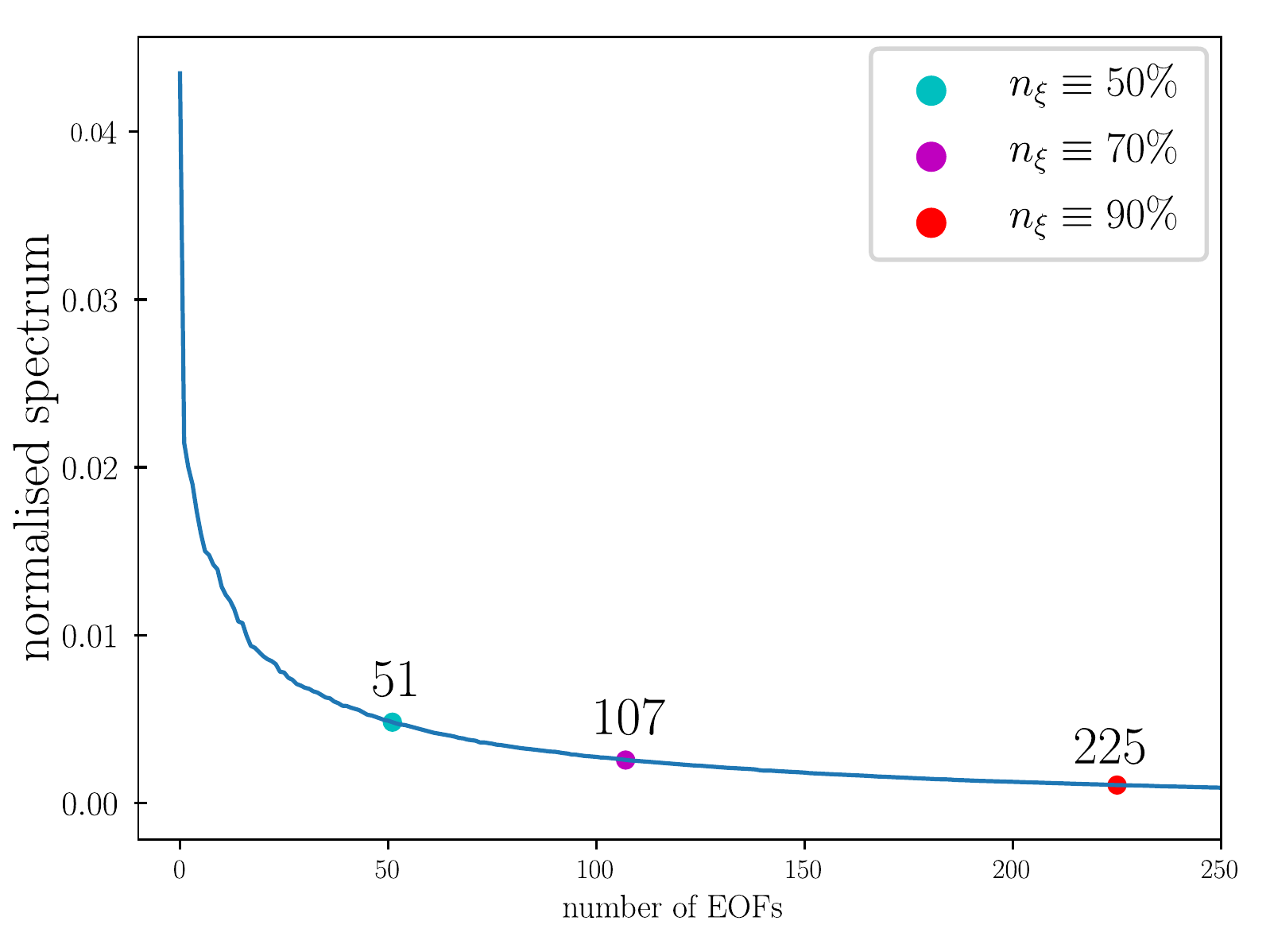}
                        \par\end{center}
                \caption{\label{fig:EOF-normalised-spectrum-64}EOF normalised spectrum for coarse grid
                        size $64\times64$. To illustrate how we choose $n_{\xi}$, the coloured dots: cyan, magenta and red, mark the number of $\vecx{\xi}_i$ required to capture $50\%$, $70\%$ and $90\%$ of the total variability in $\Delta \vecx{X}_{ij}$ respectively. As shown in the plot, to capture $50\%$ of the total variability, we need $n_{\xi}=51$; to capture $70\%$ of the total variability, we need $n_{\xi}=107$; to capture $90\%$ of the total variability, we need $n_{\xi}=225$. Note that at this resolution, there are $4096$ EOFs in total. See Section \ref{subsec:Lagrangian-trajectories-tests}.}
        \end{minipage}\hfill{}%
        \begin{minipage}[t]{0.48\textwidth}%
                \begin{center}
                        \includegraphics[width=1\textwidth]{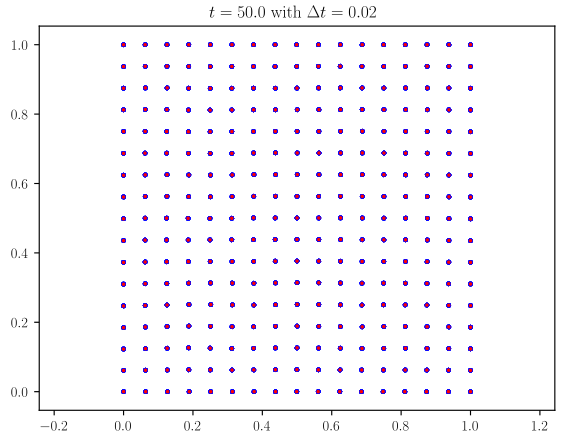}
                        \par\end{center}
                \caption{\label{fig:Lagrangian-trajectory-uq}Lagrangian trajectory uncertainty quantification corresponding to the time interval $[50,52)$ in simulation time
                        units, see Section \ref{subsec:Methodology}. There are $200$ stochastic Lagrangian trajectories (driven by the stochastic equation \eqref{eq:barx}), using a number of EOFs capturing $90\%$ of the total variance, $n_{\xi}\equiv90\%$. The length of time, $\Delta t$, over which \eqref{eq:barx} and \eqref{eq:dx} are simulated, corresponds to one coarse resolution PDE CFL time step. The stochastic trajectory positions are coloured using blue, and the deterministic trajectory positions are coloured using red. 
                        The time step $\Delta t$ is too small to show significant deviations between the stochastic trajectories and the deterministic trajectory. Nevertheless, the plot shows that at each position, the ensemble perfectly captures the deterministic trajectory. See Section \ref{subsec:Lagrangian-trajectories-tests}.}
        \end{minipage}
\end{figure}

We also apply the methodology described in Section \ref{subsec:Methodology} to more refined coarse grids, in particular grids of size $128\times128$ and $256\times256$, in order to investigate the impact of mesh refinement on uncertainty quantification for the SPDE. The results will be shown in the next subsection. 

Figure \ref{fig:EOF-normalised-spectrum-128} and Figure \ref{fig:EOF-normalised-spectrum-256} show plots of the normalised spectrum for coarse grids of size $128\times128$ and $256\times256$ respectively. The same coloured dots: cyan, magenta and red, are used to indicate the number of EOFs needed to capture $50\%$, $70\%$ and $90\%$ of the total variability in $\Delta \vecx{X}_{ij}$, as is shown in Figure \ref{fig:EOF-normalised-spectrum-64}. The results show that at each variance level, as the coarse grid gets refined, $n_{\xi}$ gets larger. For example, to capture $90\%$ variance, $n_{\xi} = 225$ for $64\times64$ coarse grid, $n_{\xi} = 277$ for $128\times128$ coarse grid, and $n_{\xi}=300$ for $256\times256$ coarse grid. At each resolution, \eqref{eq:barx} and \eqref{eq:dx} are simulated over a length of time, $\Delta t$, equivalent to one CFL time step for that resolution. Therefore it is likely that the deviations in $\Delta \vecx{X}_{ij}$ become more homogeneous as $\Delta t$ decreases, resulting in more EOFs required to explain a given percentage of the total variation.

\begin{figure}
        \begin{minipage}[t]{0.48\textwidth}%
                \begin{center}
                        \includegraphics[width=1\textwidth]{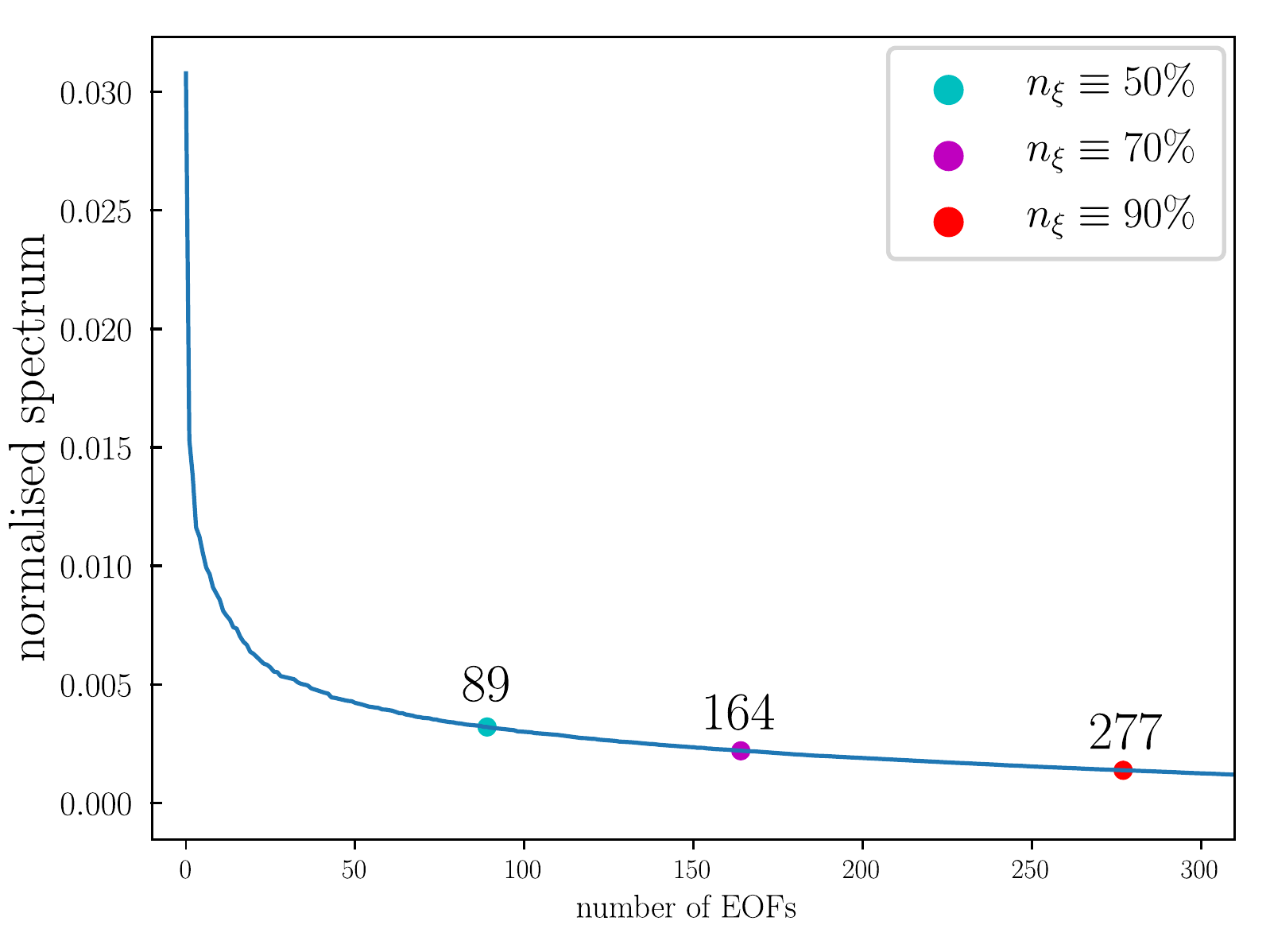}
                        \par\end{center}
                \caption{\label{fig:EOF-normalised-spectrum-128}EOF normalised spectrum for coarse grid
                        size $128\times128$. To illustrate how we choose $n_{\xi}$, the coloured dots: cyan, magenta and red, mark the number of $\vecx{\xi}_i$ required to capture $50\%$, $70\%$ and $90\%$ of the total variability in $\Delta \vecx{X}_{ij}$ respectively. As shown in the plot, to capture $50\%$ of the total variability, we need $n_{\xi}=89$; to capture $70\%$ of the total variability, we need $n_{\xi}=164$; to capture $90\%$ of the total variability, we need $n_{\xi}=277$. Note that at this resolution, there are $16384$ EOFs in total. See Section \ref{subsec:Lagrangian-trajectories-tests}.}
        \end{minipage}\hfill{}%
        \begin{minipage}[t]{0.48\textwidth}%
                \begin{center}
                        \includegraphics[width=1\textwidth]{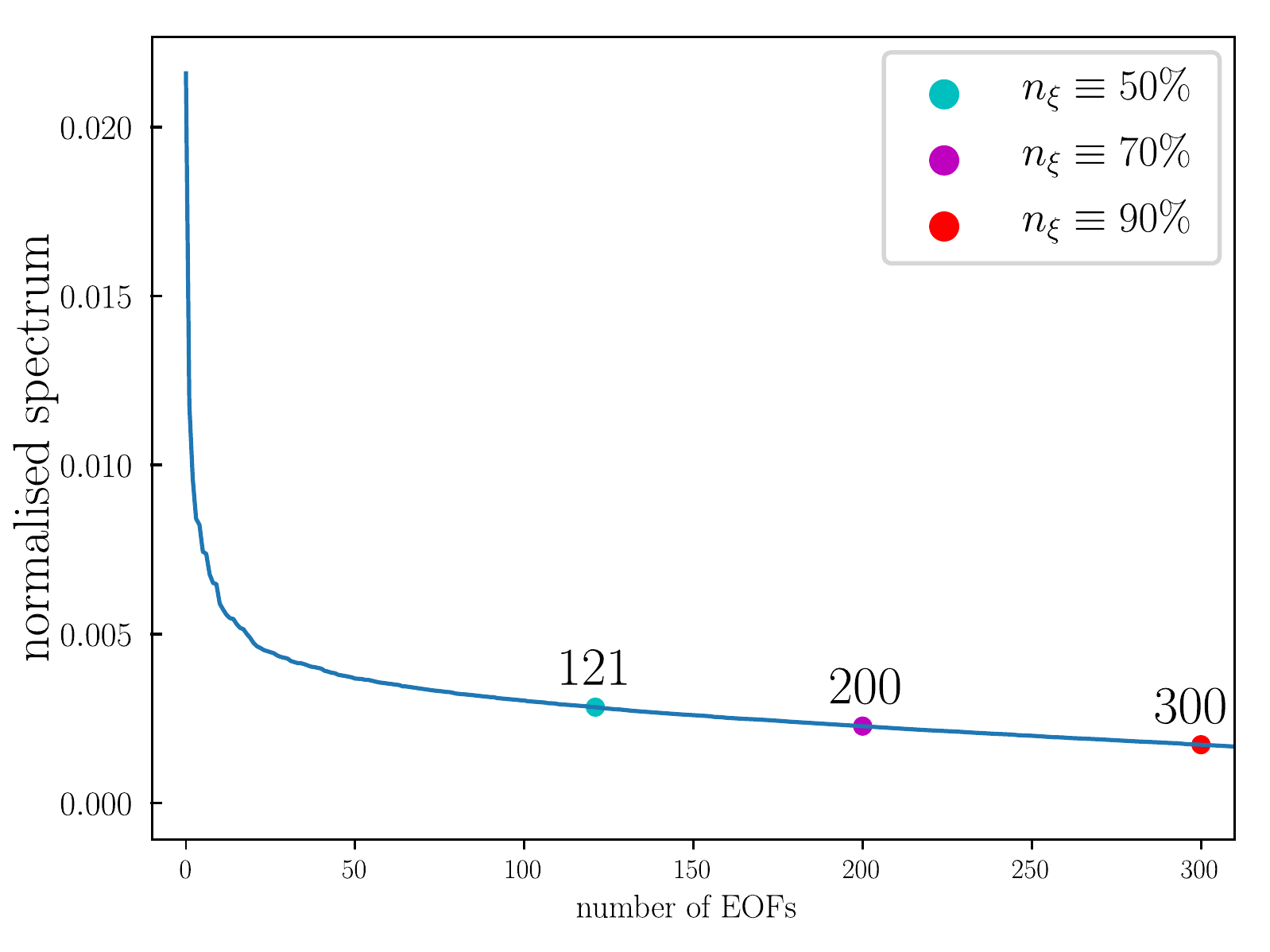}
                        \par\end{center}
                \caption{\label{fig:EOF-normalised-spectrum-256}EOF normalised spectrum for coarse grid
                        size $256\times256$. To illustrate how we choose $n_{\xi}$, the coloured dots: cyan, magenta and red, mark the number of $\vecx{\xi}_i$ required to capture $50\%$, $70\%$ and $90\%$ of the total variability in $\Delta \vecx{X}_{ij}$ respectively. As shown in the plot, to capture $50\%$ of the total variability, we need $n_{\xi}=121$; to capture $70\%$ of the total variability, we need $n_{\xi}=200$; to capture $90\%$ of the total variability, we need $n_{\xi}=300$. Note that at this resolution, there are $65536$ EOFs in total. See Section \ref{subsec:Lagrangian-trajectories-tests}.}
        \end{minipage}
\end{figure}

\subsubsection{SPDE ensemble and SPDE initial conditions \label{subsec:spde_initial_results}}

The SPDE \eqref{eq:stochastic2dEuler} is simulated on a coarse grid of size $64\times64$. An ensemble of $N_{p}$ initial conditions for the SPDE are generated, each of which leads to an independent
realisation of the SPDE. Motivated by data assimilation and in particular particle filtering vocabulary, we call each realisation in the ensemble a \emph{particle}. We denote the particles by $\hat{q}^{i},$ $i=1,\dots,N_{p}.$

The goal when generating the initial conditions is to obtain an ensemble which contain particles that are `close' to the truth. More precisely, we would like an initial prior distribution in which the truth lies in the concentration of the probability density. To achieve this, we take a truth $\omega_{t_0}$ (see Remark \ref{remark:pde_period_division}) and deform it using the following `modified' Euler equation:
\begin{equation}
\partial_{t}\omega_{t_0}+\beta_{i}\vecu(\tau_{i})\cdot\nabla\omega_{t_0}=0\label{eq:initial_cond_deformation}
\end{equation}
where $\beta_{i}\sim\mathcal{N}(0,\epsilon),$ $i=1,\dots,N_{p}$ are centered Gaussian weights with an apriori variance parameter $\epsilon,$ and $\tau_{i}\sim\mathcal{U}\left(t_{\text{initial}},t_0\right),$ $i=1,\dots,N_{p}$ are uniform random numbers. Thus each $\vecu\left(\tau_{i}\right)$ corresponds to a PDE solution in the time period $\left[t_{\text{initial}},t_0\right)$. Equation \eqref{eq:initial_cond_deformation} is solved for one or two ett to obtain a deformation $\hat{\omega}_{t_0}^{i}$ of the coarse grained initial condition $\omega_{t_0}$. These are then used as initial conditions for the SPDE realisations, i.e. $\hat{q}_{t_0}^{i}:=\hat{\omega}_{t_0}^{i},$ $i=1,\dots,N_{p}.$ 

\begin{remark}\label{remark:pde_period_division}
        At the start of Section \ref{sec:numerical_experiments} we discussed simulating the PDE on the fine grid. The PDE initial condition $\omega_{\text{initial}}$ was obtained after spin--up. The overall simulation time interval from the point of $\omega_{\text{initial}}$ is of length $146$ ett. Let $t_{\text{initial}}$ denote the time point that corresponds to $\omega_{\text{initial}}$. We divide the overall simulation time interval into two halves $\left[ t_{\text{initial}}, t_{0} \right)$  and $\left[t_0, 146\right].$ The PDE solutions in the period $\left[ t_{\text{initial}}, t_{0} \right)$ are used for generating the initial condition ensemble for the SPDE, see \eqref{eq:initial_cond_deformation}. This way, the velocity field used in \eqref{eq:initial_cond_deformation} is physical. $t_0$ is set as the initial time to start the SPDE uncertainty quantification numerical experiments. 
\end{remark}

\begin{figure}
        \begin{centering}
                \noindent\begin{minipage}[t]{1\textwidth}%
                        \begin{center}
                                \includegraphics[width=1\textwidth]{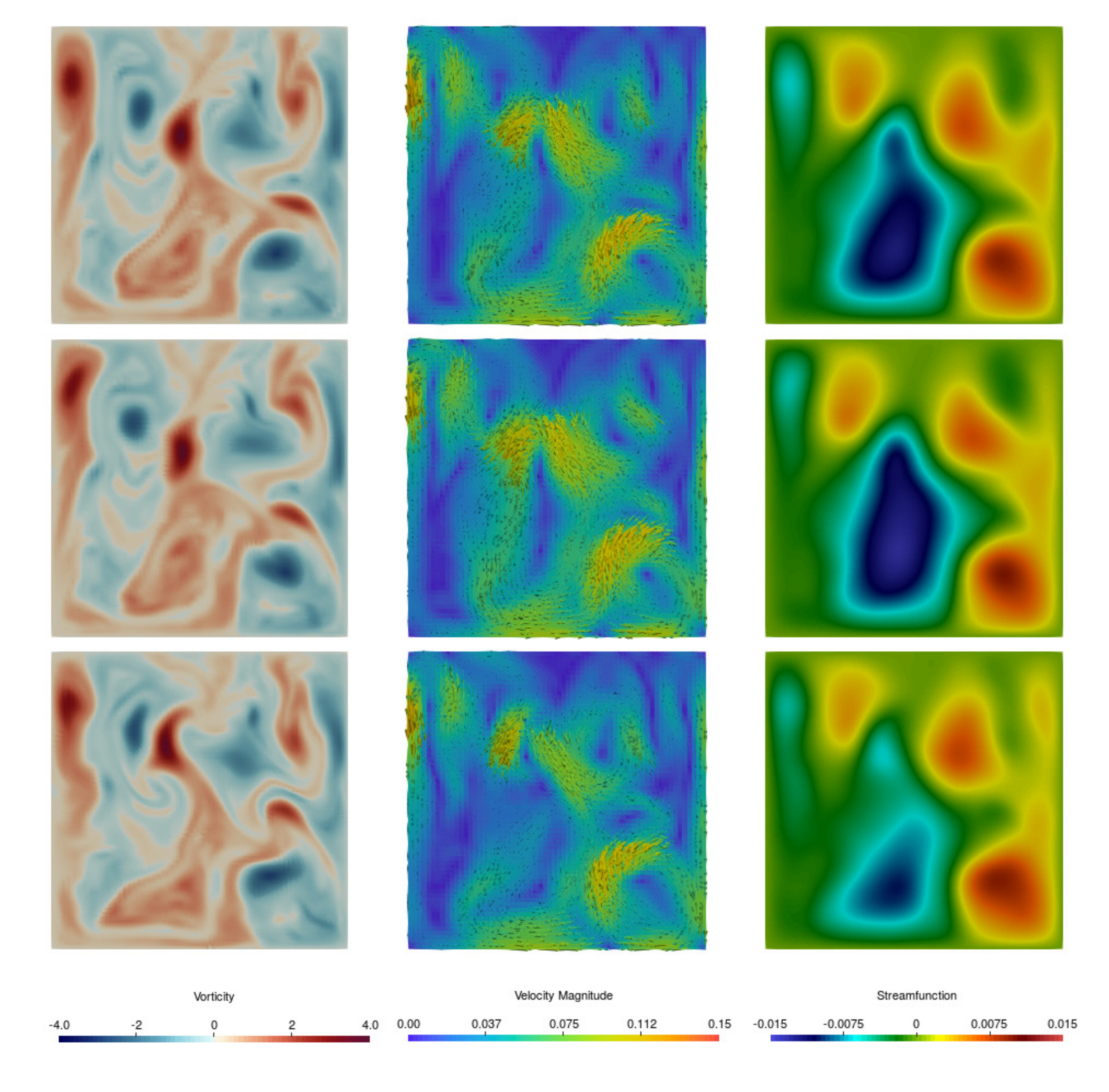}
                                \par\end{center}%
                \end{minipage}
                \par\end{centering}
        \centering{}\caption{\label{fig:spde_solution} The plots show vorticity (left column), velocity (middle column) and streamfunction (right column) for the coarse-grained fine solution (top row), which we call the \emph{truth}, and two independent realisations of the SPDE (middle and bottom rows) at time $t_{0}$, which we call \emph{particles}. The particles at this time point are generated using the deformation procedure \eqref{eq:initial_cond_deformation}. For the vorticity scalar field, the red and blue colours represent opposing signs of the function, and the colour shades indicate the different levels of magnitude of the function. For the velocity field, scaled arrow fields are plotted to indicate the direction and magnitude of the velocity vectors at each spatial location, and the colours highlight the magnitude of the velocity vectors. For the streamfunction scalar field, the colours indicate the contour lines of the function, along which the velocity vectors travel. The middle particle seems to be `closer' to the truth in terms of the visible large scale features than the bottom particle, see \eqref{eq:initial_cond_deformation}. See Section \ref{subsec:spde_initial_results}.}
\end{figure}

\begin{figure}
        \begin{centering}
                \noindent\begin{minipage}[t]{1\textwidth}%
                        \begin{center}
                                \includegraphics[width=1\textwidth]{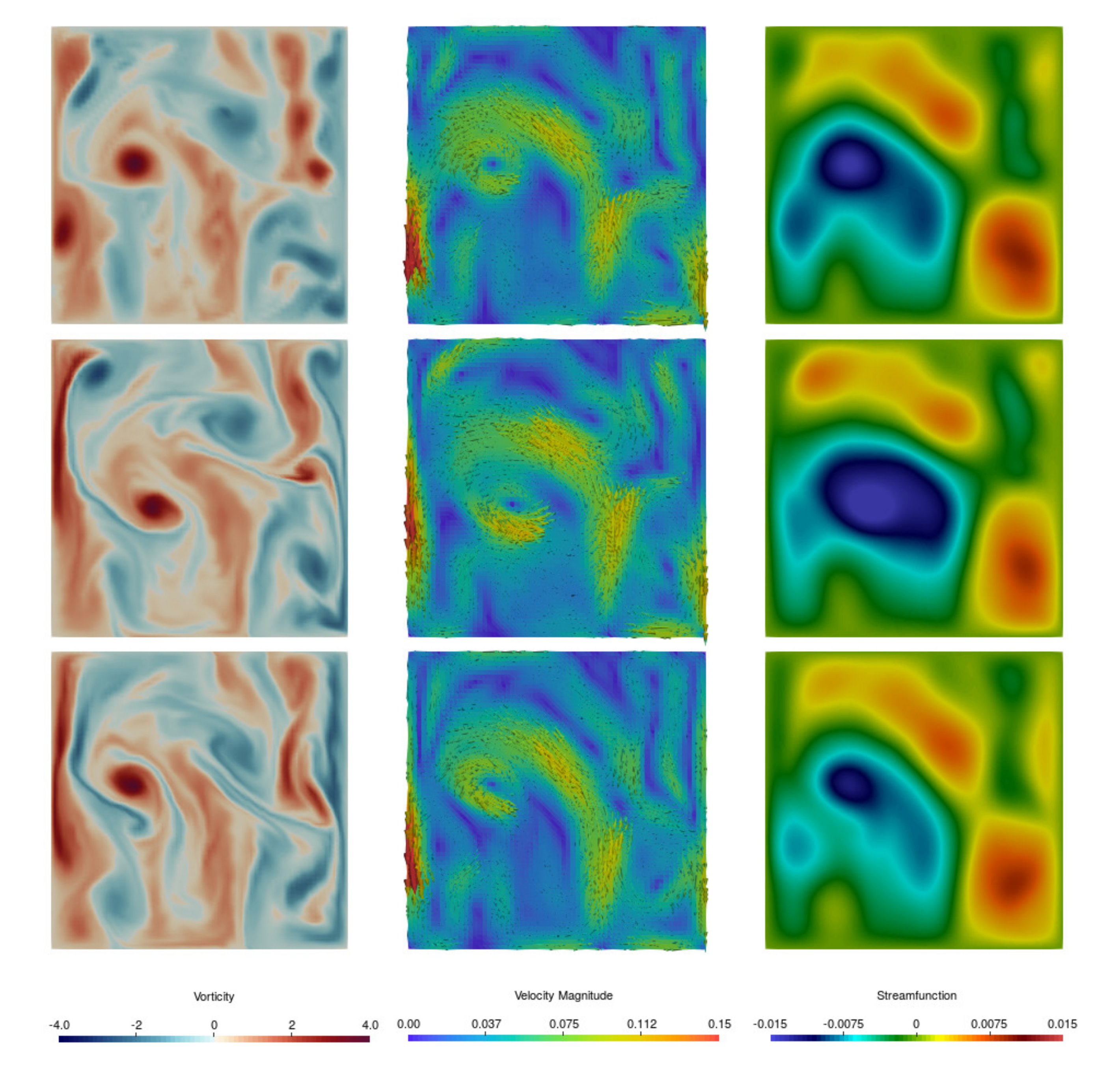}
                                \par\end{center}%
                \end{minipage}
                \par\end{centering}
        \centering{}\caption{\label{fig:spde_solution-1} Following on from Figure \ref{fig:spde_solution}, the plots show vorticity (left column), velocity (middle column) and streamfunction (right column) for the truth (top row), and the two particles (middle and bottom rows) at time $t_{0} + 3$ ett. For the vorticity scalar field, the red and blue colours represent opposing signs of the function, and the colour shades indicate the different levels of magnitude of the function. For the velocity field, scaled arrow fields are plotted to indicate the direction and magnitude of the velocity vectors at each spatial location, and the colours highlight the magnitude of the velocity vectors. For the streamfunction scalar field, the colours indicate the contour lines of the function, along which the velocity vectors travel. The middle particle started  `closer' to the truth in terms of the visible large scale features at time $t_{0}$ than the bottom particle, see Figure \ref{fig:spde_solution}, however as can be seen, different large and small scale features to the truth have developed. Comparing the streamfunction plots, the bottom particle seems to have diverged from the truth even further. See Section \ref{subsec:spde_initial_results}.}
\end{figure}

\begin{figure}
        \begin{centering}
                \noindent\begin{minipage}[t]{1\textwidth}%
                        \begin{center}
                                \includegraphics[width=1\textwidth]{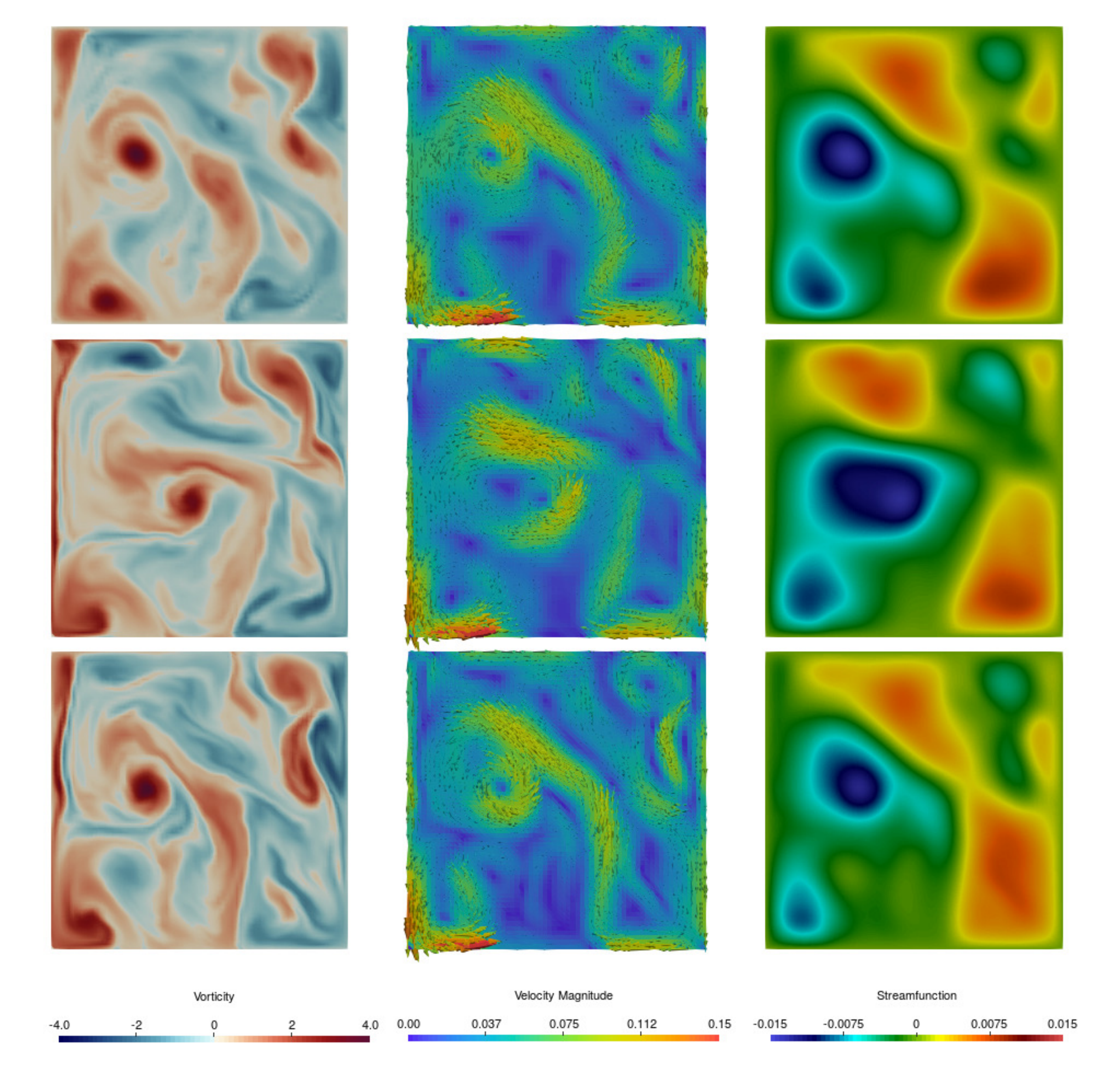}
                                \par\end{center}%
                \end{minipage}
                \par\end{centering}
        \centering{}\caption{\label{fig:spde_solution-2} Following on from Figure \ref{fig:spde_solution-1}, the plots show vorticity (left column), velocity (middle column) and streamfunction (right column) for the truth (top row), and the two particles (middle and bottom rows) at time $t_{0} + 5$ ett. For the vorticity scalar field, the red and blue colours represent opposing signs of the function, and the colour shades indicate the different levels of magnitude of the function. For the velocity field, scaled arrow fields are plotted to indicate the direction and magnitude of the velocity vectors at each spatial location, and the colours highlight the magnitude of the velocity vectors. For the streamfunction scalar field, the colours indicate the contour lines of the function, along which the velocity vectors travel. The middle particle started  `closer' to the truth in terms of the visible large scale features at time $t_{0}$ than the bottom particle, see Figure \ref{fig:spde_solution}, however as can be seen, different large and small scale features to the truth have developed. Comparing the streamfunction plots, the bottom particle seems to have diverged from the truth even further. See Section \ref{subsec:spde_initial_results}.}
\end{figure}

Figures \ref{fig:spde_solution}, \ref{fig:spde_solution-1} and \ref{fig:spde_solution-2} show plots for the truth and two independent realisations of the SPDE at times $t_0$, $t_0 + 3$ ett and $t_0 + 5$ ett, respectively. In each figure, the plots show vorticity (left column), velocity (middle column) and streamfunction (right column) for the \emph{truth} (top row), and the two \emph{particles} (middle and bottom rows).

In Figure \ref{fig:spde_solution} the particles are generated using the deformation procedure \eqref{eq:initial_cond_deformation}. Intuitively, deformations of the truth hope to capture the idea of ``location uncertainty'' in the initial conditions. 
The middle particle seems to be `closer' to the truth in terms of the visible large scale features than the bottom particle.

At time $t_0 + 3$ ett, which is shown in Figure \ref{fig:spde_solution-1}, although the middle particle started  `closer' to the truth than the bottom particle at time $t_0$, different large and small scale features to the truth have developed. Comparing the streamfunction plots, the bottom particle seems to have diverged from the truth even further. Seen at $t_0 + 5$ ett, which is shown in Figure \ref{fig:spde_solution-2}, the diverging features in the two particles develop further, representing increasing uncertainty as time goes on. Data assimilation techniques can be applied to incorporate observation data in order to correct for the increasing uncertainty.

\subsubsection{SPDE uncertainty quantification \label{subsec:spde_uq_results}}

In this section we show uncertainty quantification test results which test
our stochastic parameterisation for the Euler equation \eqref{eq:2DEulerVorticity}. 

Each particle in the initial ensemble generated using the deformation procedure \eqref{eq:initial_cond_deformation} is { independently} evolved forward using the SPDE \eqref{eq:stochastic2dEuler} for a time length of $20$ ett. Two main sets of tests are done: uncertainty quantification and distance between the ensemble and the truth.

For uncertainty quantification, consider a uniform grid of size $4\times4$ (see Remark \ref{remark:observationgrid}). At each of its interior points, we plot the ensemble one standard deviation region about the ensemble mean and compare that with the truth at the same location. We also look at how the one standard deviation region is affected by changing the number of EOFs used and the number of particles in the ensemble. The tests are done for vorticity, streamfunction and velocity separately. These results are shown in Figures \ref{fig:mike-cullen-psi}---\ref{fig:mike-cullen-u}. In each plot, the solid line represents the truth and the coloured regions represent the one standard deviation regions. The results are plotted for discrete ett time values and are linearly interpolated in between times. 

The solid lines in all the plots start within their respective spreads, see the deformation procedure \eqref{eq:initial_cond_deformation}. As can be seen, the spreads capture the solid lines for roughly $4$ or $5$ ett before deviating at certain grid locations, for example see Figure \ref{fig:mike-cullen-psi}. The streamfunction is the smoothest of the three functions, thus in Figure \ref{fig:mike-cullen-psi} the solid lines and the one standard deviation regions show smoother features compared to those in the other figures. For example compare Figure \ref{fig:mike-cullen-psi} with Figure \ref{fig:mike-cullen-q}.

\begin{remark} \label{remark:observationgrid}
        For data assimilation, we consider a observation grid of size $4\times4$ and thus would like the parameterisation methodology to work well at grid points which correspond to the observation grid.
\end{remark}

For a fixed number of particles in the ensemble, $N_p=500$, we look at the differences between the one standard deviation regions due to using different number of EOFs: $n_{\xi}\equiv0.5$ versus $n_{\xi}\equiv0.9$. The results are shown in the left hand side plots in each of the Figures \ref{fig:mike-cullen-psi}---\ref{fig:mike-cullen-u}, where the pink regions correspond to $n_{\xi}\equiv0.9$ and the grey regions correspond to $n_{\xi}\equiv0.5$. We expect differences in spread size and location but the differences shown in the plots are insignificant.

With $n_{\xi}\equiv0.5$ fixed, we compare the differences between the one standard deviation regions due to changing the number of particles in the ensemble: $N_p = 500$ versus $N_p = 225$. The results are shown in the right hand plots in each of the Figures  \ref{fig:mike-cullen-psi}---\ref{fig:mike-cullen-u}, where the pink regions correspond to $N_p = 225$ and the grey regions correspond to $N_p=500$. Again the differences are insignificant.

\begin{figure}
        \begin{minipage}[t]{0.49\textwidth}%
                \begin{center}
                        \includegraphics[width=1\textwidth]{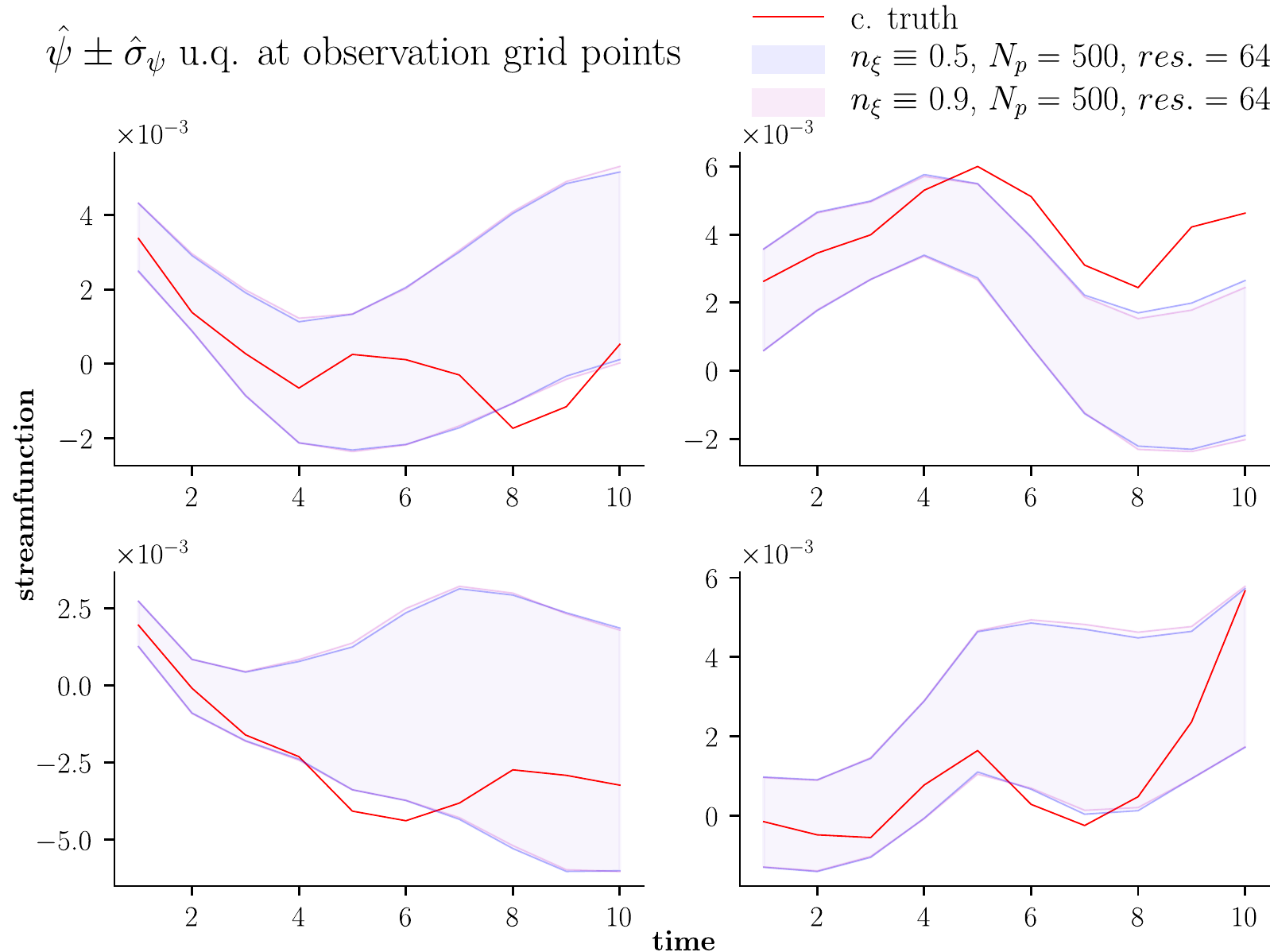}
                        \par\end{center}%
        \end{minipage}\hfill{}%
        \begin{minipage}[t]{0.49\textwidth}%
                \begin{center}
                        \includegraphics[width=1\textwidth]{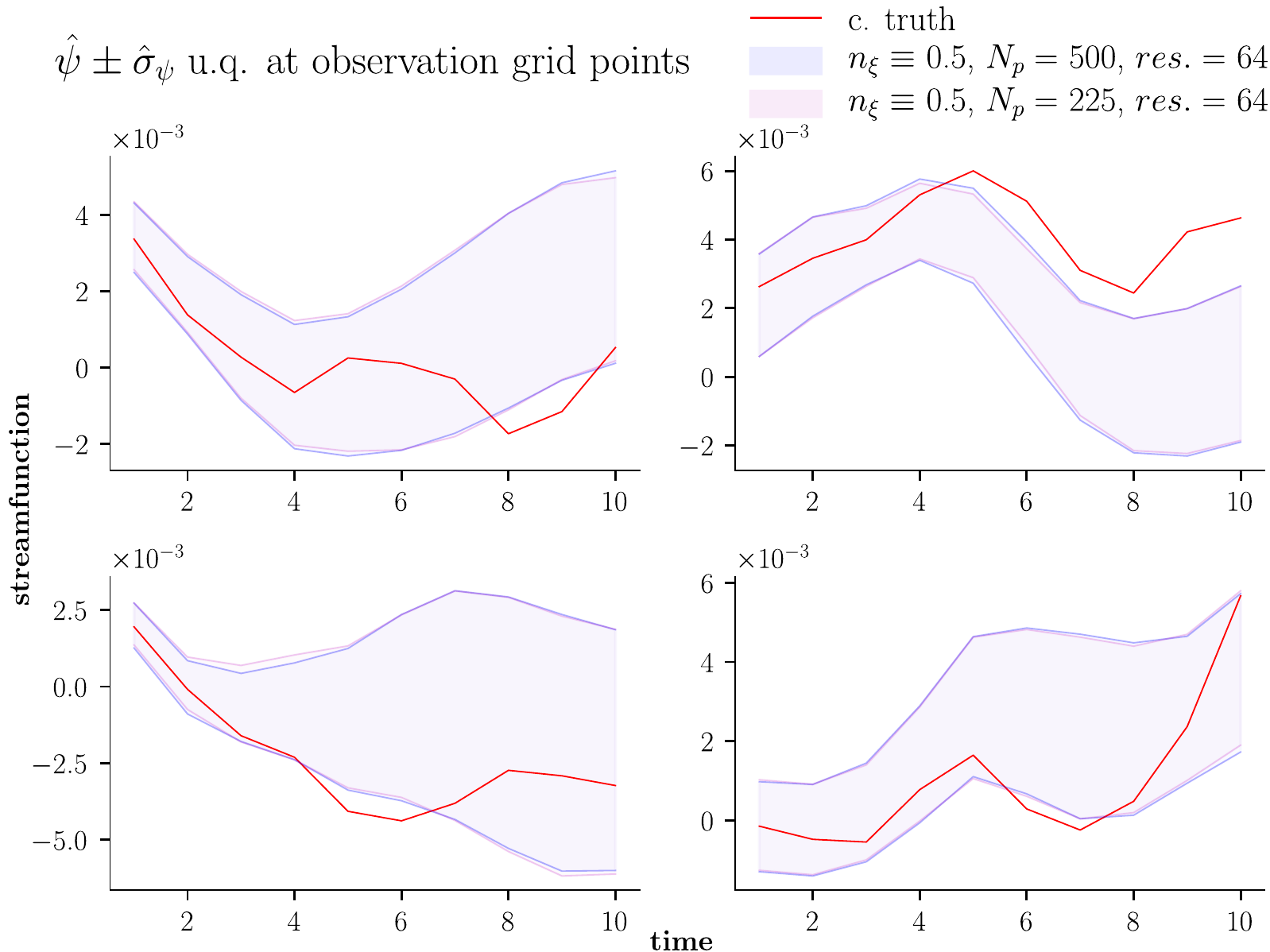}
                        \par\end{center}%
        \end{minipage}
        
        \caption{\label{fig:mike-cullen-psi} Uncertainty quantification plots comparing the truth with the ensemble one standard deviation region about the ensemble mean for the streamfunction at four interior grid points of a $4\times4$ observation grid. In each plot, the solid line represents the truth and the coloured regions represent the one standard deviation regions. In the figure on the left, for a fixed ensemble size ($N_p = 500$), we compare the spread differences at the individual observation grid points due to using a different number of EOFs: $n_{\xi}\equiv0.9$ (pink) versus $n_{\xi}\equiv0.5$ (grey). In the figure on the right, for a fixed the number of EOFs ($n_{\xi}\equiv0.5$), we compare the spreads differences at individual grid points due having a different number of particles in the ensemble: $N_p=500$ (grey) versus $N_p=225$ (pink). The results are plotted for discrete ett time values and are linearly interpolated in between times. The solid lines in all the plots start within their respective spreads, see \eqref{eq:initial_cond_deformation}. As can be seen, the spreads capture the solid lines for roughly $4$ or $5$ ett before deviating at certain grid locations. We expect differences in spread size and location but the differences shown in the plots are insignificant. See Section \ref{subsec:spde_uq_results}.}
\end{figure}

\begin{figure}
        \begin{minipage}[t]{0.49\textwidth}%
                \begin{center}
                        \includegraphics[width=1\textwidth]{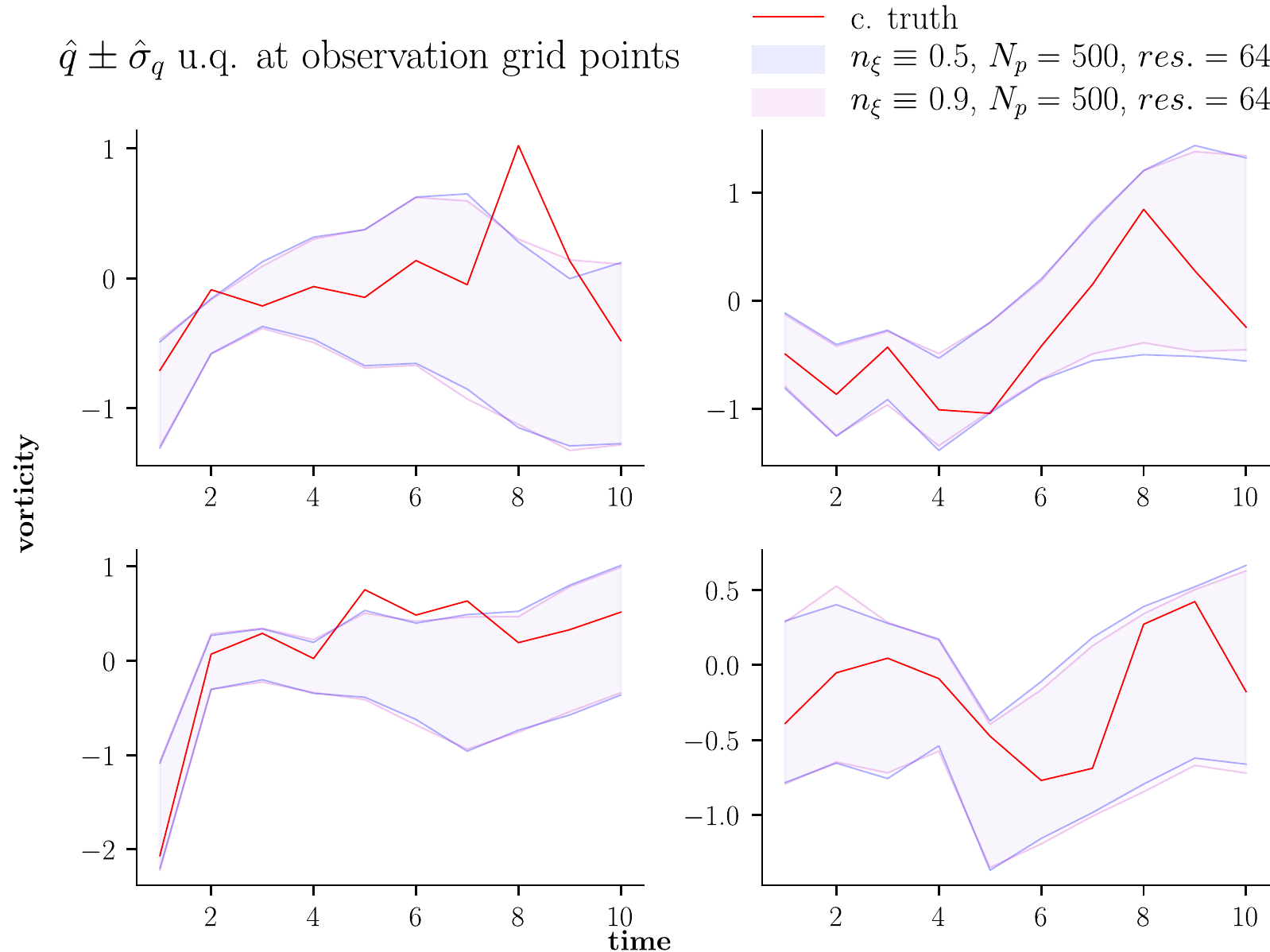}
                        \par\end{center}%
        \end{minipage}\hfill{}%
        \begin{minipage}[t]{0.49\textwidth}%
                \begin{center}
                        \includegraphics[width=1\textwidth]{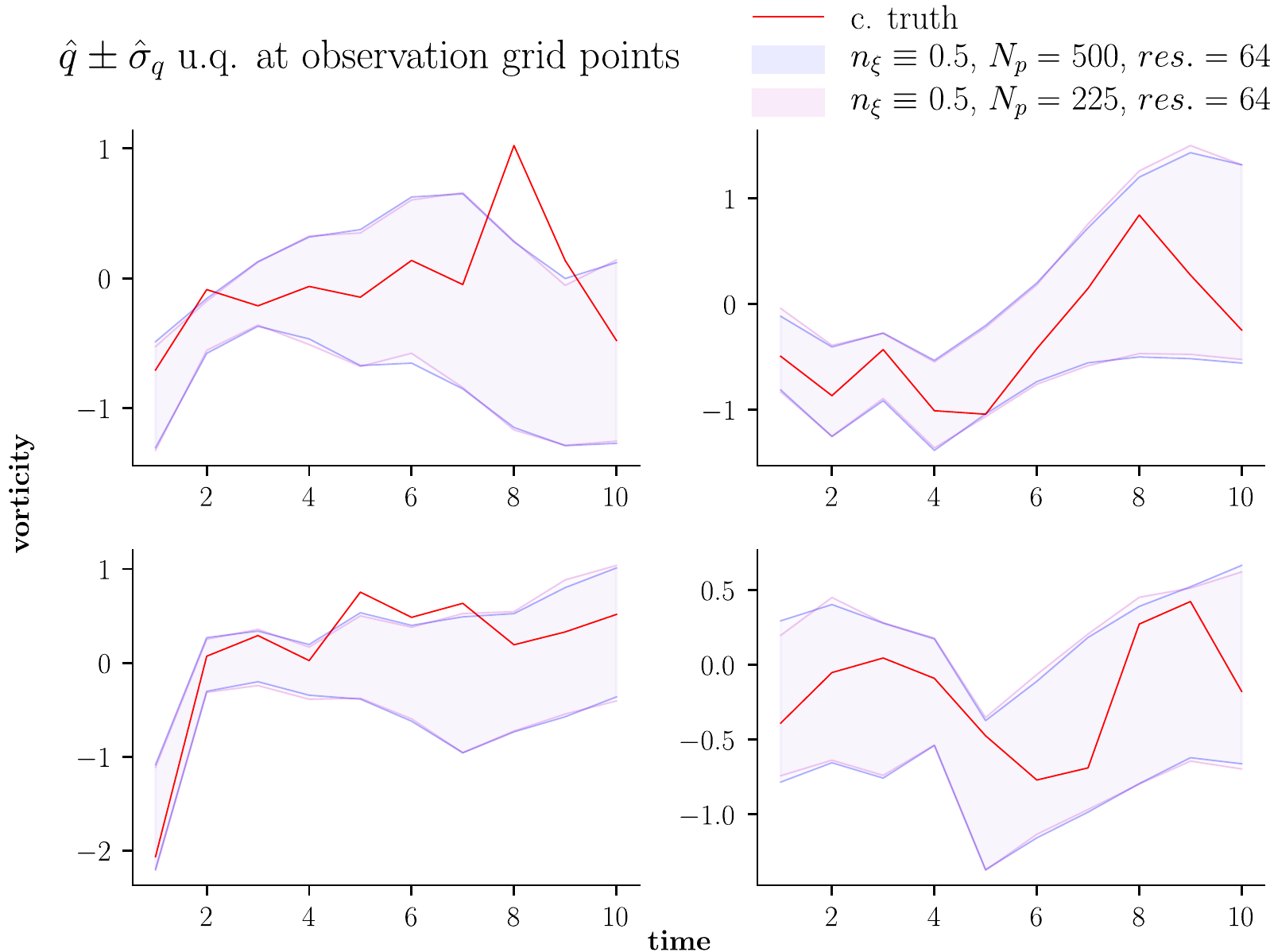}
                        \par\end{center}%
        \end{minipage}
        
        \caption{\label{fig:mike-cullen-q}Uncertainty quantification plots comparing the truth with the ensemble one standard deviation region about the ensemble mean for the vorticity at four interior grid points of a $4\times4$ observation grid. In each plot, the solid line represents the truth and the coloured regions represent the one standard deviation regions. In the figure on the left, for a fixed ensemble size ($N_p = 500$), we compare the spread differences at the individual observation grid points due to using a different number of EOFs: $n_{\xi}\equiv0.9$ (pink) versus $n_{\xi}\equiv0.5$ (grey). In the figure on the right, for a fixed the number of EOFs ($n_{\xi}\equiv0.5$), we compare the spreads differences at individual grid points due having a different number of particles in the ensemble: $N_p=500$ (grey) versus $N_p=225$ (pink). The results are plotted for discrete ett time values and are linearly interpolated in between times. The solid lines in all the plots start within their respective spreads, see \eqref{eq:initial_cond_deformation}. As can be seen, the spreads capture the solid lines for roughly $4$ or $5$ ett before deviating at certain grid locations. We expect differences in spread size and location but the differences shown in the plots are insignificant. See Section \ref{subsec:spde_uq_results}.}
\end{figure}

\begin{figure}
        \begin{minipage}[t]{0.49\textwidth}%
                \begin{center}
                        \includegraphics[width=1\textwidth]{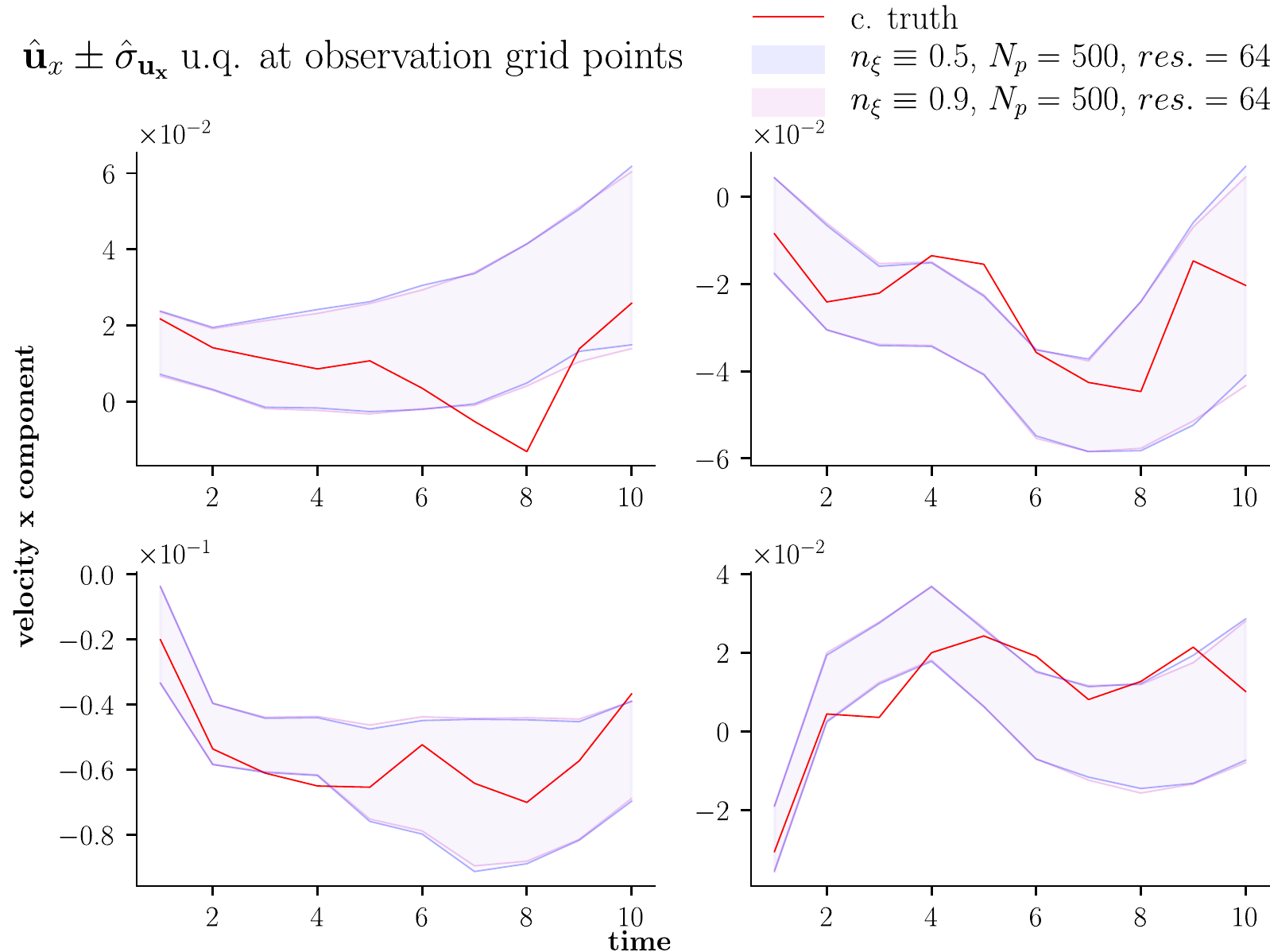}
                        \par\end{center}%
        \end{minipage}\hfill{}%
        \begin{minipage}[t]{0.49\textwidth}%
                \begin{center}
                        \includegraphics[width=1\textwidth]{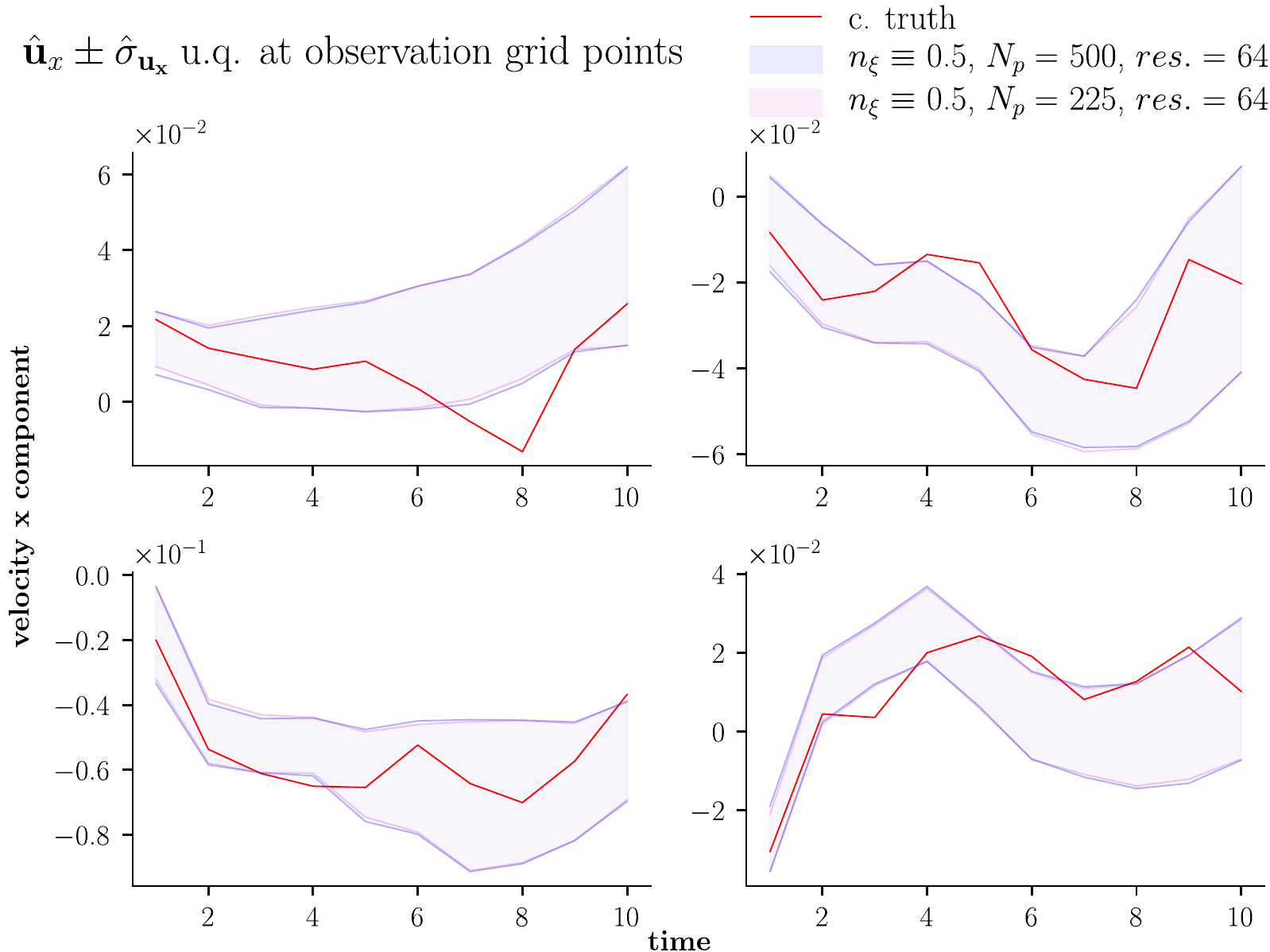}
                        \par\end{center}%
        \end{minipage}
        
        \smallskip{}
        \begin{minipage}[t]{0.49\textwidth}%
                \begin{center}
                        \includegraphics[width=1\textwidth]{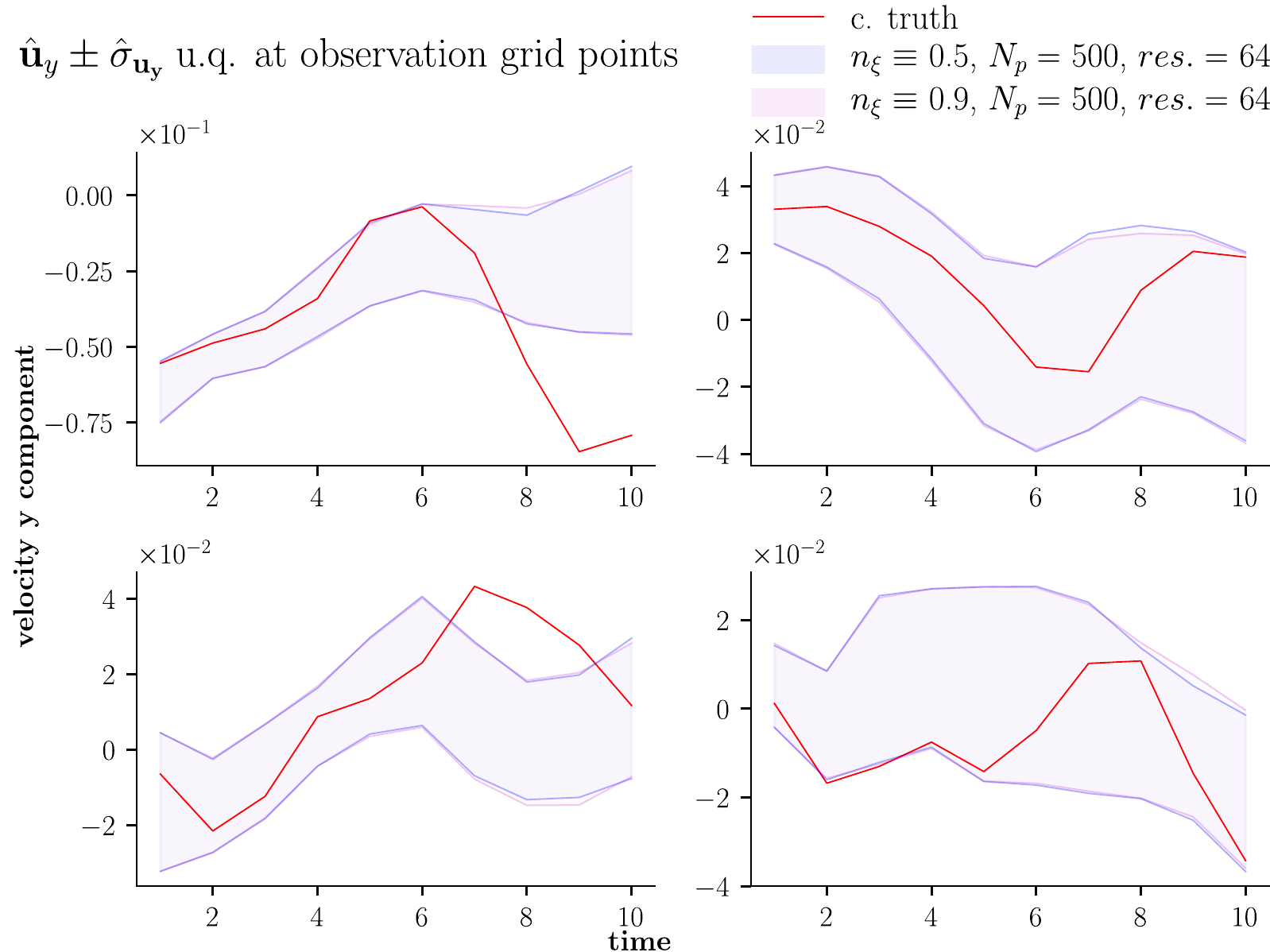}
                        \par\end{center}%
        \end{minipage}\hfill{}%
        \begin{minipage}[t]{0.49\textwidth}%
                \begin{center}
                        \includegraphics[width=1\textwidth]{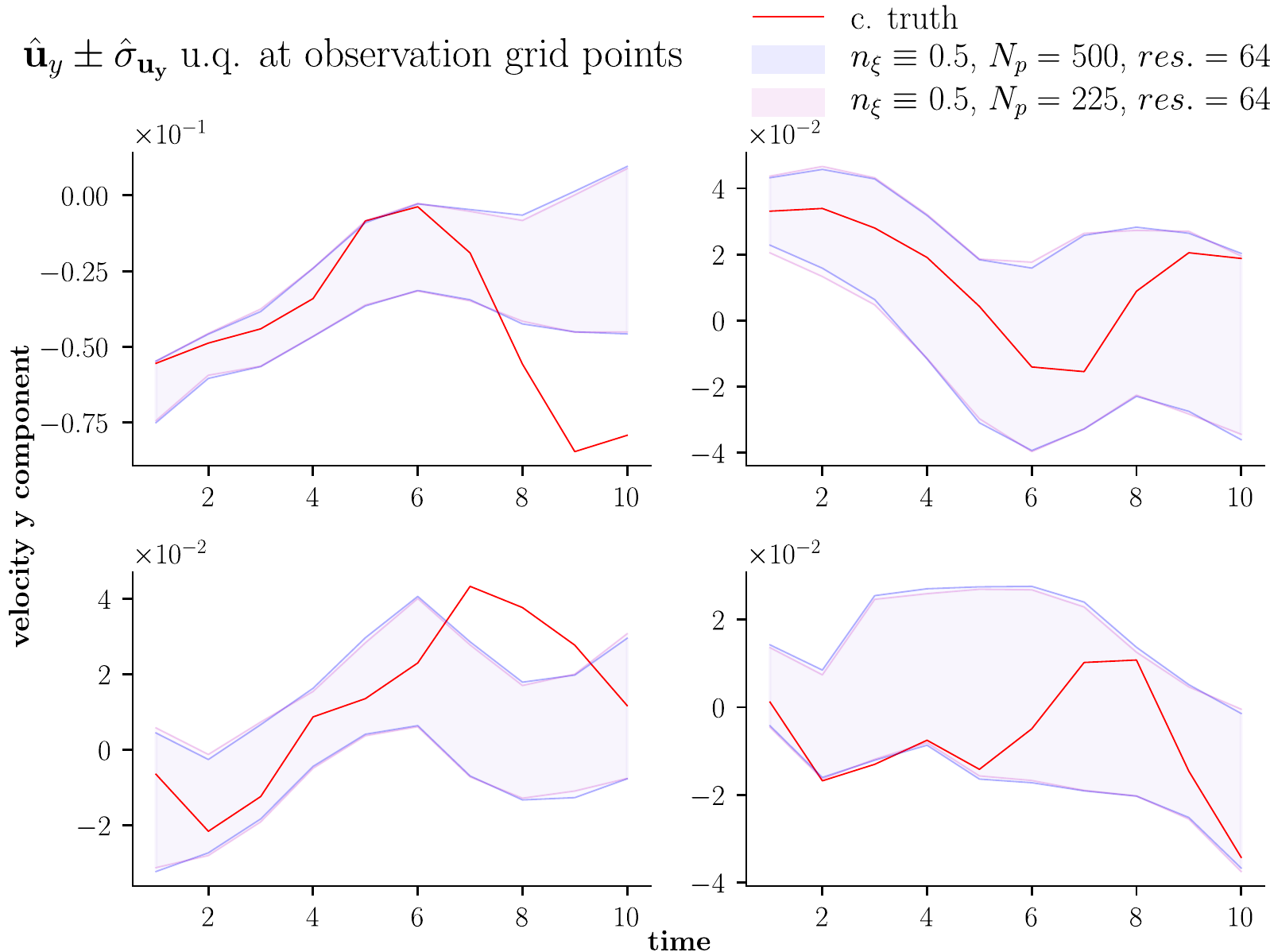}
                        \par\end{center}%
        \end{minipage}
        
        \caption{\label{fig:mike-cullen-u} Uncertainty quantification plots comparing the truth with the ensemble one standard deviation region about the ensemble mean for the velocity field (top two figures for the $x$-component, bottom two figures for the $y$-component) at four interior grid points of a $4\times4$ observation grid. In each plot, the solid line represents the truth and the coloured regions represent the one standard deviation regions. In the figures on the left, for a fixed ensemble size ($N_p = 500$), we compare the spread differences at the individual observation grid points due to using a different number of EOFs: $n_{\xi}\equiv0.9$ (pink) versus $n_{\xi}\equiv0.5$ (grey). In the figures on the right, for a fixed the number of EOFs ($n_{\xi}\equiv0.5$), we compare the spreads differences at individual grid points due having a different number of particles in the ensemble: $N_p=500$ (grey) versus $N_p=225$ (pink). The results are plotted for discrete ett time values and are linearly interpolated in between times. The solid lines in all the plots start within their respective spreads, see \eqref{eq:initial_cond_deformation}. As can be seen, the spreads capture the solid lines for roughly $4$ or $5$ ett before deviating at certain grid locations. We expect differences in spread size and location but the differences shown in the plots are insignificant. See Section \ref{subsec:spde_uq_results}.}
\end{figure}


The results shown in Figures \ref{fig:mike-cullen-psi}---\ref{fig:mike-cullen-u} correspond to the truth and particles defined on the $64\times64$ coarse grid. The uncertainty quantification tests are repeated for more refined coarse grids ($128\times128$ and $256\times256$) with the number of EOFs and number of particles fixed ($n_{\xi}\equiv0.9$, $N_p=225$), to investigate the effect of mesh grid size on uncertainty quantification. The results are shown in Figure \ref{fig:mike-cullen-psi-multires} for streamfunction, Figure \ref{fig:mike-cullen-q-multires} for vorticity and \ref{fig:mike-cullen-u-multires} for velocity. The one standard deviation regions for the different grid sizes are plotted together to compare their differences. 
Note that as we make spatial refinements the coarse grained truth also changes, because the coarse graining procedure used to obtain the truth depends on the underlying coarse grid resolution.

In all plots for the multi-resolution analysis, the results show that the ensemble one standard deviation region 'converge' toward the respective coarse grained truth, and capture the truth for longer period of time as the grids get refined. In this sense, our parameterisation methodology is consistent under grid refinement. 

\begin{figure}
        \begin{minipage}[t]{0.49\textwidth}%
                \begin{center}
                        \includegraphics[width=1\textwidth]{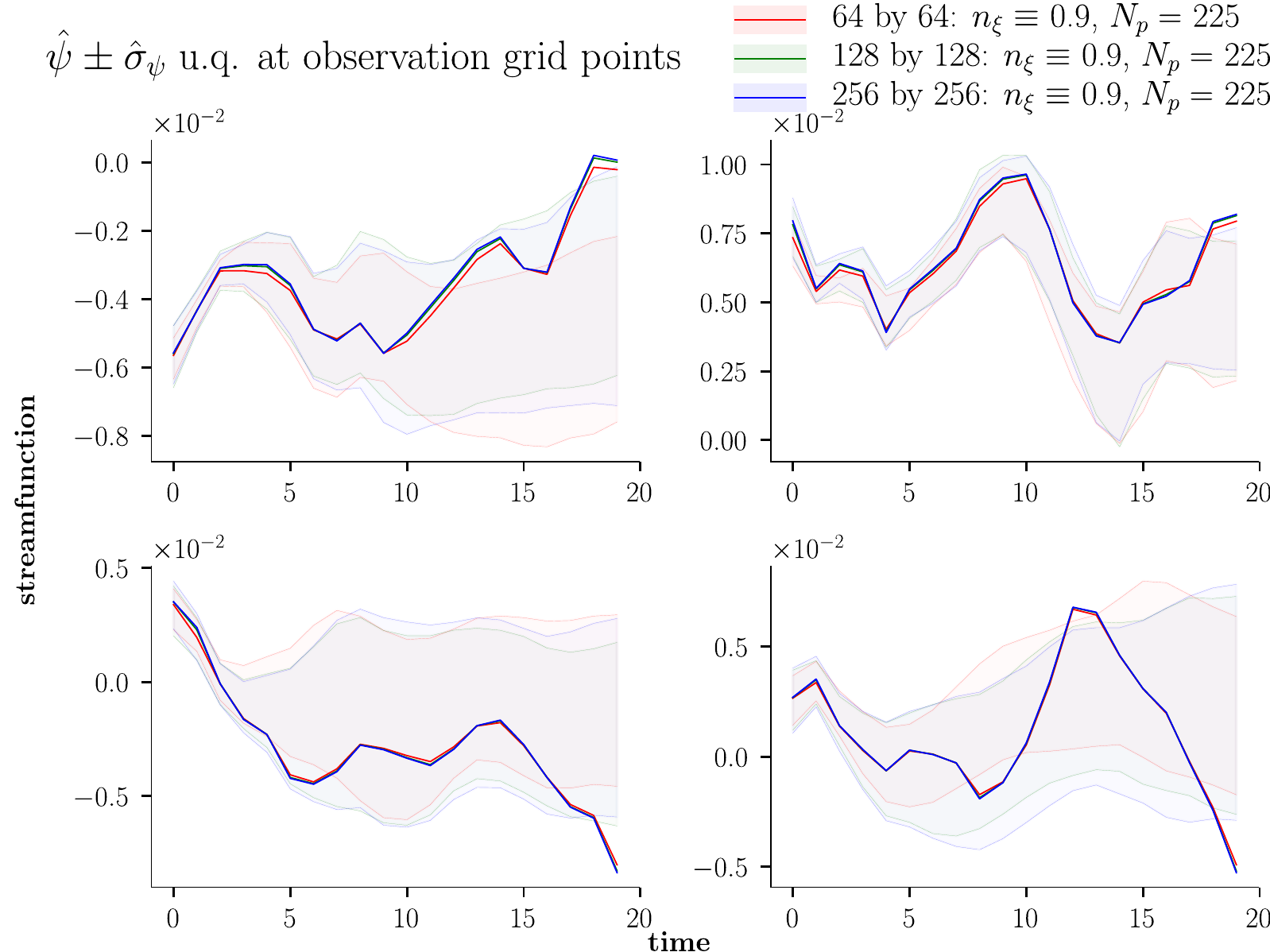}
                        \par\end{center}%
        \end{minipage}\hfill{}%
        \begin{minipage}[t]{0.49\textwidth}%
                \begin{center}
                        \includegraphics[width=1\textwidth]{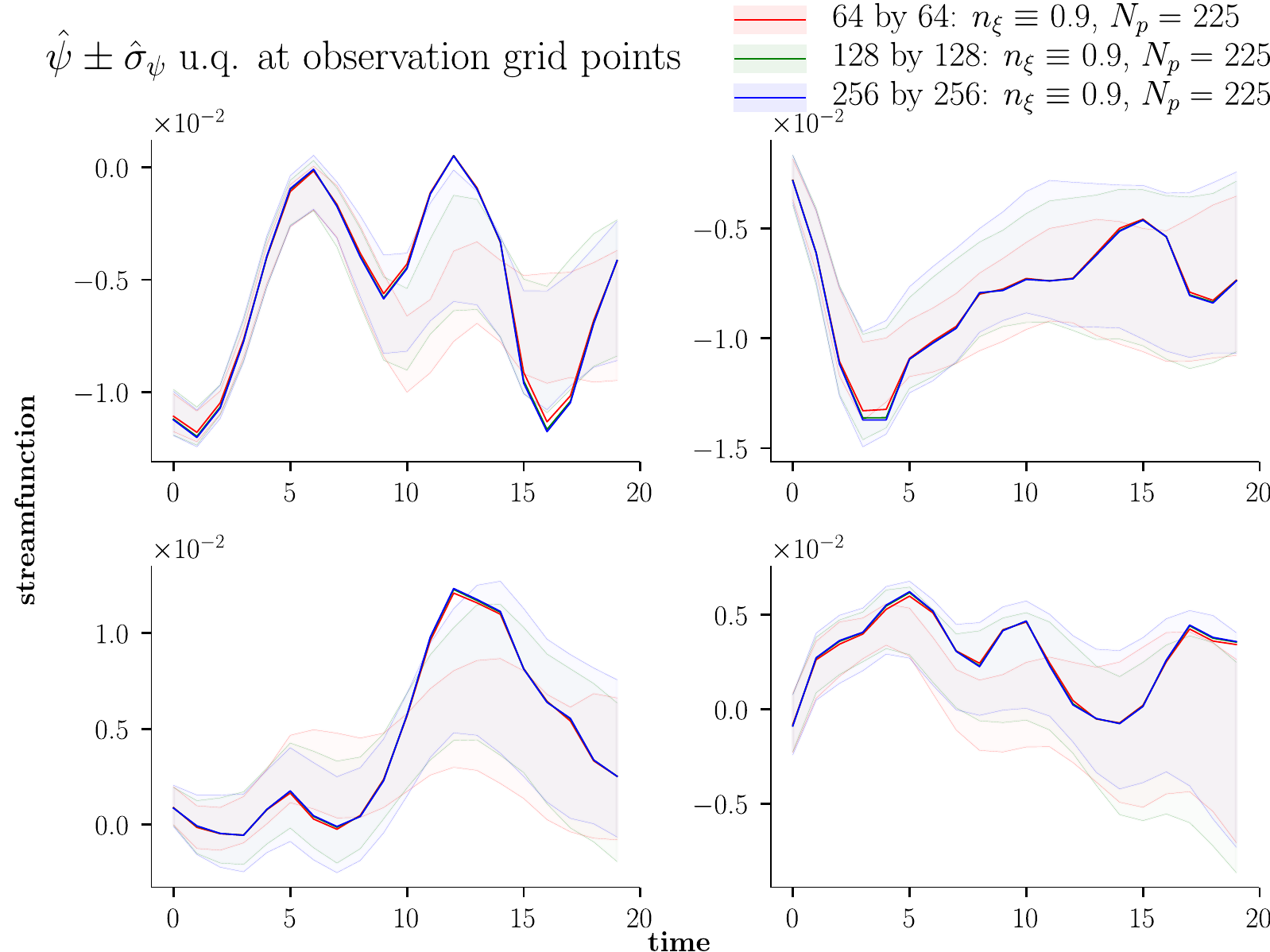}
                        \par\end{center}%
        \end{minipage}
        
        \caption{\label{fig:mike-cullen-psi-multires} Uncertainty quantification comparing the truth with the ensemble one standard deviation region about the ensemble mean for the streamfunction. The left and right hand figures each contain four plots. Each plot corresponds to a fixed grid point on a observation grid of size $4\times4$. For a fixed  number of EOFs $n_{\xi}\equiv0.9$ and a fixed number of particles in the ensemble $N_p=225$, each plot shows the truths and spreads corresponding to three coarse resolutions:  $64\times64,$ $128\times128$ and $256\times256.$. The solid lines represent the truth and the coloured regions represent the one standard deviation regions. Recall that the truth depends on the coarse grid size, see Figure \ref{fig:pde_solution_t0}. The red line and spread correspond to the $64\times64$ coarse grid. The green line and spread correspond to the $128\times128$ coarse grid. The blue line and spread correspond to the $256\times256$ coarse grid. The results are plotted for discrete ett time values and are linearly interpolated in between times. We see that as the coarse grid resolution gets refined, the one standard deviation region stays closer to the truth for longer time periods. This confirms that the parameterisation methodology is consistent under grid refinement. See Section \ref{subsec:Methodology} and \ref{subsec:spde_uq_results}}.
\end{figure}

\begin{figure}
        \begin{minipage}[t]{0.49\textwidth}%
                \begin{center}
                        \includegraphics[width=1\textwidth]{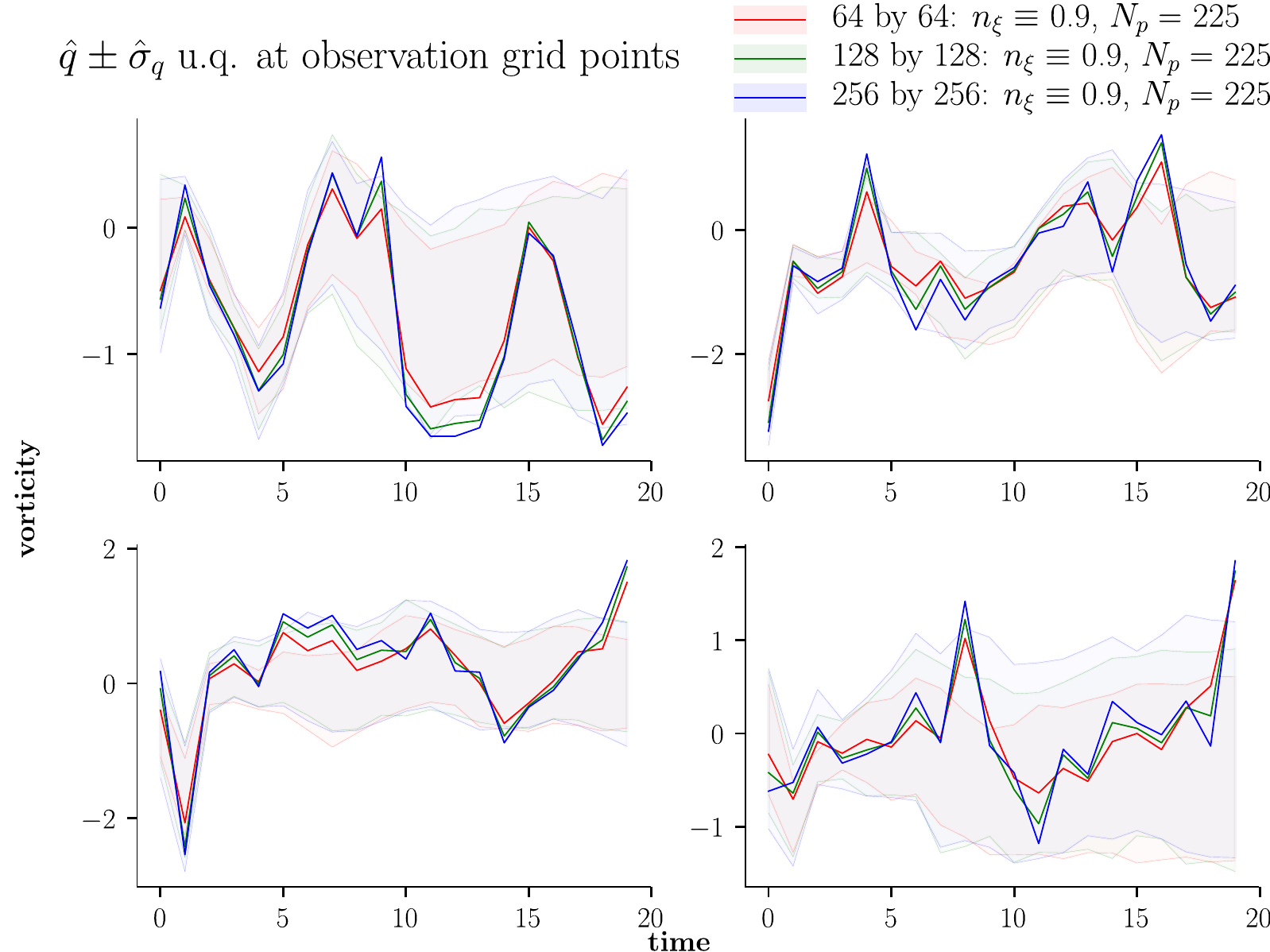}
                        \par\end{center}%
        \end{minipage}\hfill{}%
        \begin{minipage}[t]{0.49\textwidth}%
                \begin{center}
                        \includegraphics[width=1\textwidth]{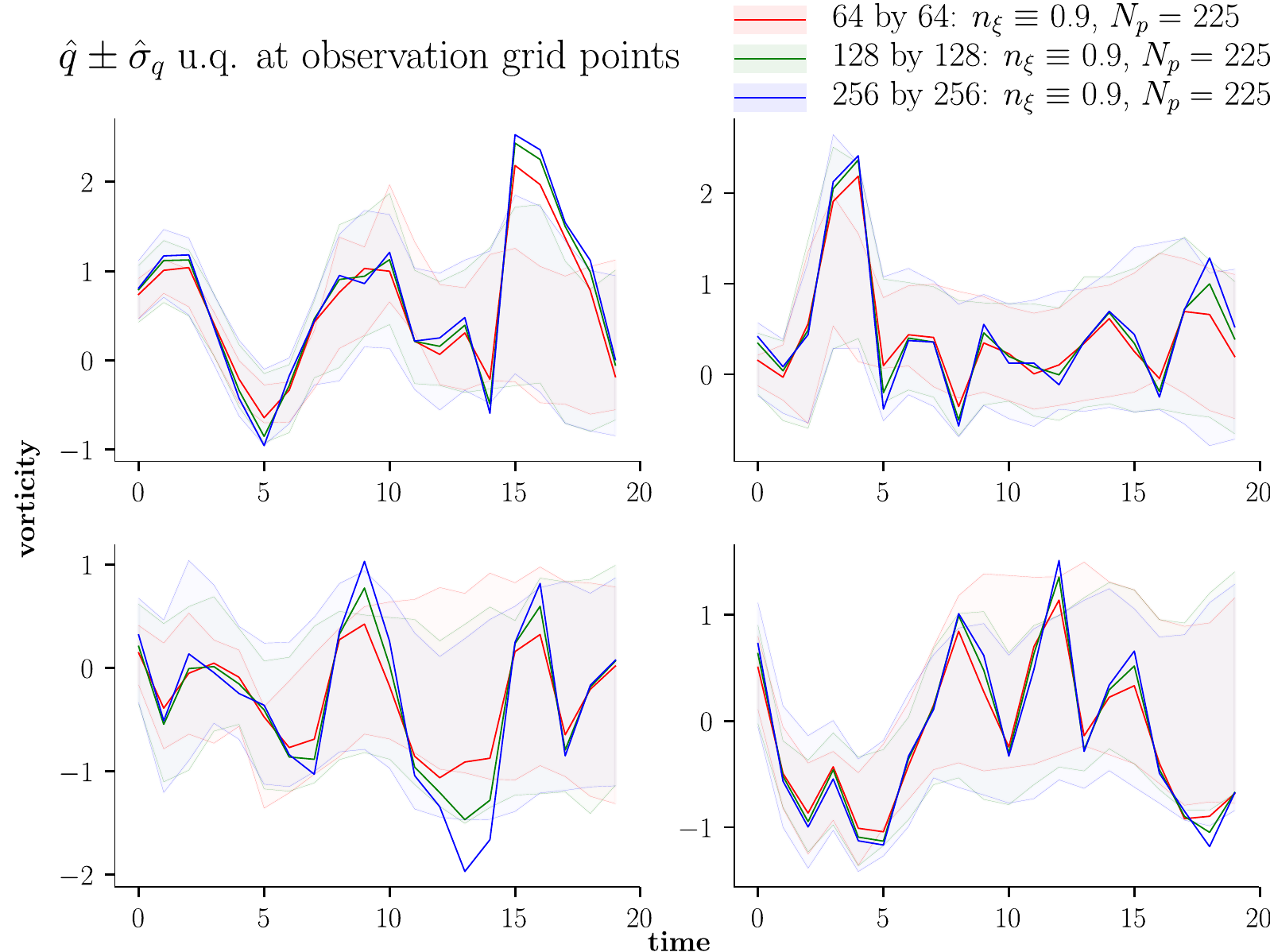}
                        \par\end{center}%
        \end{minipage}
        
        \caption{\label{fig:mike-cullen-q-multires}
                Uncertainty quantification comparing the truth with the ensemble one standard deviation region about the ensemble mean for the vorticity. The left and right hand figures each contain four plots. Each plot corresponds to a fixed grid point on a observation grid of size $4\times4$. For a fixed  number of EOFs $n_{\xi}\equiv0.9$ and a fixed number of particles in the ensemble $N_p=225$, each plot shows the truths and spreads corresponding to three coarse resolutions:  $64\times64,$ $128\times128$ and $256\times256.$. The solid lines represent the truth and the coloured regions represent the one standard deviation regions. Recall that the truth depends on the coarse grid size, see Figure \ref{fig:pde_solution_t0}. The red line and spread correspond to the $64\times64$ coarse grid. The green line and spread correspond to the $128\times128$ coarse grid. The blue line and spread correspond to the $256\times256$ coarse grid. The results are plotted for discrete ett time values and are linearly interpolated in between times. We see that as the coarse grid resolution gets refined, the one standard deviation region stays closer to the truth for longer time periods. This confirms that the parameterisation methodology is consistent under refinement. See Section \ref{subsec:Methodology} and \ref{subsec:spde_uq_results}.
        }
\end{figure}

\begin{figure}
        \begin{minipage}[t]{0.49\textwidth}%
                \begin{center}
                        \includegraphics[width=1\textwidth]{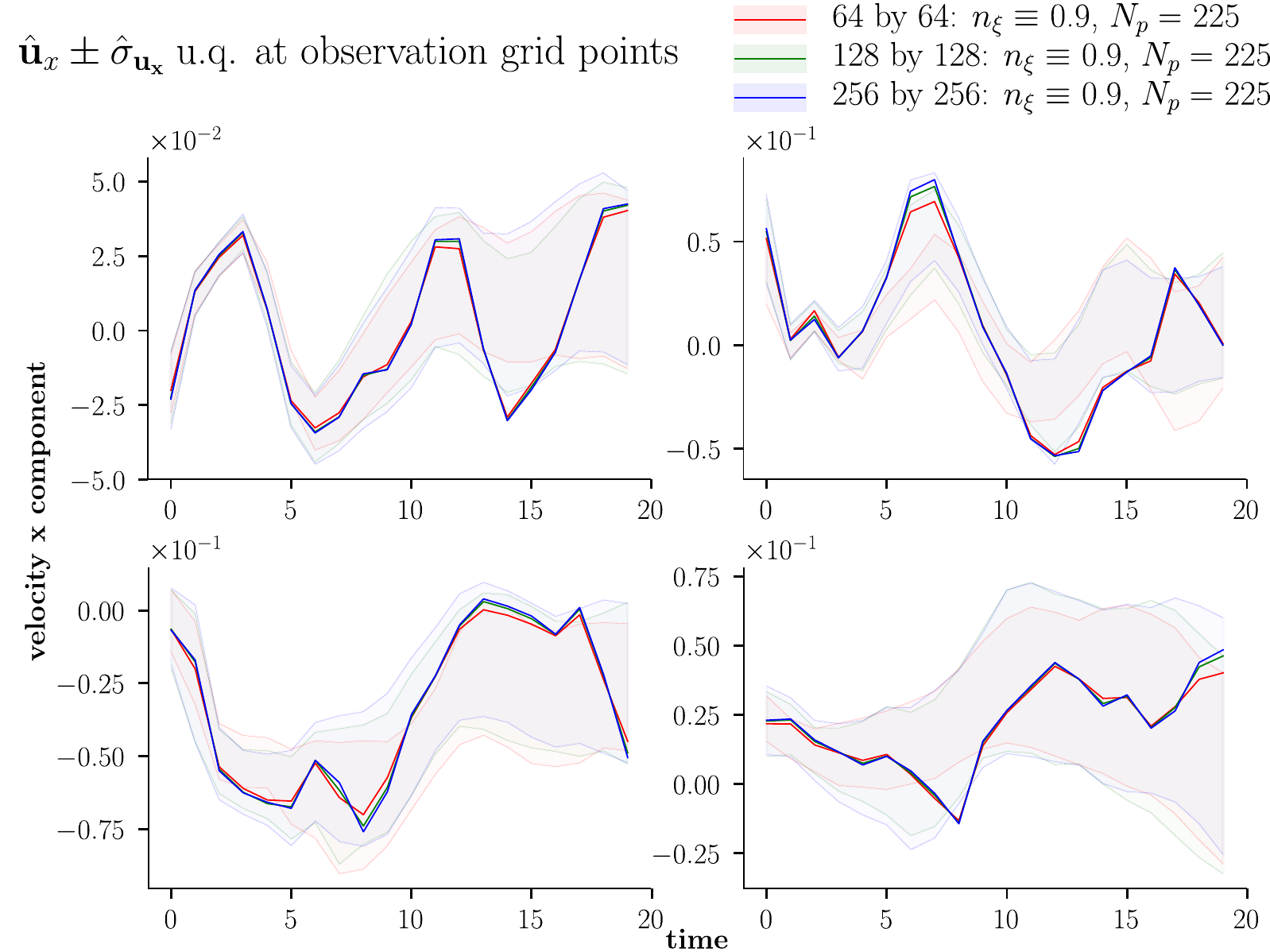}
                        \par\end{center}%
        \end{minipage}\hfill{}%
        \begin{minipage}[t]{0.49\textwidth}%
                \begin{center}
                        \includegraphics[width=1\textwidth]{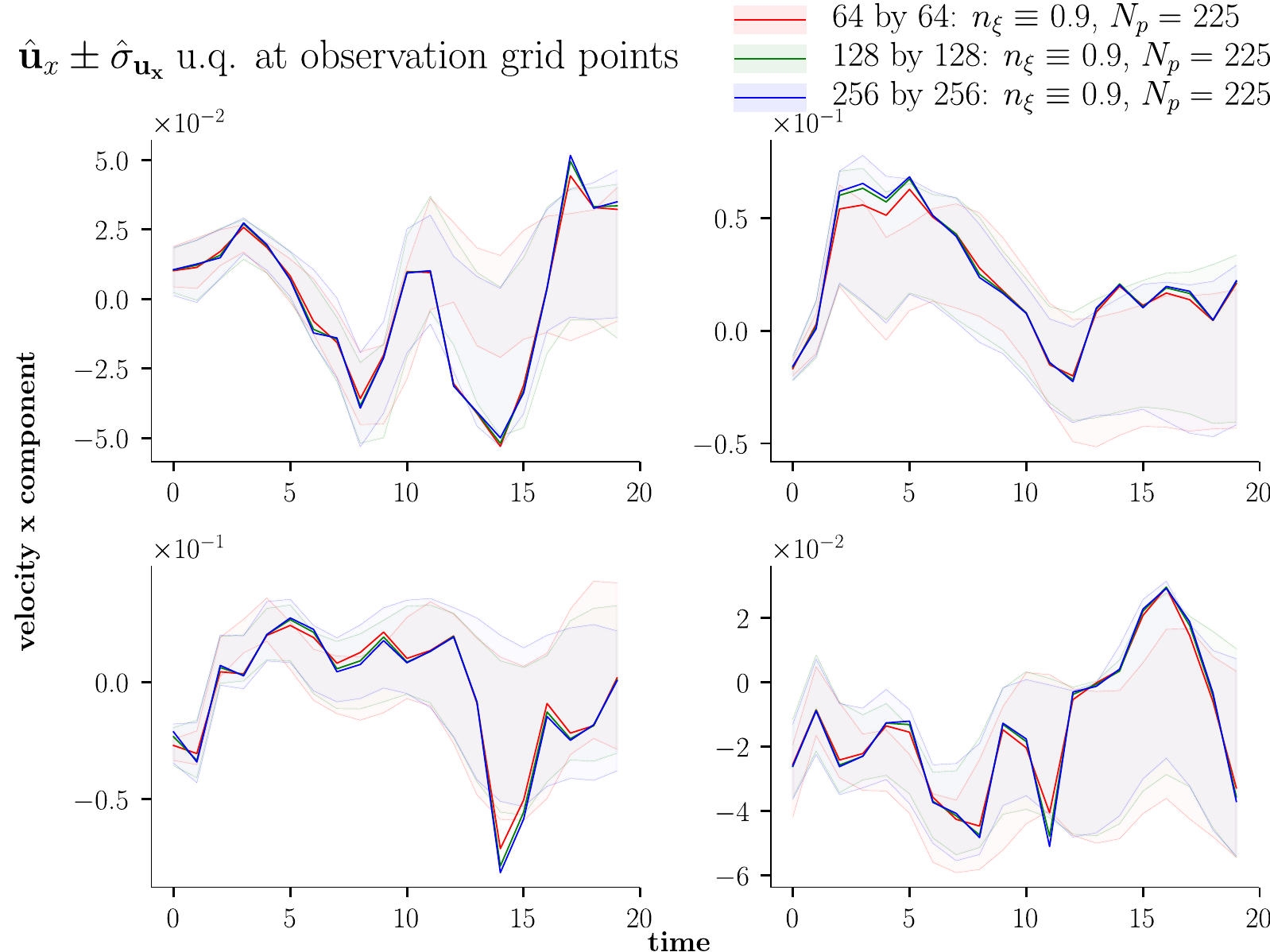}
                        \par\end{center}%
        \end{minipage}
        
        \smallskip{}
        \begin{minipage}[t]{0.49\textwidth}%
                \begin{center}
                        \includegraphics[width=1\textwidth]{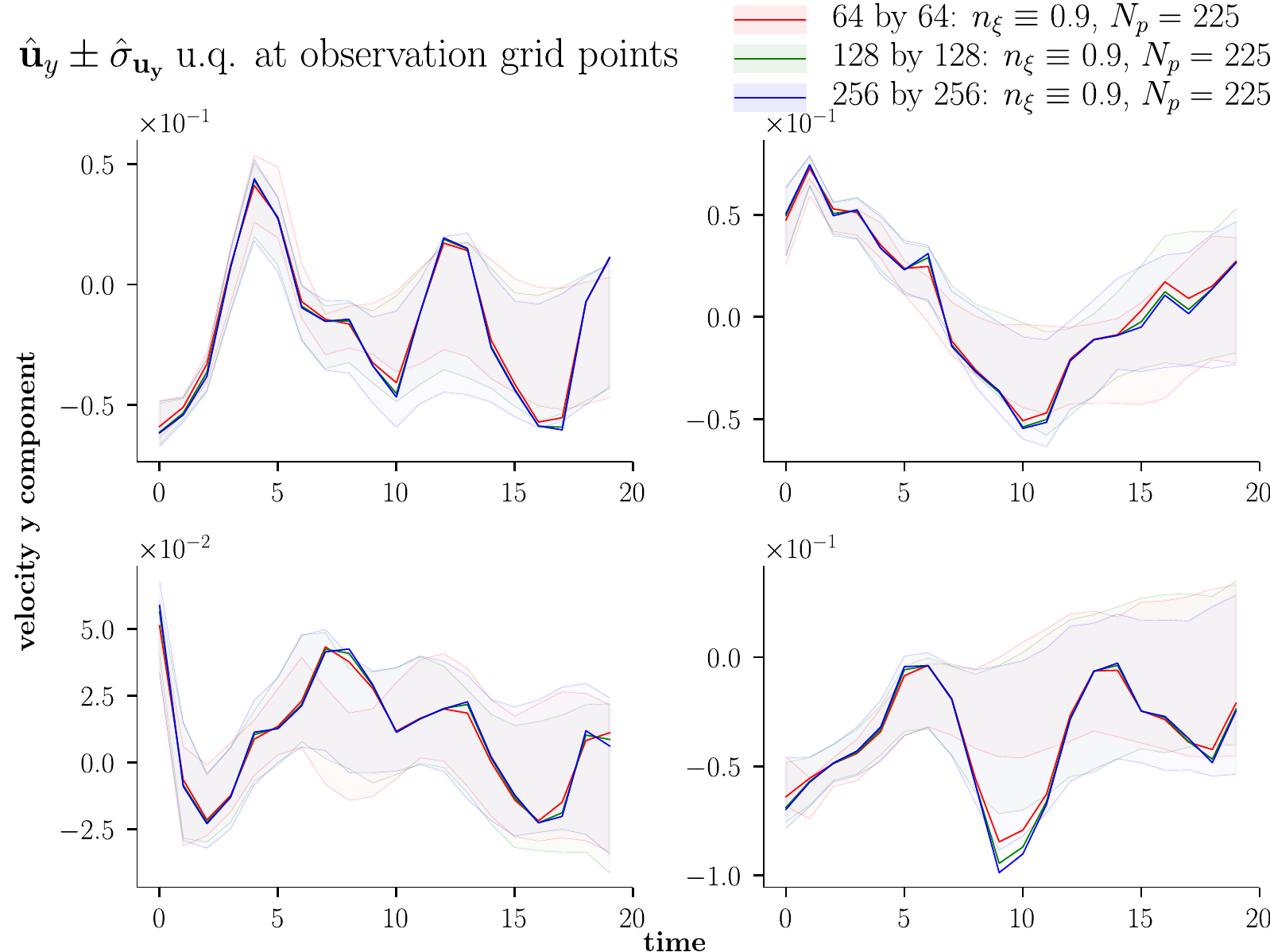}
                        \par\end{center}%
        \end{minipage}\hfill{}%
        \begin{minipage}[t]{0.49\textwidth}%
                \begin{center}
                        \includegraphics[width=1\textwidth]{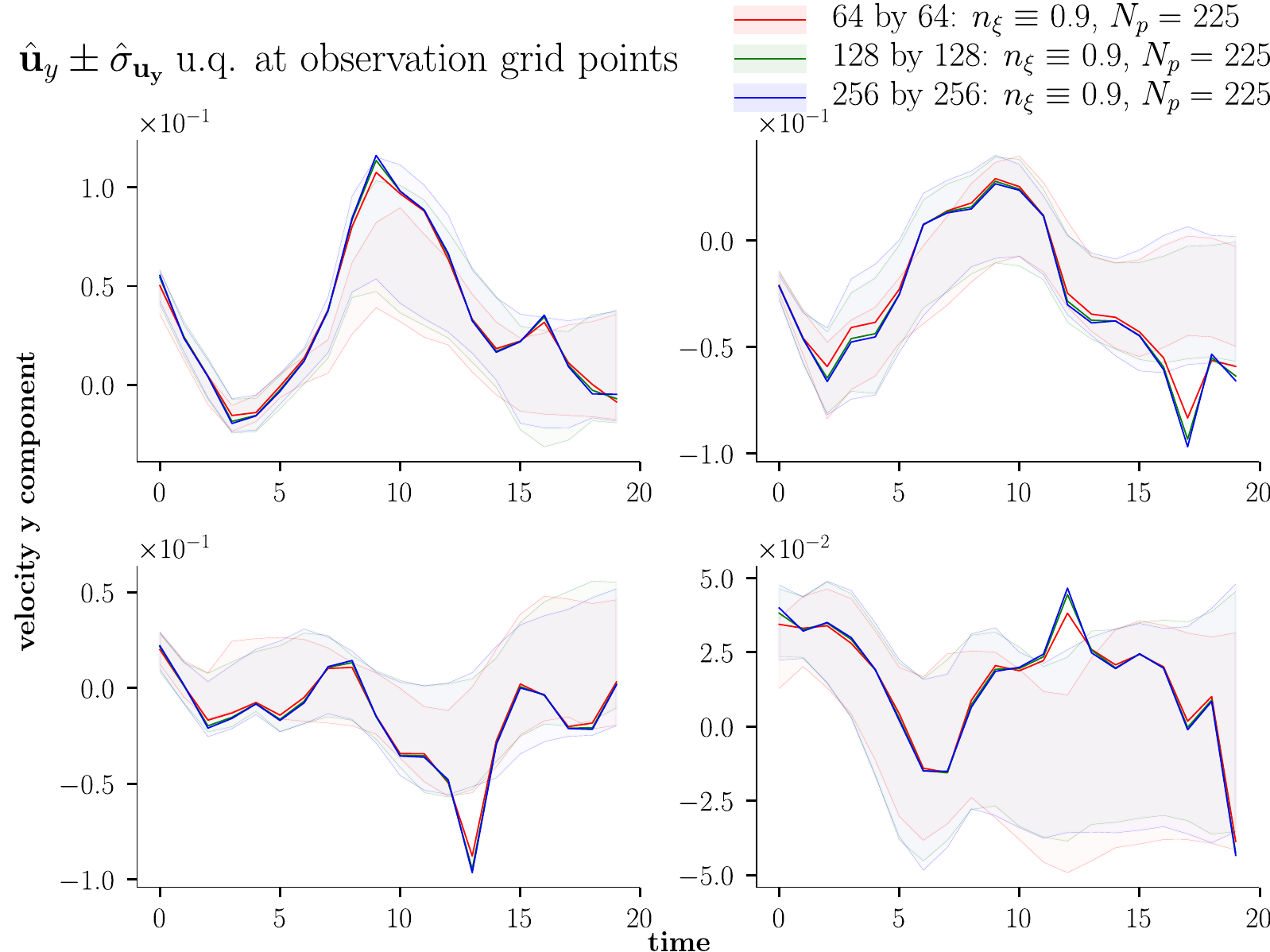}
                        \par\end{center}%
        \end{minipage}
        
        \caption{\label{fig:mike-cullen-u-multires} Uncertainty quantification comparing the truth with the ensemble one standard deviation region about the ensemble mean for the velocity, separate into the two components. The top two figures show plots corresponding to the $x$-component. The bottom two figures show plots corresponding to the $y$-component. For each component, the left and right hand figures each contain four plots. Each plot corresponds to a fixed grid point on a observation grid of size $4\times4$. For a fixed  number of EOFs $n_{\xi}\equiv0.9$ and a fixed number of particles in the ensemble $N_p=225$, each plot shows the truths and spreads corresponding to three coarse resolutions:  $64\times64,$ $128\times128$ and $256\times256.$. The solid lines represent the truth and the coloured regions represent the one standard deviation regions. Recall that the truth depends on the coarse grid size, see Figure \ref{fig:pde_solution_t0}. The red line and spread correspond to the $64\times64$ coarse grid. The green line and spread correspond to the $128\times128$ coarse grid. The blue line and spread correspond to the $256\times256$ coarse grid. The results are plotted for discrete ett time values and are linearly interpolated in between times. We see that as the coarse grid resolution gets refined, the one standard deviation region stays closer to the truth for longer time periods. This confirms that the parameterisation methodology is consistent under grid refinement. See Section \ref{subsec:Methodology} and \ref{subsec:spde_uq_results}}.
\end{figure}

We also investigate the relative $L^{2}$ distance between the SPDE ensemble and the coarse grained truth defined by
\begin{equation}
d\left(\left\{ \hat{q}^{i},i=1,\dots,N_{p}\right\} ,\omega,t\right):=\min_{i\in\left\{ 1,\dots,N_{p}\right\} }\frac{\left\Vert \omega(t)-\hat{q}^{i}\left(t\right)\right\Vert _{L^{2}\left(\mathcal{D}\right)}}{\left\Vert \omega(t)\right\Vert _{L^{2}\left(\mathcal{D}\right)}}\label{eq:min_l2_distance}
\end{equation}
for vorticity, and similarly defined for the streamfunction and velocity. We compute the results for different combinations of values of $n_{\xi},$ $N_{p},$ and the three coarse grids we are considering. Unlike the uncertainty quantification tests, where the analysis is done at individual grid points, here we consider the error between the truth and the particles over the whole domain $\mathcal{D}$. The results are shown in Figure \ref{fig:min_l2_q} for the vorticity, Figure \ref{fig:min_l2_psi} for the streamfunction and Figure \ref{fig:min_l2_u}
for the velocity. 

\begin{figure}
        \begin{minipage}[t]{1\textwidth}%
                \begin{center}
                        \includegraphics[width=1\textwidth]{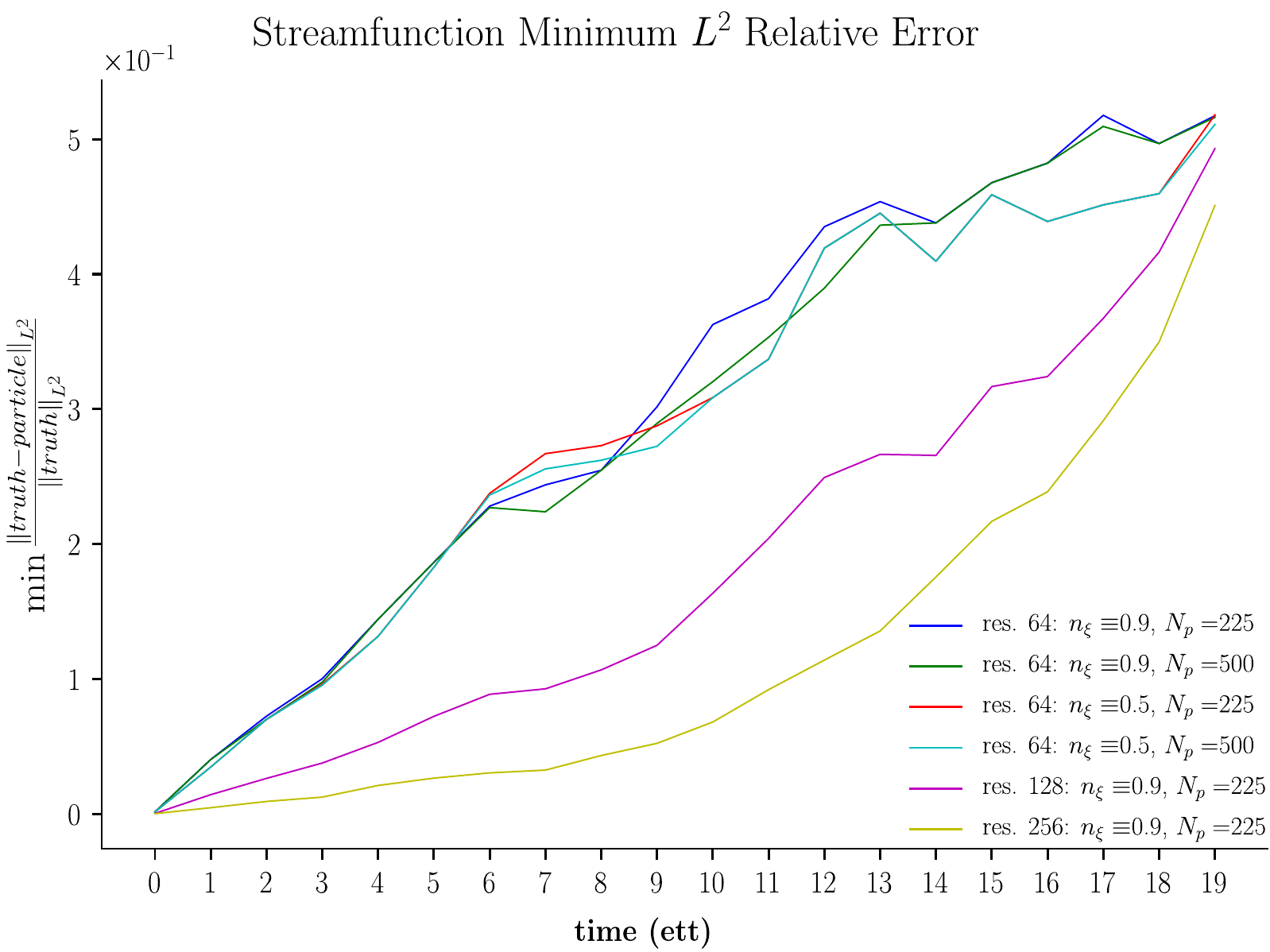}
                        \par\end{center}
                \caption{\label{fig:min_l2_psi}Plots of the relative $L^{2}$ distance between
                        SPDE ensemble and truth for the streamfunction $\psi$, starting at time $t_0$ for $20$ ett. The values are plotted at discrete time points, and are linearly interpolated in between times. The relative $L^2$ distance is defined in \eqref{eq:min_l2_distance}.      Each individual plot corresponds to different combinations of values of $n_{\xi}$, $N_{p}$ and coarse gird resolutions ($64\times64$, $128\times128$ and $256\times256$). Blue is for $n_{\xi}\equiv0.9$, $N_p = 225$ and resolution $64\times64$. Green is for $n_{\xi}\equiv0.9$, $N_p=500$, and resolution $64\times64$. Red is for $n_{\xi}\equiv0.5$, $N_p = 225$ and resolution $64\times64$. Cyan is for $n_{\xi}\equiv0.5$, $N_p=500$ and resolution $64\times64$. Magenta is for $n_{\xi}\equiv0.9$, $N_p=225$ and resolution $128\times128$. Yellow is for $n_{\xi}\equiv0.9$, $N_p=225$ and resolution $256\times256$. For fixed coarse grid resolution, changing the number of EOFs used to capture $90\%$ or $50\%$ variance and/or changing the ensemble size to $500$ or $225$ do not seem to impact on the relative $L^2$ distance between the ensemble and the truth. Changing the resolution of the coarse grid reduces the $L^2$ distance much more than the other parameters. Regardless of parameter combinations, the steady increase in the $L^2$ distance indicates the growing uncertainty in whether the ensemble captures the truth or not. See Section \ref{subsec:spde_uq_results}.}
        \end{minipage}
\end{figure}    
        
\begin{figure}
        \begin{minipage}[t]{1\textwidth}%
                \begin{center}
                        \includegraphics[width=1\textwidth]{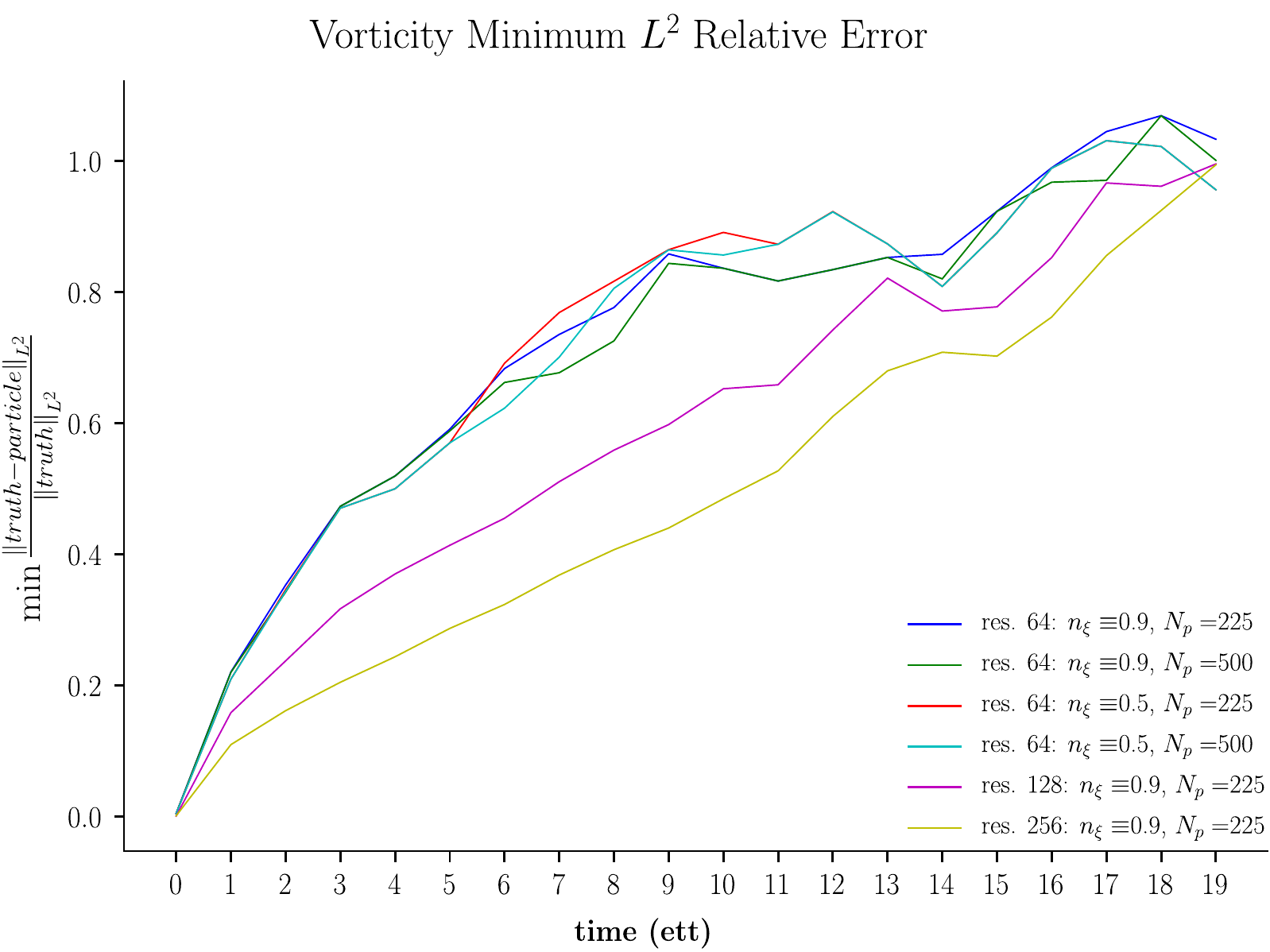}
                        \par\end{center}
                \caption{\label{fig:min_l2_q}Plots of the relative $L^{2}$ distance between
                        SPDE ensemble and truth for the vorticity, starting at time $t_0$ for $20$ ett. The values are plotted at discrete time points, and are linearly interpolated in between times. The relative $L^2$ distance is defined in \eqref{eq:min_l2_distance}.       Each individual plot corresponds to different combinations of values of $n_{\xi}$, $N_{p}$ and coarse gird resolutions ($64\times64$, $128\times128$ and $256\times256$). Blue is for $n_{\xi}\equiv0.9$, $N_p = 225$ and resolution $64\times64$. Green is for $n_{\xi}\equiv0.9$, $N_p=500$, and resolution $64\times64$. Red is for $n_{\xi}\equiv0.5$, $N_p = 225$ and resolution $64\times64$. Cyan is for $n_{\xi}\equiv0.5$, $N_p=500$ and resolution $64\times64$. Magenta is for $n_{\xi}\equiv0.9$, $N_p=225$ and resolution $128\times128$. Yellow is for $n_{\xi}\equiv0.9$, $N_p=225$ and resolution $256\times256$. For fixed coarse grid resolution, changing the number of EOFs used to capture $90\%$ or $50\%$ variance and/or changing the ensemble size to $500$ or $225$ do not seem to impact on the relative $L^2$ distance between the ensemble and the truth. Changing the resolution of the coarse grid reduces the $L^2$ distance much more than the other parameters. Regardless of parameter combinations, the steady increase in the $L^2$ distance indicates the growing uncertainty in whether the ensemble captures the truth or not. See Section \ref{subsec:spde_uq_results}.}
        \end{minipage}
\end{figure}

\begin{figure}
        \begin{minipage}[t]{1\textwidth}%
                \begin{center}
                        \includegraphics[width=1\textwidth]{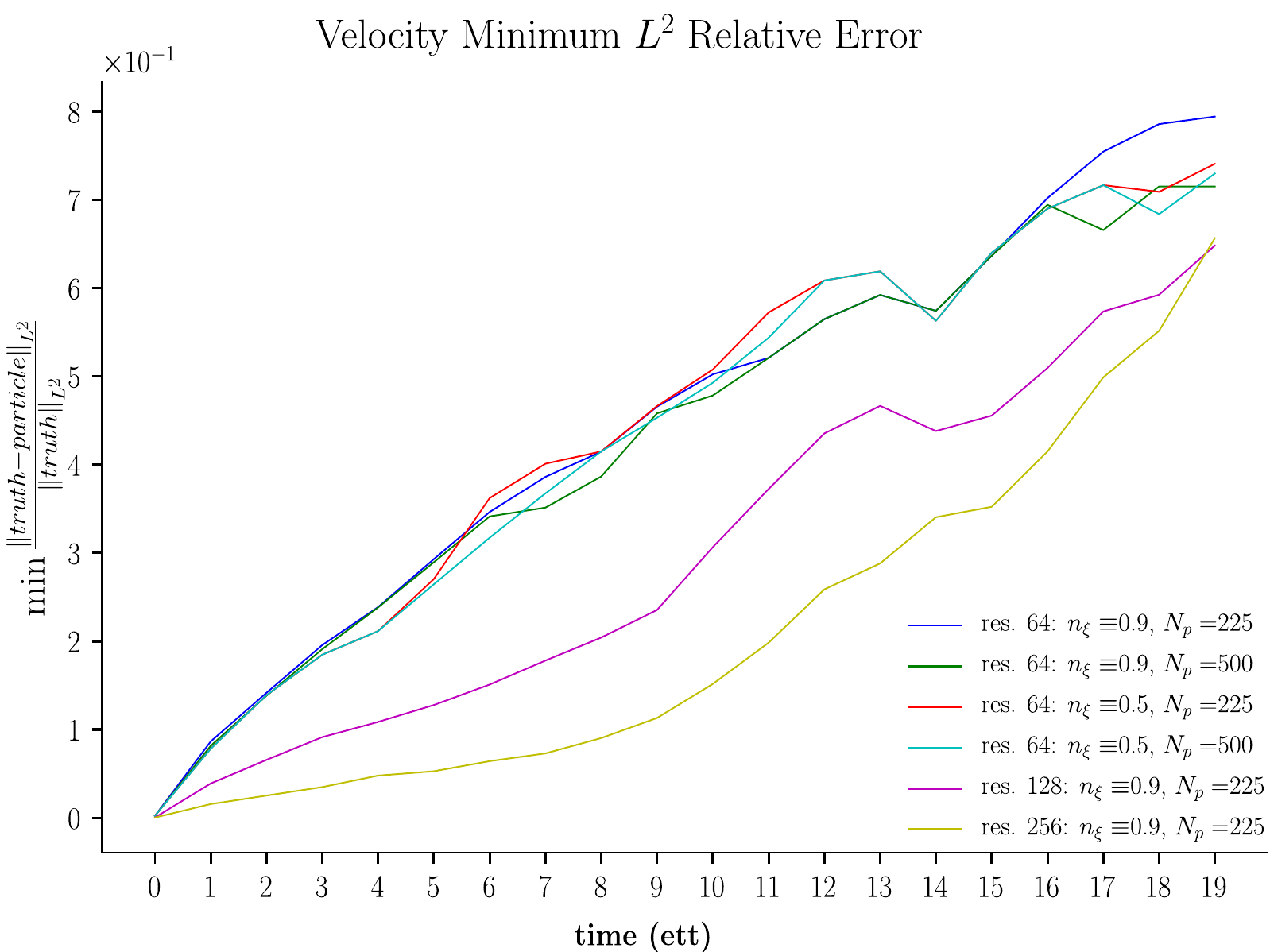}
                        \par\end{center}
                \caption{\label{fig:min_l2_u}Plots of the relative $L^{2}$ distance between
                        SPDE ensemble and truth for the velocity, starting at time $t_0$ for $20$ ett. The values are plotted at discrete time points, and are linearly interpolated in between times. The relative $L^2$ distance is defined in \eqref{eq:min_l2_distance}.       Each individual plot corresponds to different combinations of values of $n_{\xi}$, $N_{p}$ and coarse gird resolutions ($64\times64$, $128\times128$ and $256\times256$). Blue is for $n_{\xi}\equiv0.9$, $N_p = 225$ and resolution $64\times64$. Green is for $n_{\xi}\equiv0.9$, $N_p=500$, and resolution $64\times64$. Red is for $n_{\xi}\equiv0.5$, $N_p = 225$ and resolution $64\times64$. Cyan is for $n_{\xi}\equiv0.5$, $N_p=500$ and resolution $64\times64$. Magenta is for $n_{\xi}\equiv0.9$, $N_p=225$ and resolution $128\times128$. Yellow is for $n_{\xi}\equiv0.9$, $N_p=225$ and resolution $256\times256$. For fixed coarse grid resolution, changing the number of EOFs used to capture $90\%$ or $50\%$ variance and/or changing the ensemble size to $500$ or $225$ do not seem to impact on the relative $L^2$ distance between the ensemble and the truth. Changing the resolution of the coarse grid reduces the $L^2$ distance much more than the other parameters. Regardless of parameter combinations, the steady increase in the $L^2$ distance indicates the growing uncertainty in whether the ensemble captures the truth or not. See Section \ref{subsec:spde_uq_results}. }
        \end{minipage}
\end{figure}

In these figures, we see that the $L^{2}$ error between the ensemble and the coarse grained truth for each parameter set increases over time. The increase in error is much slower initially for the higher resolutions. This again gives us confidence in our parameterisation. Whilst the error remains small initially, also indicated by the uncertainty quantification
results, its increase can be corrected for using data assimilation techniques to incorporate observation data. This is part of our on going research. 

\subsubsection{Additional statistical tests \label{subsec:additional_stats_results}}

The Lie transport noise is not additive (nor is it multiplicative), thus we do not expect the SPDE solutions to be Gaussian. We can visually check whether our SPDE ensembles are non-Gaussian by computing boxplots and quantile-quantile (QQ) plots at fixed times. In a QQ plot, quantiles of two probability distributions are plotted against each other, see \cite{koch2013analysis}. If two distributions are similar, the QQ plot would show points lying on the line $y=x$. 
Figures \ref{fig:qqplot-psi}, \ref{fig:qqplot-q} and \ref{fig:qqplot-u} show the QQ plots for $\psi$, $q$ and $\vecu$ respectively at individual observation grid points at time $t_0 + 4$ ett. In many of the plots, we observe `smiles' with extremely curved tails, thus providing strong evidence to the fact that the ensembles are non-Gaussian. Non-Gaussian scaling is interpreted as intermittency in turbulence theory, see \cite{She1991intermittency}. Figures \ref{fig:boxplot-psi}, \ref{fig:boxplot-q} and \ref{fig:boxplot-u} show the boxplots for $\psi$, $q$ and $\vecu$ respectively at individual observation grid points at time points $t_0 + 1$ ett, $t_0 + 2$ ett and $t_0 + 5$ ett. The plots show non-symmetry and fat tails in the distribution of the ensembles, again providing strong evidence to the fact that the ensembles are non-Gaussian.

\begin{figure}
        \begin{minipage}[t]{0.49\textwidth}%
                \begin{center}
                        \includegraphics[width=1\textwidth]{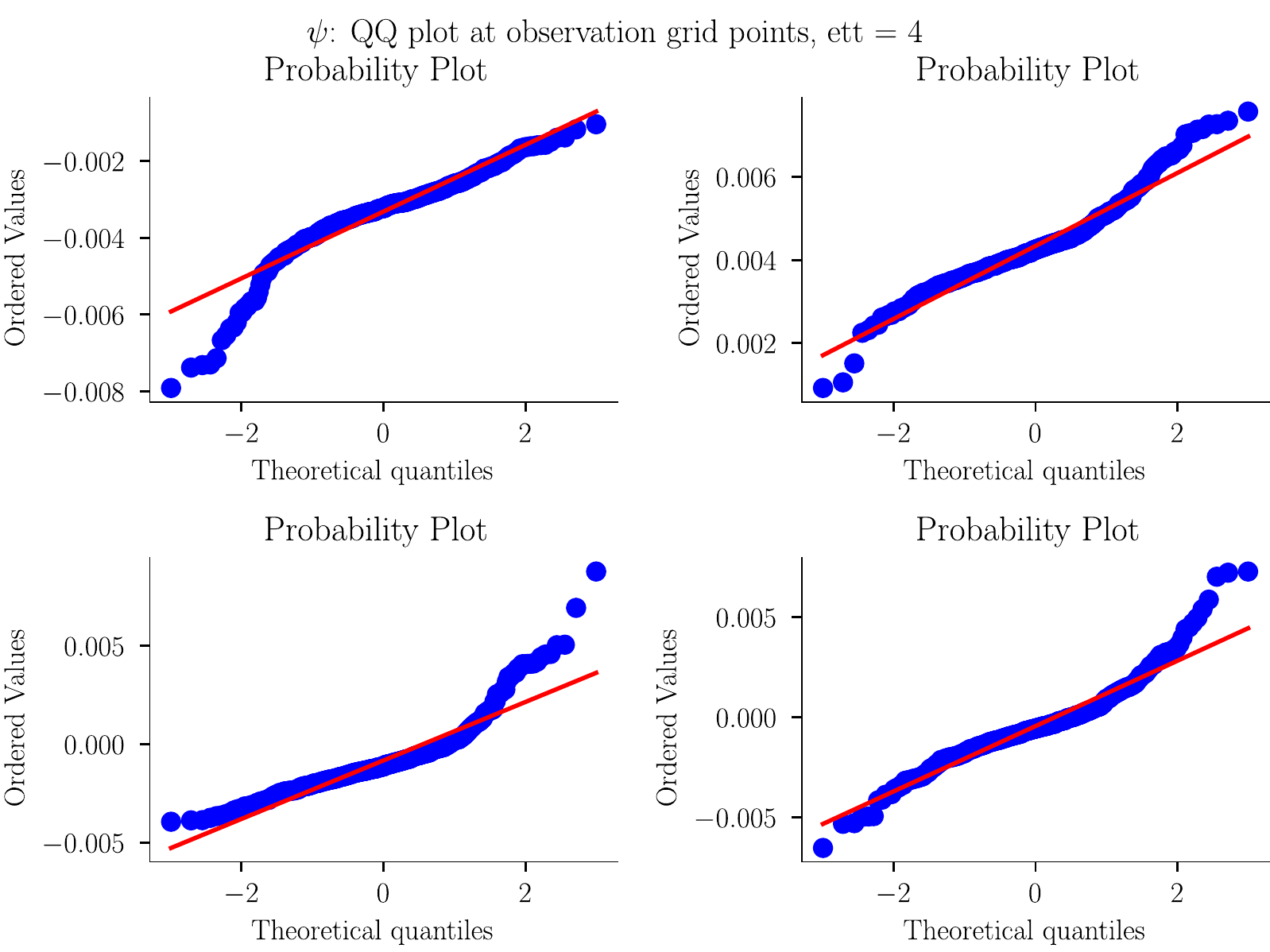}
                        \par\end{center}%
        \end{minipage}\hfill{}%
        \begin{minipage}[t]{0.49\textwidth}%
                \begin{center}
                        \includegraphics[width=1\textwidth]{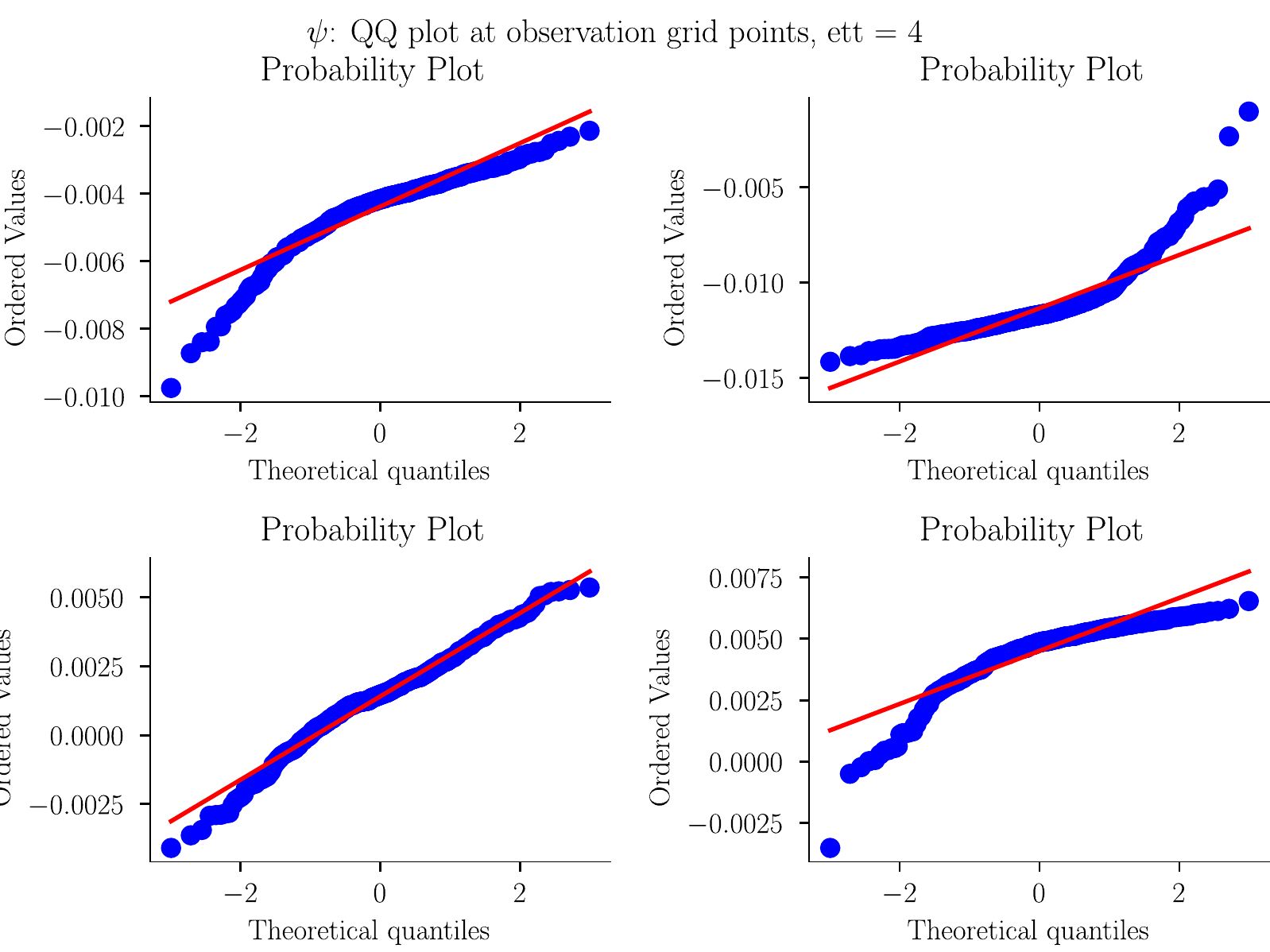}
                        \par\end{center}%
        \end{minipage}
        
        \caption{\label{fig:qqplot-psi}Quantile-Quantile (QQ) plots for the SPDE ensemble streamfunction
                at time $t=4$ ett, at eight observation grid (size $4\times4$) points, shown here in
                two separate figures of four plots each. Each plot correspond to an individual grid point. The plot compares the ensemble quantiles to the theoretical quantiles from the Gaussian distribution. If the ensemble is Gaussian, then we would see the plotted points (in blue) lying on the line $y=x$ (shown in red). The fact that the plots show fat tails give strong evidence to the fact that the ensembles are not Gaussian. See Section \ref{subsec:additional_stats_results}.}
\end{figure}

\begin{figure}
        \begin{minipage}[t]{0.49\textwidth}%
                \begin{center}
                        \includegraphics[width=1\textwidth]{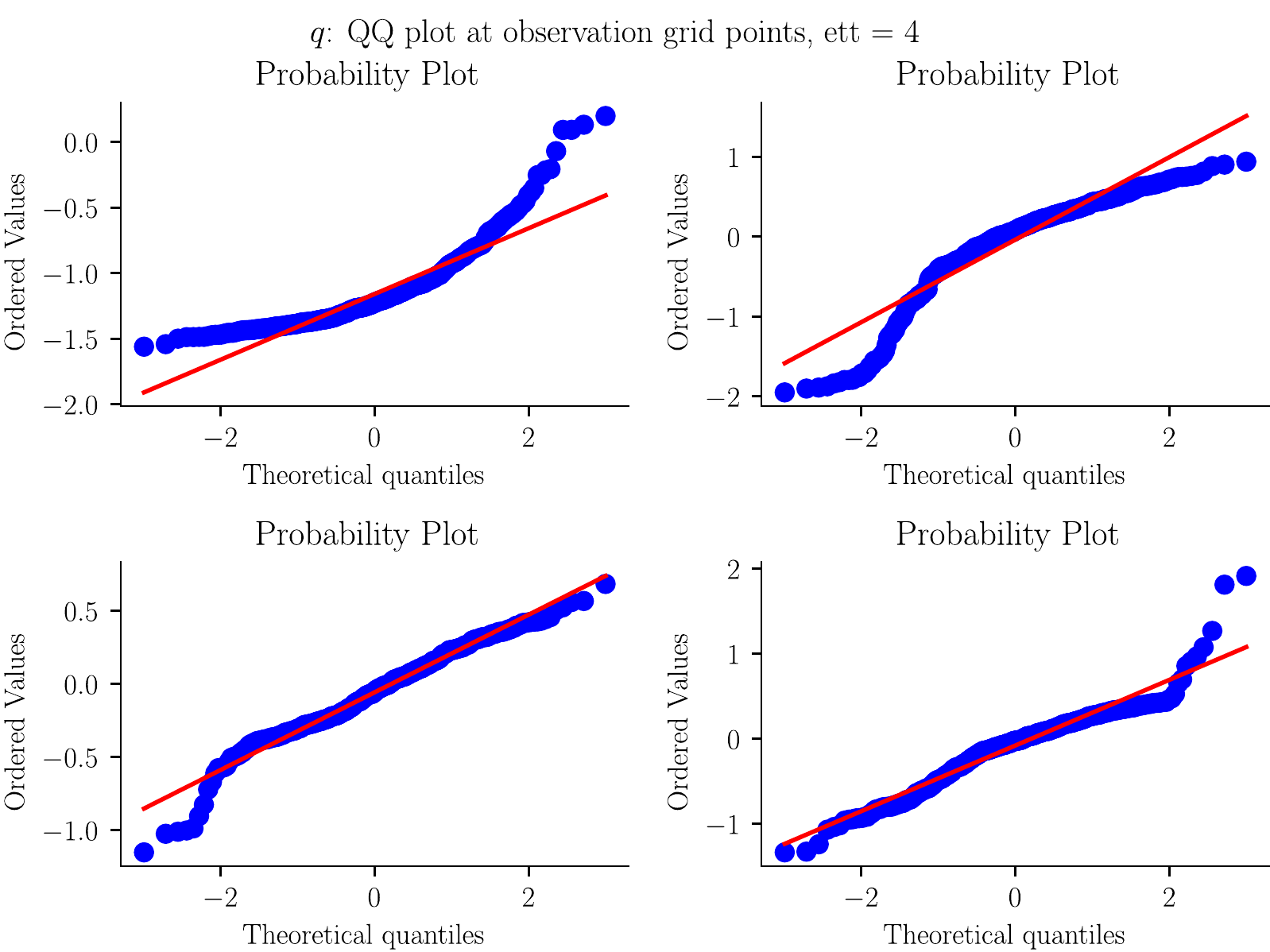}
                        \par\end{center}%
        \end{minipage}\hfill{}%
        \begin{minipage}[t]{0.49\textwidth}%
                \begin{center}
                        \includegraphics[width=1\textwidth]{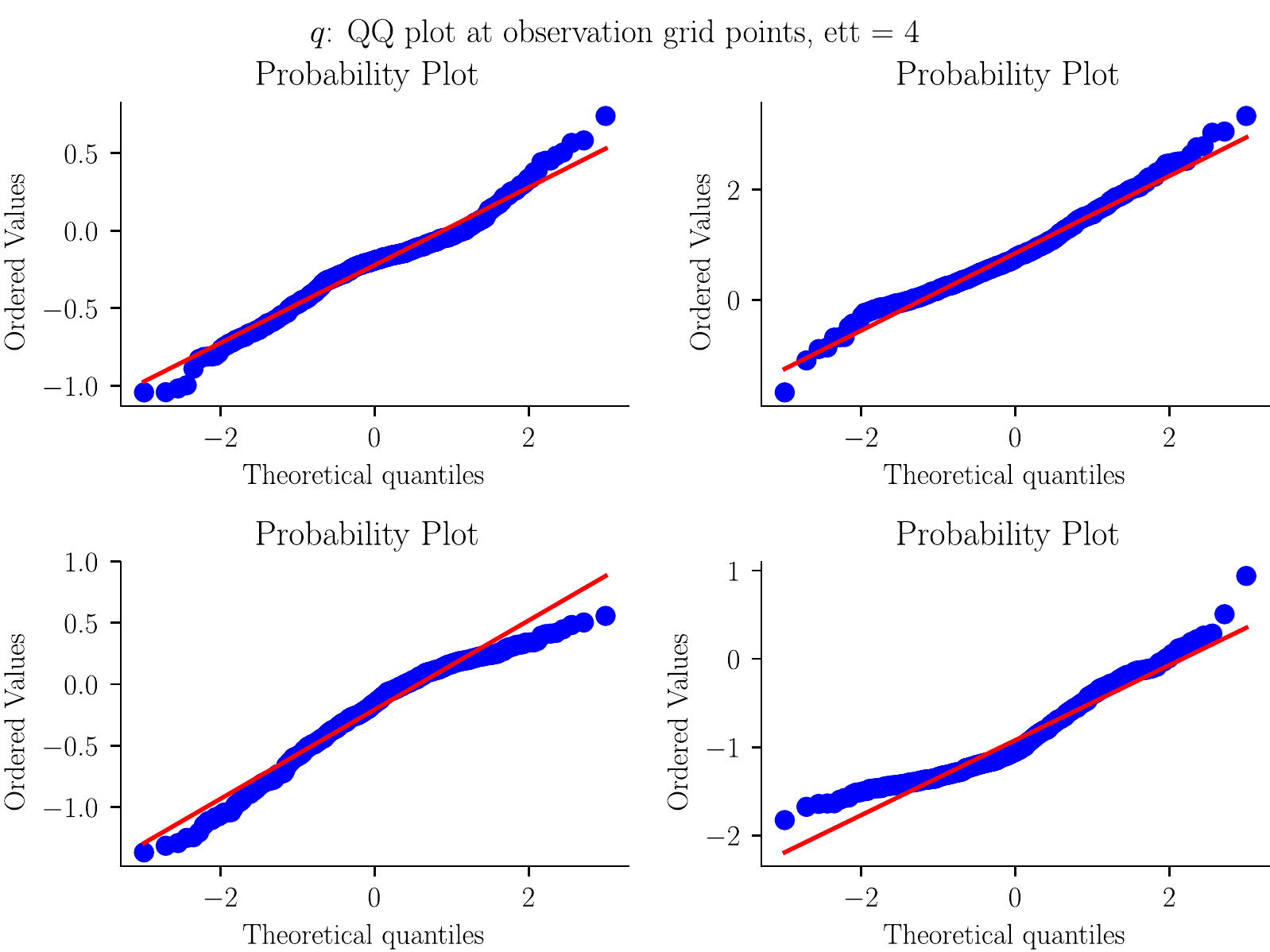}
                        \par\end{center}%
        \end{minipage}
        
        \caption{\label{fig:qqplot-q}Quantile-Quantile (QQ) plots for the SPDE ensemble vorticity
                at time $t=4$ ett, at eight observation grid (size $4\times4$) points, shown here in
                two separate figures of four plots each. Each plot correspond to an individual grid point. The plot compares the ensemble quantiles to the theoretical quantiles from the Gaussian distribution. If the ensemble is Gaussian, then we would see the plotted points (in blue) lying on the line $y=x$ (shown in red). The fact that the plots show fat tails give strong evidence to the fact that the ensembles are not Gaussian. See Section \ref{subsec:additional_stats_results}.}
\end{figure}

\begin{figure}
        \begin{minipage}[t]{0.49\textwidth}%
                \begin{center}
                        \includegraphics[width=1\textwidth]{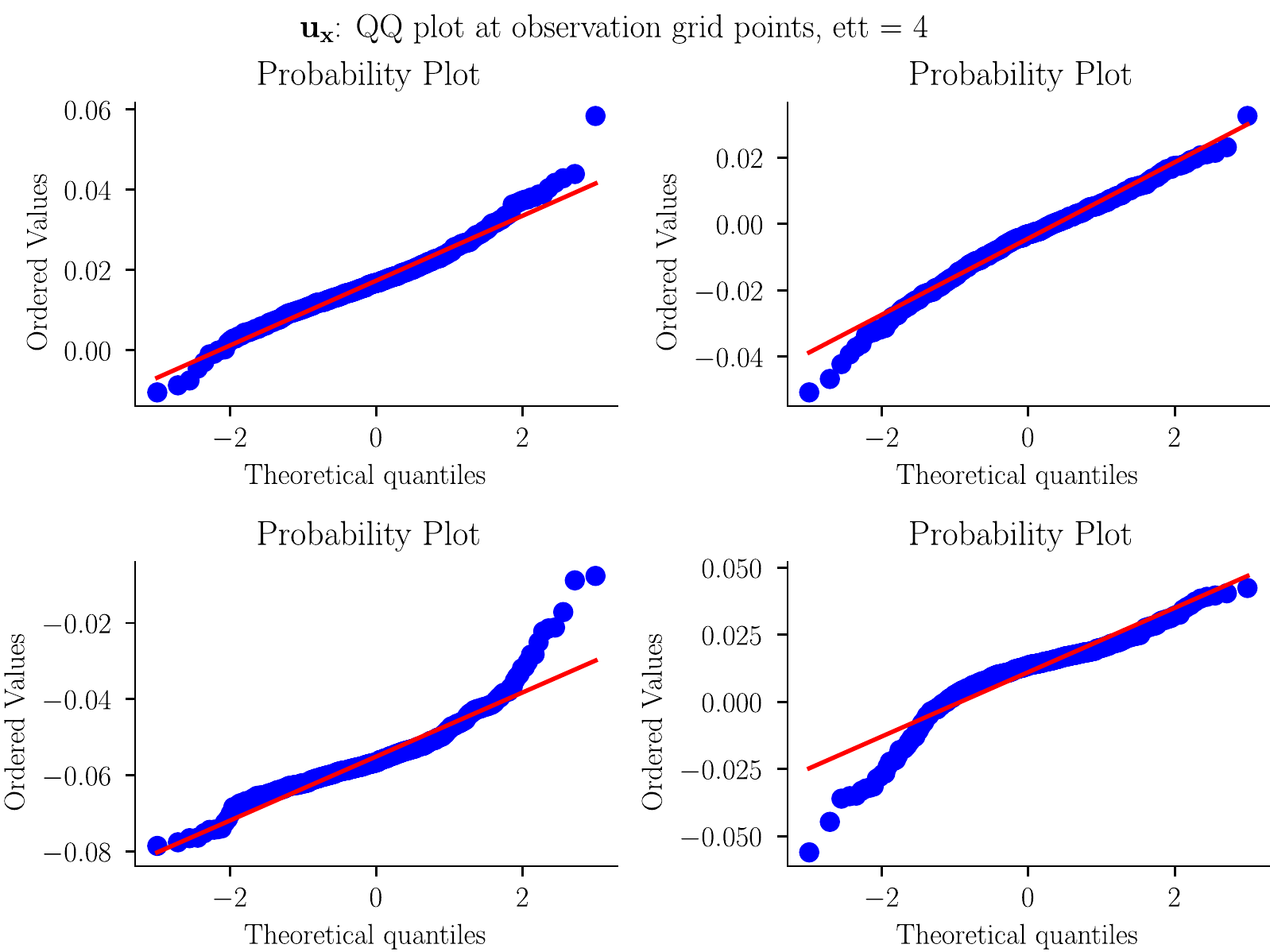}
                        \par\end{center}%
        \end{minipage}\hfill{}%
        \begin{minipage}[t]{0.49\textwidth}%
                \begin{center}
                        \includegraphics[width=1\textwidth]{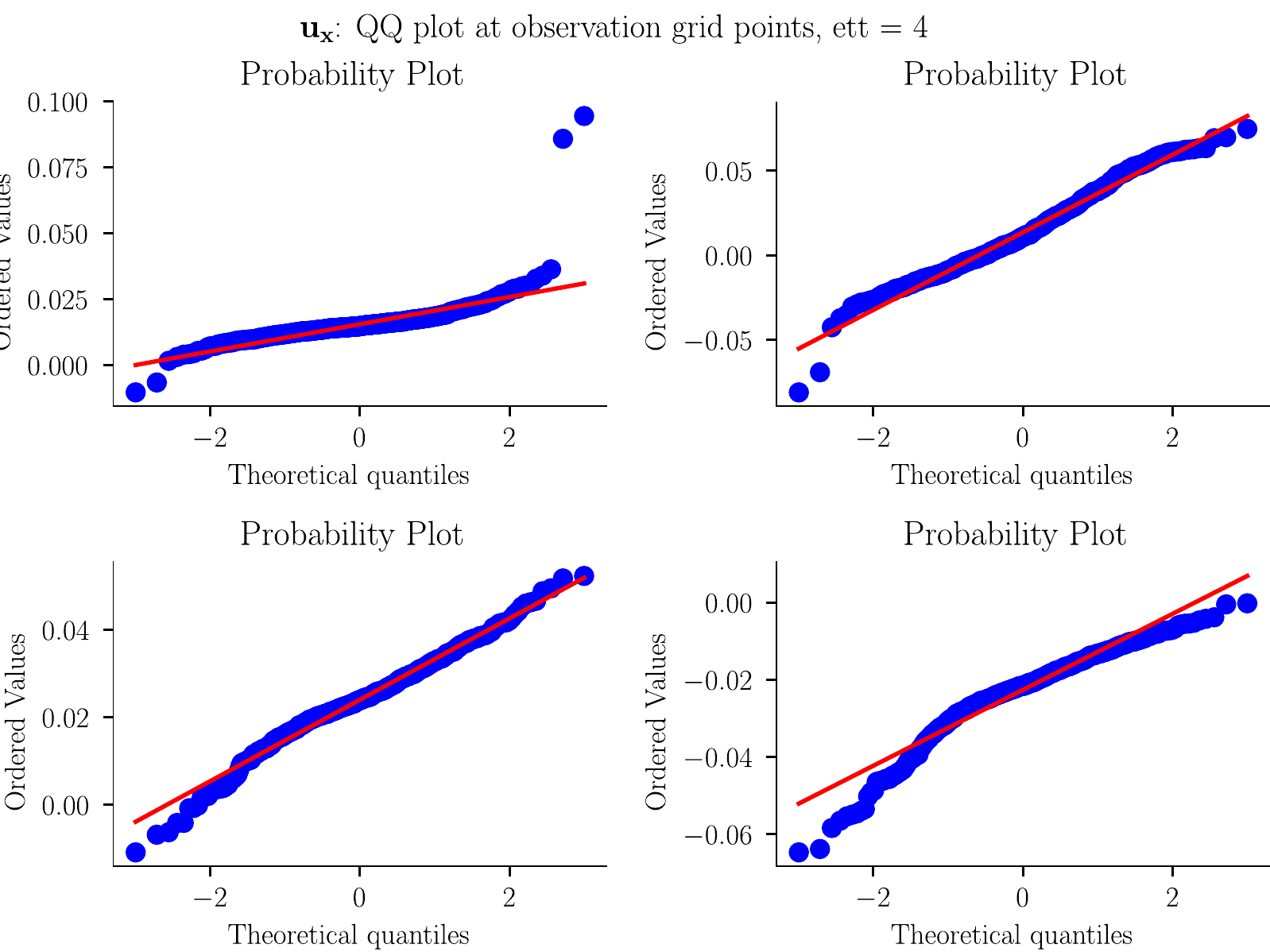}
                        \par\end{center}%
        \end{minipage}
        
        \smallskip{}
        \begin{minipage}[t]{0.49\textwidth}%
                \begin{center}
                        \includegraphics[width=1\textwidth]{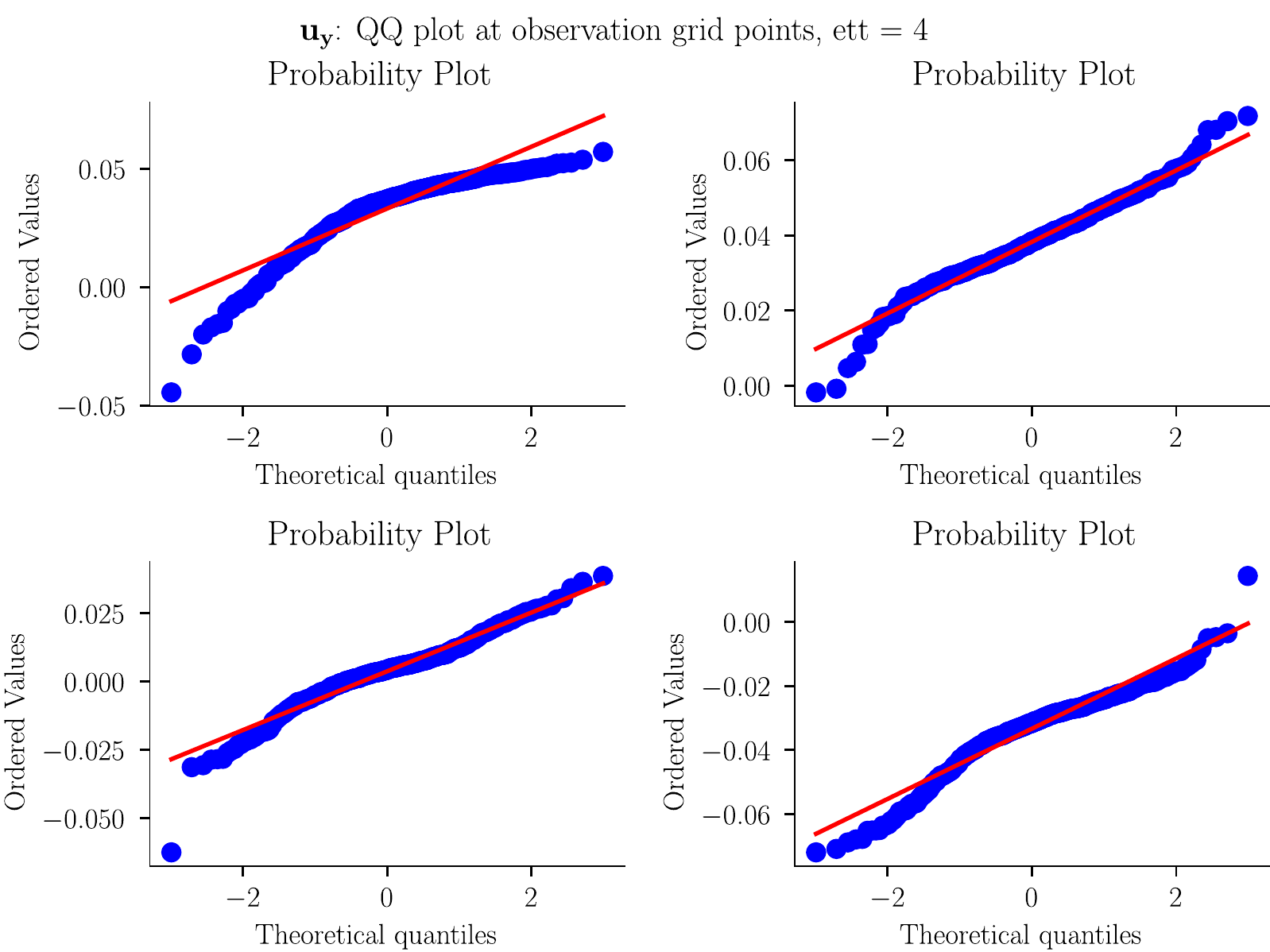}
                        \par\end{center}%
        \end{minipage}\hfill{}%
        \begin{minipage}[t]{0.49\textwidth}%
                \begin{center}
                        \includegraphics[width=1\textwidth]{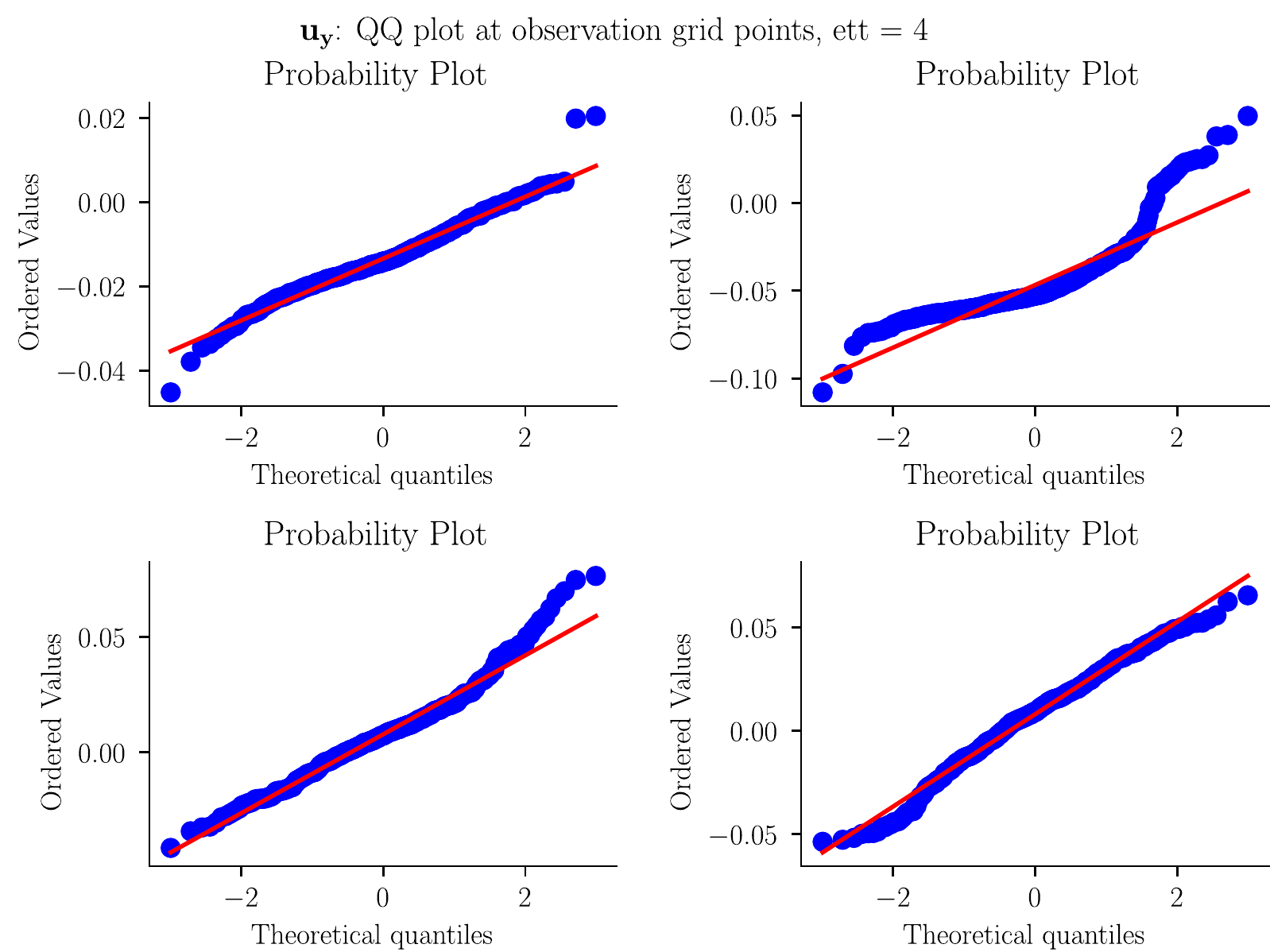}
                        \par\end{center}%
        \end{minipage}
        
        \caption{\label{fig:qqplot-u}Quantile-Quantile (QQ) plots for the SPDE ensemble velocity components
                at time $t=4$ ett, at eight observation grid (size $4\times4$) points. Top two figures show plots for the $x$-component, bottom two figures show plots for the $y$-component. Each figure contains four plots each. Each plot correspond to an individual grid point. The plot compares the ensemble quantiles to the theoretical quantiles from the Gaussian distribution. If the ensemble is Gaussian, then we would see the plotted points (in blue) lying on the line $y=x$ (shown in red). The fact that the plots show fat tails give strong evidence to the fact that the ensembles are not Gaussian. See Section \ref{subsec:additional_stats_results}.}
\end{figure}

\begin{figure}
        \begin{minipage}[t]{0.49\textwidth}%
                \begin{center}
                        \includegraphics[width=1\textwidth]{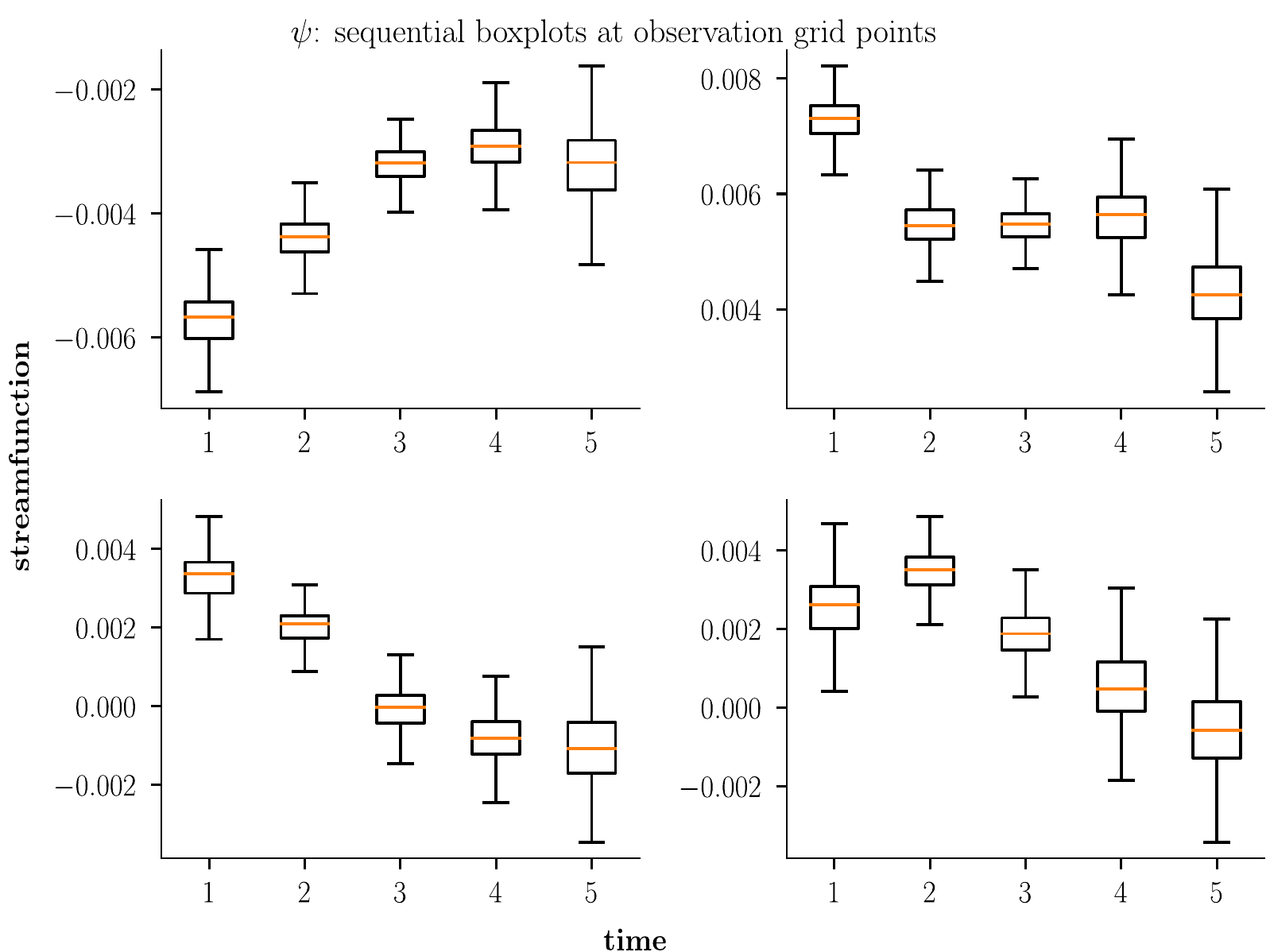}
                        \par\end{center}%
        \end{minipage}\hfill{}%
        \begin{minipage}[t]{0.49\textwidth}%
                \begin{center}
                        \includegraphics[width=1\textwidth]{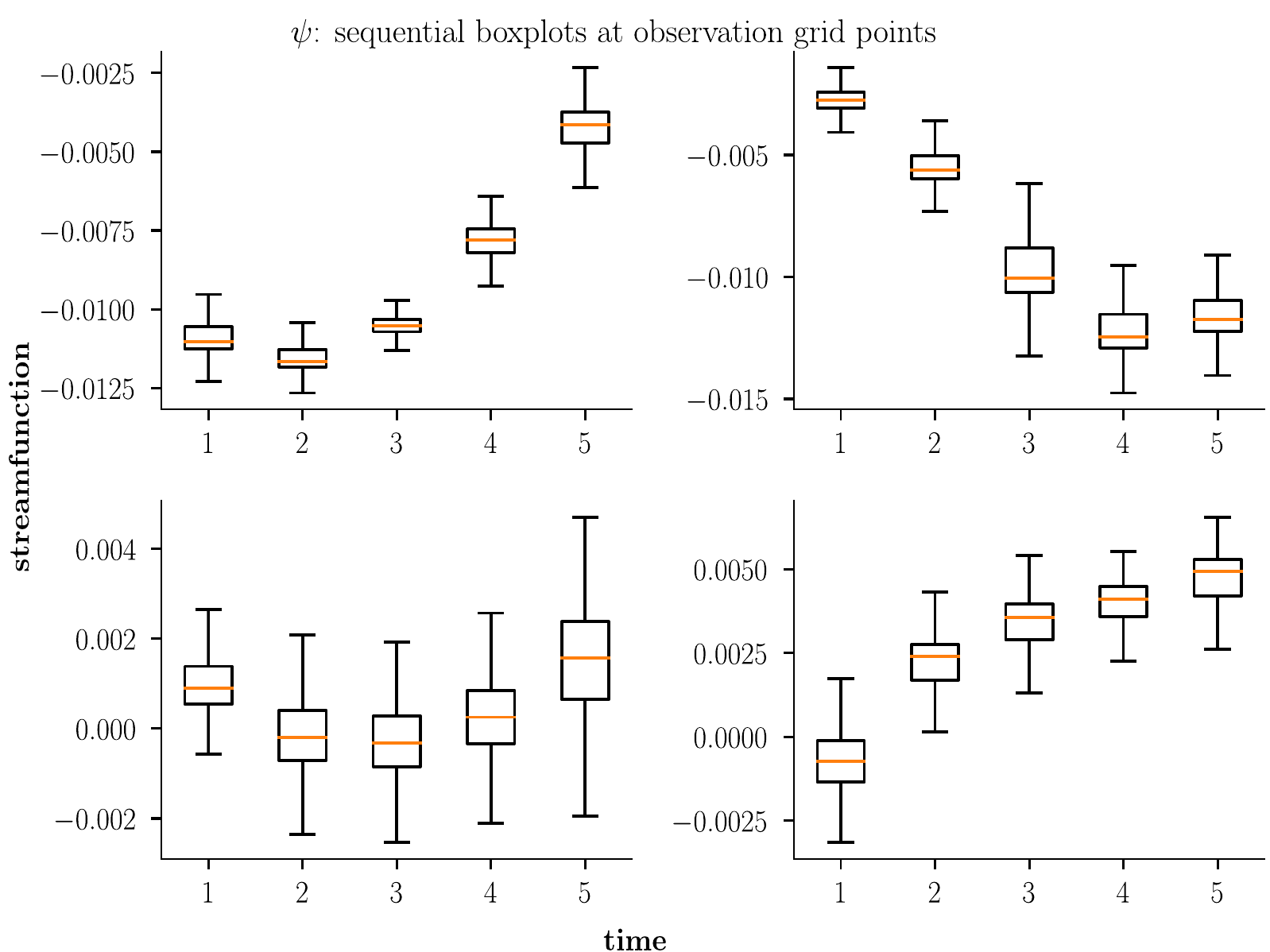}
                        \par\end{center}%
        \end{minipage}
        
        \caption{\label{fig:boxplot-psi}Box plots for the SPDE ensemble streamfunction
                $\psi$ at times $t=1,2,\dots,5$ in ett, at eight observation grid
                points, shown here in two separate figures of four plots each (each plot corresponds to one grid point). The boxplots show non-symmetry and fat tails in the distribution of the ensembles, again providing strong evidence to the fact that the ensembles are non-Gaussian. }
\end{figure}

\begin{figure}
        \begin{minipage}[t]{0.49\textwidth}%
                \begin{center}
                        \includegraphics[width=1\textwidth]{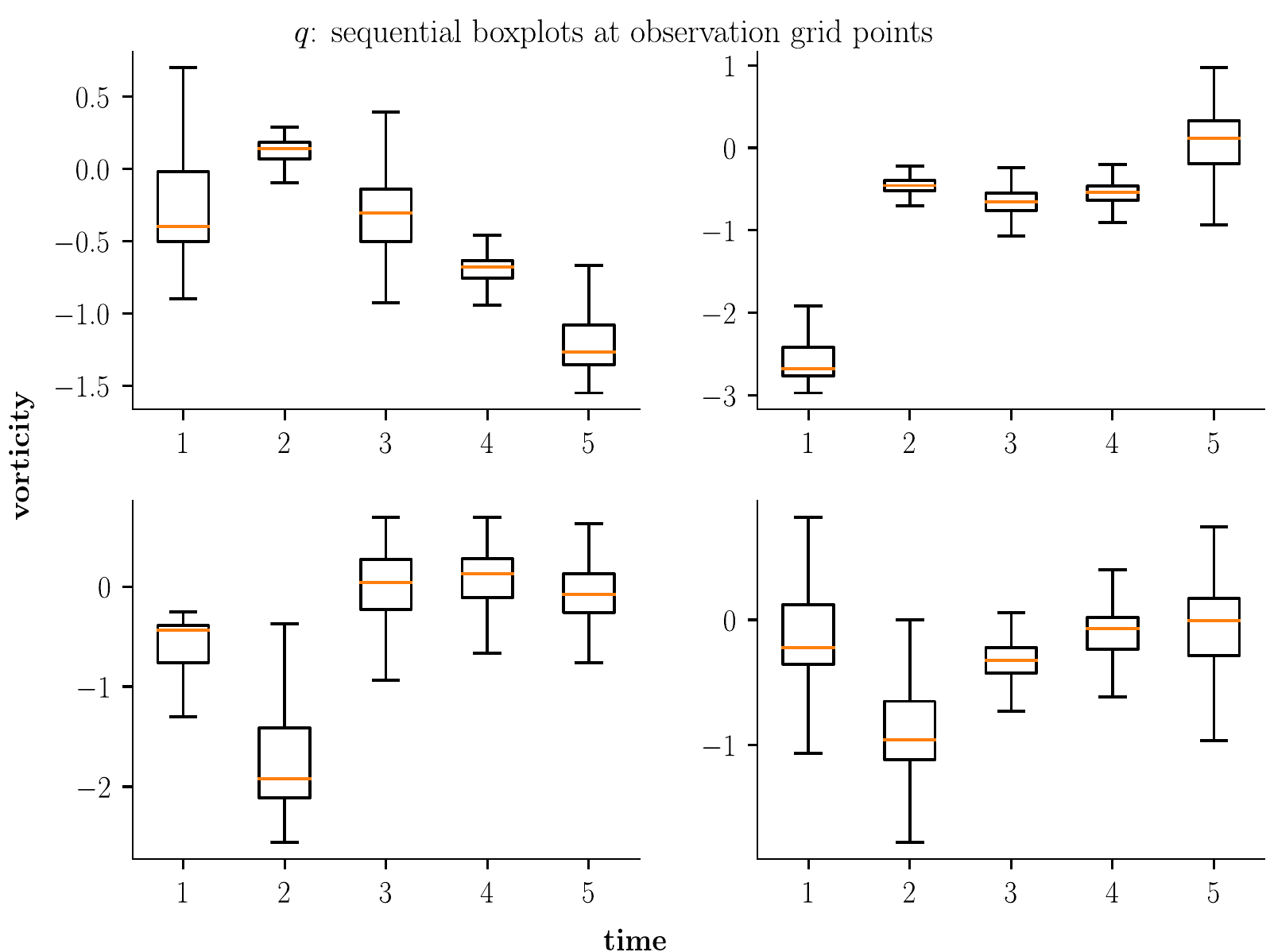}
                        \par\end{center}%
        \end{minipage}\hfill{}%
        \begin{minipage}[t]{0.49\textwidth}%
                \begin{center}
                        \includegraphics[width=1\textwidth]{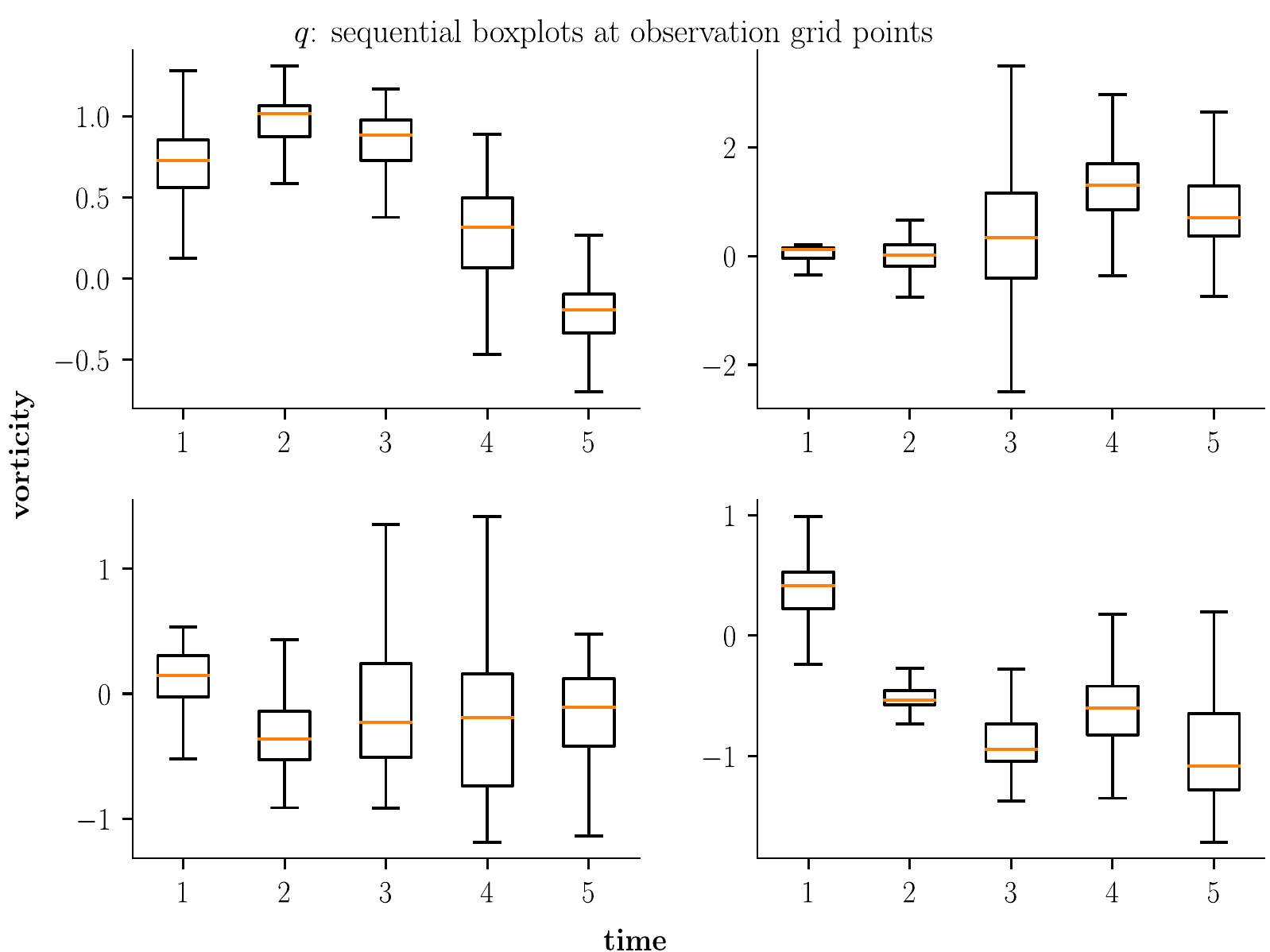}
                        \par\end{center}%
        \end{minipage}
        
        \caption{\label{fig:boxplot-q}Box plots for the SPDE ensemble vorticity $q$
                at times $t=1,2,\dots,5$ in ett, at eight observation grid points,
                shown here in two separate figures of four plots each (each plot corresponds to one grid point). The plots show non-symmetry and fat tails in the distribution of the ensembles, again providing strong evidence to the fact that the ensembles are non-Gaussian.}
\end{figure}

\begin{figure}
        \begin{minipage}[t]{0.49\textwidth}%
                \begin{center}
                        \includegraphics[width=1\textwidth]{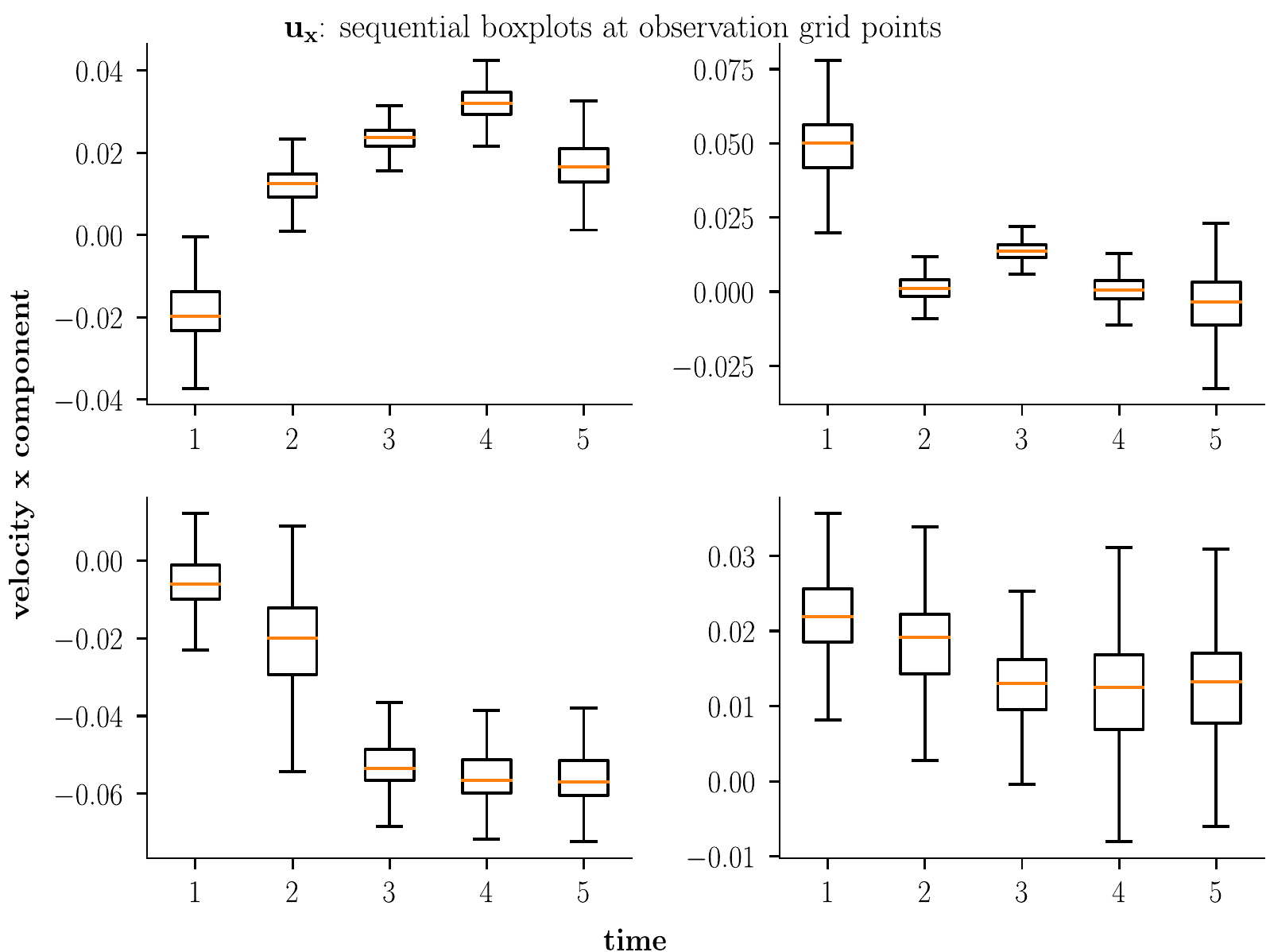}
                        \par\end{center}%
        \end{minipage}\hfill{}%
        \begin{minipage}[t]{0.49\textwidth}%
                \begin{center}
                        \includegraphics[width=1\textwidth]{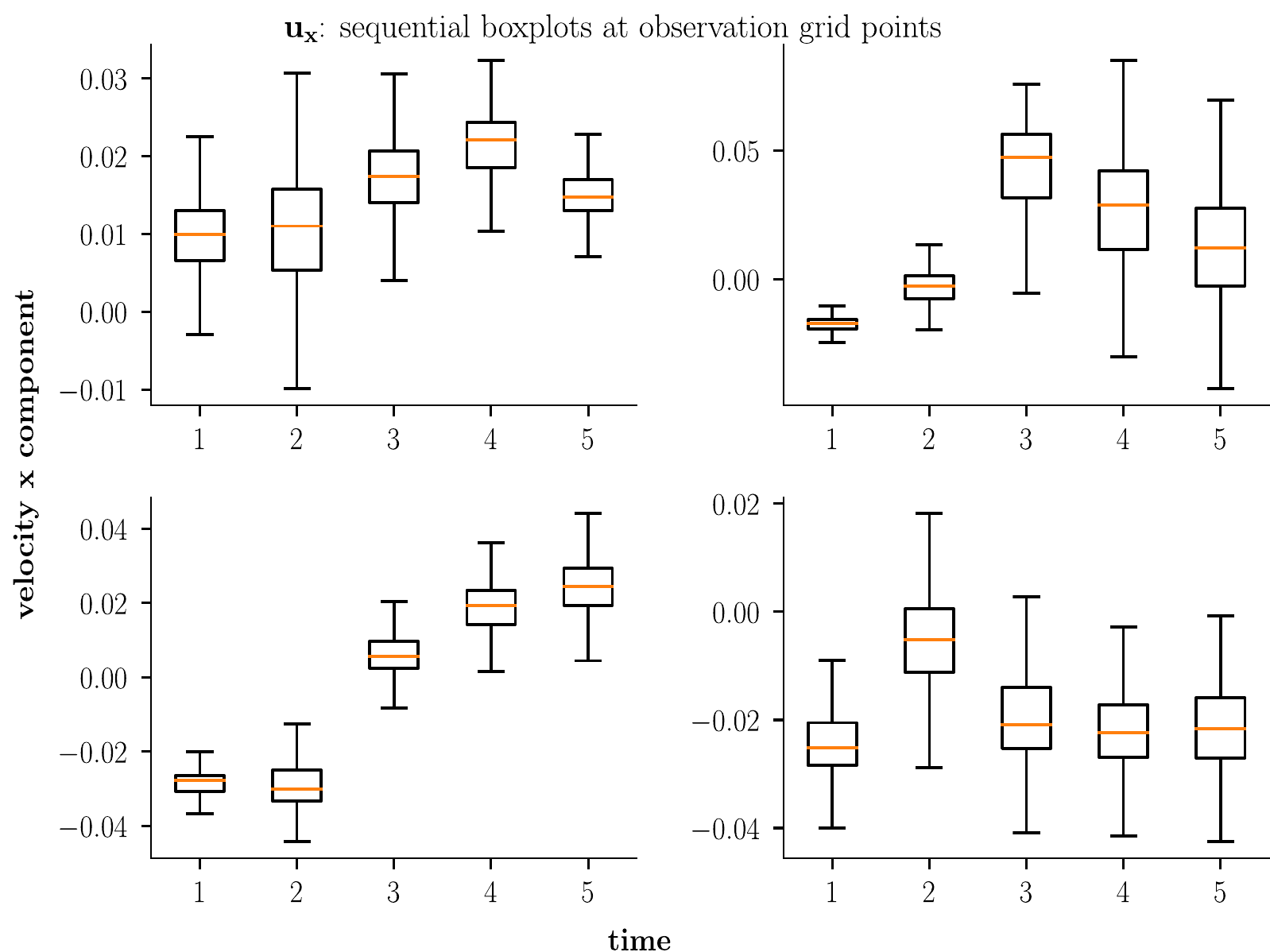}
                        \par\end{center}%
        \end{minipage}
        
        \smallskip{}
        \begin{minipage}[t]{0.49\textwidth}%
                \begin{center}
                        \includegraphics[width=1\textwidth]{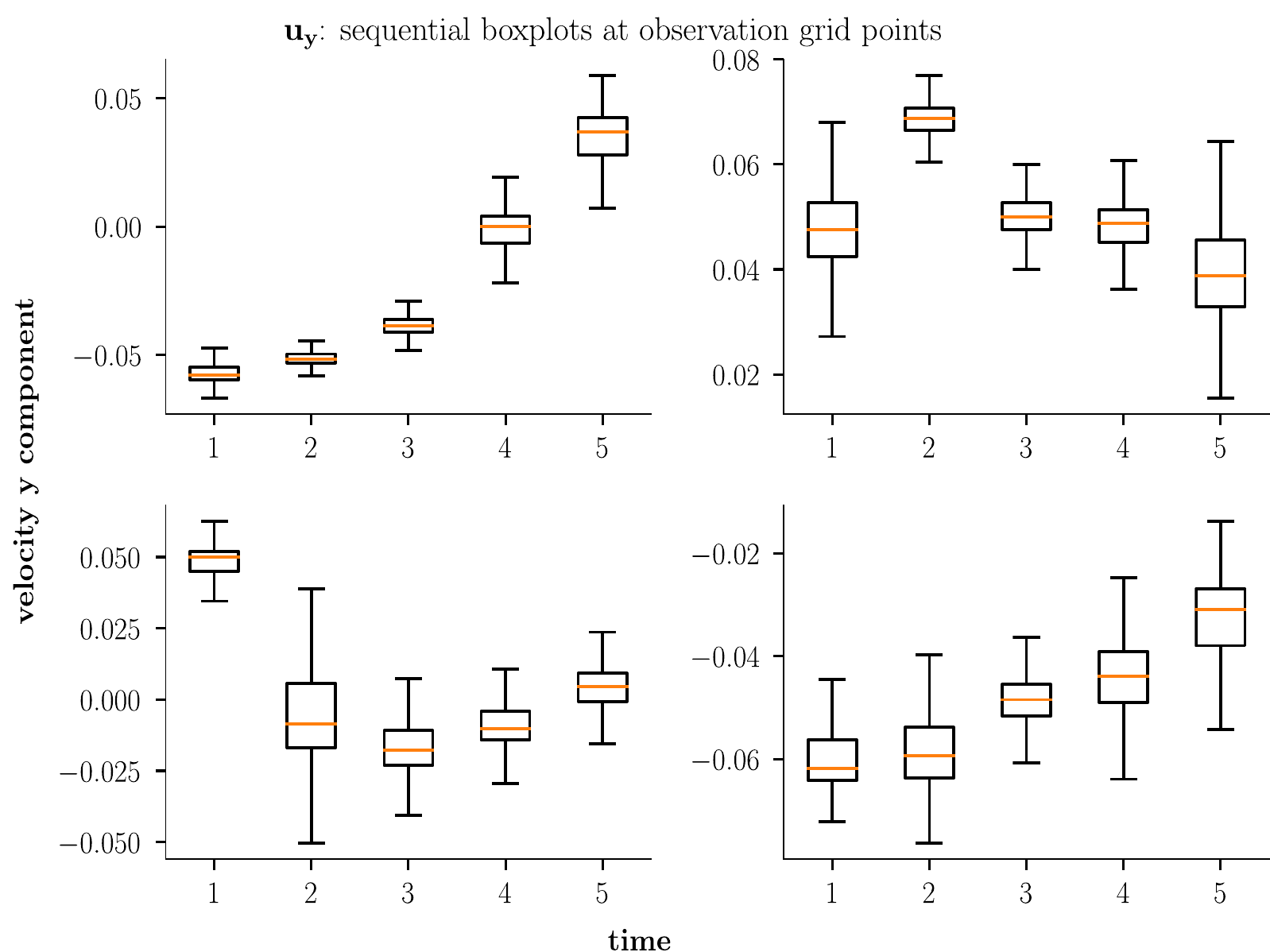}
                        \par\end{center}%
        \end{minipage}\hfill{}%
        \begin{minipage}[t]{0.49\textwidth}%
                \begin{center}
                        \includegraphics[width=1\textwidth]{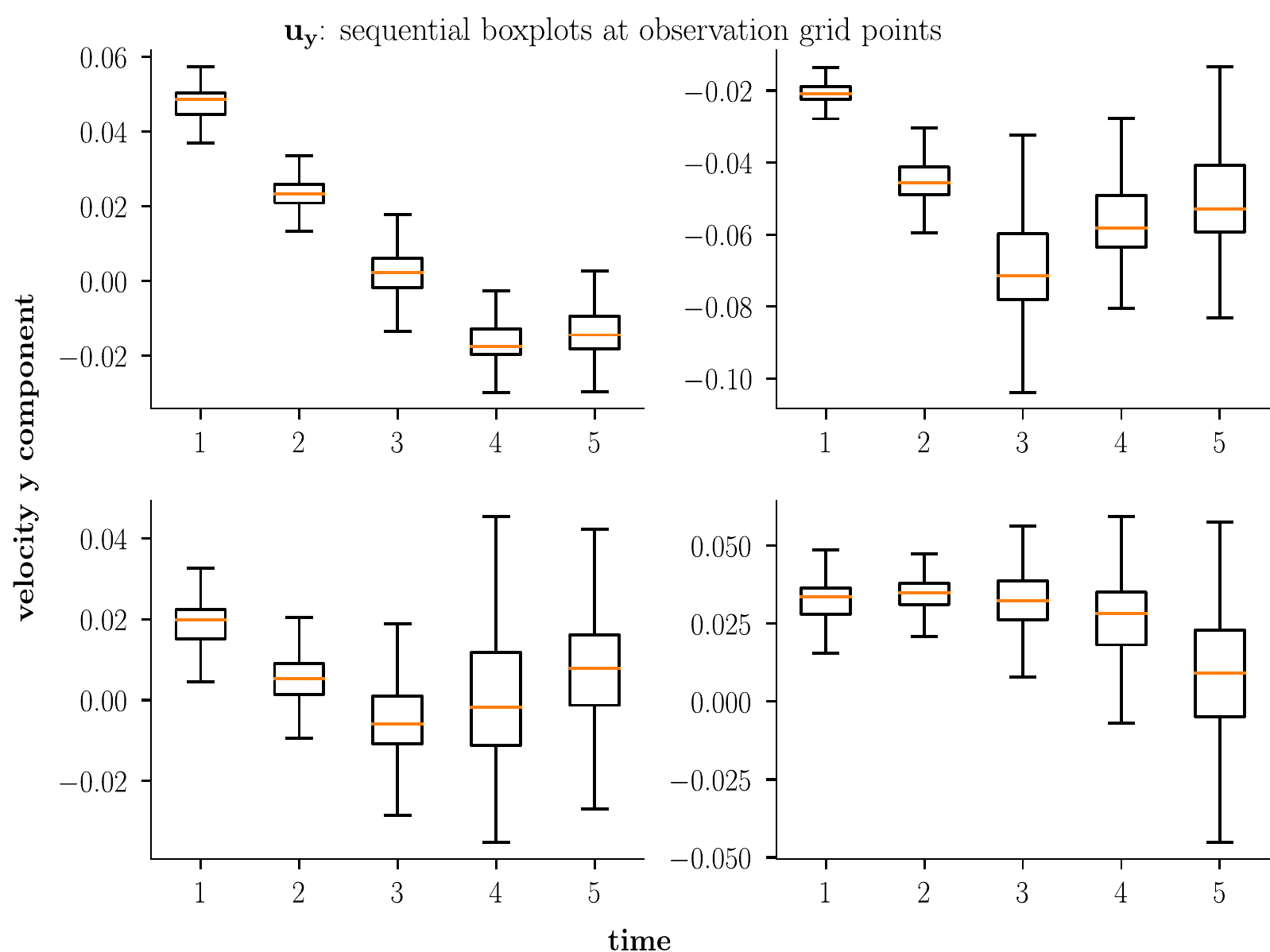}
                        \par\end{center}%
        \end{minipage}
        
        \caption{\label{fig:boxplot-u}Box plots for the SPDE ensemble velocity components,
                $\vecu_{x}$ (top row) and $\vecu_{y}$ (bottom row) at times $t=1,2,\dots,5$
                in ett, at eight observation grid points, divided into two figures of
                four plots each (each plot corresponds to one grid point). The plots show non-symmetry and fat tails in the distribution of the ensembles, again providing strong evidence to the fact that the ensembles are non-Gaussian.}
\end{figure}


\section{Conclusion and future work \label{sec:conlusion}}

In this paper, we have described the damped and forced deterministic system and the numerical methodology that we used to solve the system on a fine resolution spatial grid. We also described the stochastic version of this system, derived by using the variational approach formulated in \cite{Holm2015}, see Section \ref{sec:2dequations}, The numerical methodology we used for solving the deterministic system was then extended to solve the stochastic version and a proof for the numerical consistency was provided. In Section \ref{sec:calibration}, we have described our numerical calibration methodology for the stochastic model. Here, numerical simulations and tests were provided to show that by using our methodology, one can estimate the velocity-velocity spatial correlation structure from \emph{data}. Specifically, we showed that an ensemble of flow paths described by the stochastic system accurately tracks the large-scale behaviour of the underlying deterministic system for a physically adequate period of time. This was verified for all three fields of interest; namely,  stream function, velocity and vorticity. 

The stochastic model calibrated in this manner was used to quantify the uncertainty of the deterministic models at the coarse resolution. As expected, the uncertainty decreases as the grid becomes more refined, and the fidelity to  the true solution increases as the size of the ensemble increases. These tests prove the feasibility of the choice of stochastic velocity decomposition that was introduced as a constraint in the variational principle for fluid dynamics in \cite{Holm2015}  and was derived using multi-time homogenisation in \cite{CoGaHo2017}. 

The current work is the first stage in the development of a new ensemble-based data assimilation methodology using particle filters. The successful confirmation of the new methodology hinges crucially on maintaining a balance which ensures that the cloud of particles encompasses the true solution, while also producing sufficient spread of the ensemble realisations. To-date these two criteria have been achieved only via ad-hoc methods. Being based on the principles of stochastic geometric mechanics, the current work offers the first theoretically validated systematic approach that satisfies both of these crucial ensemble-based data assimilation criteria.  The trajectories of the particles in the ensemble arise from a decomposition of the true (deterministic) fluid velocity into a drift velocity and a (Stratonovich) stochastic perturbation at the coarse resolution. The equations for the drift velocity arise from stationary variations of Hamilton's principle for ideal fluid dynamics, constrained to follow the stochastic Lagrangian paths whose spatial correlations are determined from the present methodology. Thus, derived via fundamental principles of stochastic geometric mechanics, the present methodology has produced successful results for uncertainty quantification.    

The ensembles produced by the methodology presented here will be used in further work to forecast the future position of the true trajectory. An additional mechanism will correct the ensemble by selecting and multiplying the more likely particles and casting out those that are too far from the true solution. This ``pruning'' procedure to correct, or refine the ensemble will be based on partial (sparse/noisy) observations.  The correction mechanism assimilates the data into the system and reduces the model uncertainty.  Accomplishing this forecast will be a challenging task, because the dimensionality of the system will remain high, even after applying the coarsening using the methodology developed here. The authors will report progress toward accomplishing such forecasts in a sequel to the present paper.

\subsection*{Acknowledgements}
The authors thank The Engineering and Physical Sciences Research Council (EPSRC) for their support of this work through the grant EP/N023781/1. The authors also thank Pavel Berloff, Mike Cullen, John Gibbon, Georg Gottwald, Nikolas Kantas, Etienne Memin, Sebastian Reich, Valentin Resseguier, and Aretha Teckentrup for the many useful, constructive discussions held with them throughout the preparation of this work.

\bibliographystyle{apalike}
\bibliography{refs}

\end{document}